\documentclass[11pt]{article}
\usepackage[T1]{fontenc}
\usepackage{geometry} 
\usepackage{graphicx}
\usepackage{amsmath,amssymb,amsthm,amsfonts,dsfont,mathtools}
\usepackage{enumitem}
\usepackage{relsize}
\usepackage[margin=1cm]{caption}
\usepackage{wrapfig}
\usepackage{frame,color}
\usepackage{environ,subfig}
\usepackage{xspace}
\usepackage{framed}
\usepackage{comment}
\usepackage{changepage}
\usepackage[linesnumbered,boxed]{algorithm2e}
\usepackage{algorithmic}
\usepackage{mathrsfs}
\usepackage{tabu} % tabular used in math mode
\usepackage{authblk} 
\usepackage{setspace}

\usepackage{thmtools}
\usepackage{thm-restate}

\definecolor{ForestGreen}{rgb}{0.1333,0.5451,0.1333}
\definecolor{DarkRed}{rgb}{0.8,0,0}
\definecolor{Red}{rgb}{1,0,0}
\usepackage[linktocpage=true,
pagebackref=true,colorlinks,
linkcolor=DarkRed,citecolor=ForestGreen,
bookmarks,bookmarksopen,bookmarksnumbered]
{hyperref}
\usepackage[capitalize]{cleveref}

\crefname{theorem}{Theorem}{Theorems}
\Crefname{lemma}{Lemma}{Lemmas}
\Crefname{figure}{Figure}{Figures}

\newcommand{\Ts}{T^{\mathsf{st}}}

\newcommand{\Lf}{L_{\star}}

\newcommand{\Ls}{L_{\mathsf{source}}}

\newcommand{\Lt}{L_{\mathsf{sink}}}

\newcommand{\ecutemb}[3]{(#1, \, #2, \, #3)\text{-}\mathtt{e}\text{-}\mathtt{Cut}\text{-}\mathtt{or}\text{-}\mathtt{Expander}}

\newcommand{\vcutemb}[3]{(#1, \, #2, \, #3)\text{-}\mathtt{v}\text{-}\mathtt{Cut}\text{-}\mathtt{or}\text{-}\mathtt{Expander}}

\newcommand{\betaKKOV}{\beta_{\operatorname{KKOV}}}

\newcommand{\detbalspcut}[4]{(#1, \, #2, \, #3, \, #4)\text{-}\mathtt{Det}\text{-}\mathtt{Sparse}\text{-}\mathtt{Cut}}

\newcommand{\Cin}{C_{\mathsf{in}}}
\newcommand{\Cout}{C_{\mathsf{out}}}

\newcommand{\Troute}{T_{\mathsf{route}}}
\newcommand{\Temb}{T_{\mathsf{emb}}}
\newcommand{\Tcut}{T_{\mathsf{cut}}}

\newcommand{\Gsp}{G^\diamond}
\newcommand{\Esp}{E^\diamond}
\newcommand{\Vsp}{V^\diamond}

\newcommand{\Er}{E^\mathsf{r}}

\newcommand{\Phiv}{\Psi}
\newcommand{\Phie}{\Phi}

\newcommand{\ephicut}{\phi_{\mathsf{cut}}}
\newcommand{\ephiemb}{\phi_{\mathsf{emb}}}
\newcommand{\phicut}{\psi_{\mathsf{cut}}}
\newcommand{\phiemb}{\psi_{\mathsf{emb}}}
\newcommand{\betaleft}{\beta_{\mathsf{leftover}}}

\newcommand{\betacut}{\beta_{\mathsf{cut}}}

\newcommand{\Pactive}{\mathcal{P}_\mathsf{active}}
\newcommand{\aactive}{\textsf{active}}
\newcommand{\dead}{\textsf{dead}}
\newcommand{\success}{\textsf{successful}}

\newcommand{\cutmatch}[2]{(#1, \, #2)\text{-}\mathtt{Multi}\text{-}\mathtt{Embed}}

\newcommand{\Bernoulli}{\operatorname{Bernoulli}}

\newcommand{\PhiIn}[2]{\Phi^{\mathsf{in}}_{#1, \; #2}}

\newcommand{\vPhiIn}[2]{\Psi^{\mathsf{in}}_{#1, \; #2}}
\newcommand{\vPhiOut}[2]{\Psi^{\mathsf{out}}_{#1, \; #2}}

\newcommand{\vol}{\operatorname{Vol}}

\newtheorem{theorem}{Theorem}[section]
\newtheorem{lemma}{Lemma}[section]

\newtheorem{definition}{Definition}[section]

\newcommand{\ignore}[1]{}

\newcommand{\Expect}{\mathbf{E}}

\newcommand{\Prob}{\mathbf{Pr}}

\newcommand{\ceil}[1]{\left\lceil #1 \right\rceil}
\newcommand{\floor}[1]{\lfloor #1 \rfloor}

\newcommand{\poly}{{\operatorname{poly}}}
\newcommand{\indeg}{\operatorname{indeg}}

\newcommand{\dist}{\operatorname{dist}}

\newcommand{\ID}{\operatorname{ID}}
\newcommand{\LOCAL}{\mathsf{LOCAL}}

\newcommand{\CONGEST}{\mathsf{CONGEST}}
\newcommand{\CLIQUE}{\mathsf{CONGESTED}\text{-}\mathsf{CLIQUE}}

\newcommand{\mix}{\ensuremath{\tau_{\operatorname{mix}}}}

\newcommand{\idle}{\textsf{idle}}

\newcommand{\VVcut}{\mathcal{V}_{\mathsf{cut}}}
\newcommand{\VVemb}{\mathcal{V}_{\mathsf{emb}}}

\newcommand{\UU}{\mathcal{U}}

\newcommand{\TT}{\mathcal{T}}

\newcommand{\VV}{\mathcal{V}}
\newcommand{\PP}{\mathcal{P}}
\newcommand{\SSS}{\mathcal{S}}
\newcommand{\MM}{\mathcal{M}}
\global\long\def\eps{\epsilon}

\def\ShowComment{True} % Switch comments on or  off            

\ifdefined\ShowComment
\def\thatchaphol#1{\marginpar{$\leftarrow$\fbox{T}}\footnote{$\Rightarrow$~{\sf\textcolor{purple}{#1 --Thatchaphol}}}}

\else
\def\thatchaphol#1{}

\fi 

\def\thatchaphol#1{\marginpar{$\leftarrow$\fbox{ts}}\footnote{$\Rightarrow$~{\sffamily #1 --TS}}}

\title{Deterministic Distributed Expander Decomposition and Routing with Applications in Distributed Derandomization}

\author[1]{Yi-Jun Chang}
\author[2]{Thatchaphol Saranurak}

\affil[1]{ETH Institute for Theoretical Studies, Switzerland}
\affil[2]{Toyota Technological Institute at Chicago, USA}

\begin{document}
\date{}
\maketitle
\pagenumbering{gobble}   
\begin{abstract}
There is a recent exciting line of work in distributed graph algorithms in the $\CONGEST$ model that exploit expanders. 
\emph{All} these  algorithms so far are based on two tools: \emph{expander decomposition} and \emph{expander routing}. An $(\epsilon,\phi)$-expander
decomposition  removes $\epsilon$-fraction of the edges so that
the remaining connected components have conductance at least $\phi$,
i.e.,~they are $\phi$-expanders, and expander routing
 allows each vertex $v$ in a $\phi$-expander to very quickly exchange $\deg(v)$
messages with \emph{any} other vertices, not just its local neighbors.

In this paper, we give the first efficient \emph{deterministic} distributed algorithms for both tools. 
We show that an $(\epsilon,\phi)$-expander decomposition
can be deterministically computed in $\poly(\eps^{-1})  n^{o(1)}$ rounds
for $\phi = \poly(\eps) n^{-o(1)}$, and that expander routing
can be performed deterministically in $\poly(\phi^{-1})n^{o(1)}$ rounds.
Both results match previous bounds of randomized algorithms by [Chang and Saranurak, PODC 2019] and [Ghaffari, Kuhn, and Su, PODC 2017] up to subpolynomial factors.

Consequently, we 
\emph{derandomize} 
existing 
  distributed algorithms 
  that exploit expanders.
We show that a minimum spanning tree
on $n^{o(1)}$-expanders can be constructed deterministically in
$n^{o(1)}$ rounds, and triangle detection and enumeration
on \emph{general} graphs can be solved deterministically in $O(n^{0.58})$ and $n^{2/3 + o(1)}$ rounds,
respectively.

Using similar techniques, we also
give the first \emph{polylogarithmic-round} randomized
algorithm for constructing an $(\epsilon,\phi)$-expander decomposition
in
$\poly(\eps^{-1}, \log n)$ rounds for $\phi =  1 / \poly(\eps^{-1}, \log n)$.
This algorithm is faster than the previous algorithm by
[Chang and Saranurak, PODC 2019] in all regimes of parameters. 
The previous algorithm needs
$n^{\Omega(1)}$ rounds for any $\phi\ge 1/\poly\log n$. 
\end{abstract}

\newpage
\doublespacing
\tableofcontents
\singlespacing
\newpage

\pagenumbering{arabic}   
\section{Introduction}

In this paper, we consider the $\CONGEST$ model of distributed computing, where the underlying distributed network
is represented as an graph $G=(V,E)$, where each vertex corresponds to a computer, and each edge corresponds to 
a  communication link.
%It is common in the literature to assume that each $v\in V$ initially knows some global parameters such as  $n=|V|$, $\Delta = \max_{v \in V} \deg(v)$.
%, and $D = \text{diameter}(G)$. 
%In this paper, we only require each vertex to know  $n=|V|$.
Each vertex $v \in V$ has a distinct $\Theta(\log n)$-bit identifier $\ID(v)$, where $n = |V|$ is the number of vertices in the graph.
The computation proceeds in synchronized \emph{rounds}.
In each round, each vertex $v \in V$ can perform unlimited 
local computation, and may send a distinct
$O(\log n)$-bit message to each of its neighbors.
In the randomized variant of $\CONGEST$,
each vertex can generate unlimited local random bits, 
but there is no global randomness.
%The $\LOCAL$ model~\cite{Linial92} is simply the $\CONGEST$ model without message size constraint.
A related model called the $\CLIQUE$ model is a variant of the $\CONGEST$ model where each vertex $v \in V$ is able to send a separate $O(\log n)$-bit message to each vertex in $V \setminus \{v\}$. 

\paragraph{Expander routing.}
%\thatchaphol{TODO: refer to ``expander routing'' to be consistent to the abstract.}
Ghaffari, Kuhn, and Su~\cite{GhaffariKS17} considered a routing problem on high-conductance graphs. They proved that if each vertex $v \in V$ is the source and the destination of at most $\deg(v)$ messages, then all messages can be routed to their destinations in $\mix(G) \cdot 2^{O(\sqrt{\log n \log \log n})}$ rounds with high probability, where  $\mix(G)$ is the mixing time of the lazy random walk on $G$, and we have the following relation~\cite{JerrumS89} between the 
mixing time $\mix(G)$ and conductance $\Phie(G)$: 
\[
\Theta\left(\frac{1}{\Phie(G)}\right) \leq \mix(G) \leq \Theta\left(\frac{\log n}{\Phie(G)^2}\right).
\]
The $2^{O(\sqrt{\log n \log \log n})}$ factor was later improved by Ghaffari and Li~\cite{GhaffariL2018} to $2^{O(\sqrt{\log n })}$. 

Expander routing is a very useful tool in designing distributed algorithms on high-conductance graphs. In particular, it was shown in~\cite{GhaffariKS17}  that a minimum spanning tree can be constructed in $\poly\left( \phi^{-1} \right) \cdot  2^{O(\sqrt{\log n})}$ rounds on graphs $G$ with $\Phie(G) = \phi$, bypassing the ${\Omega}(\sqrt{n} / \log n)$ lower bound for general graphs~\cite{PelegR2000,DasSarma2012}.

More generally, as the expander routing algorithms of~\cite{GhaffariKS17,GhaffariL2018} allows the vertices to communicate arbitrarily, only subjecting to some bandwidth constraints, these routing algorithms allow us to simulate known algorithms from  other models of parallel and distributed graph algorithms with small overhead. Indeed, it was shown in~\cite{GhaffariL2018} that many work-efficient PRAM algorithms can be transformed into round-efficient distributed algorithms in the $\CONGEST$ model when the underlying graph $G$ has high conductance.

\paragraph{Expander decompositions.} A major limitation of the approach of~\cite{GhaffariKS17,GhaffariL2018} is that it is only applicable to graphs with high conductance. A natural idea to extend this line of research to general graphs is to consider \emph{expander decompositions}.  
Formally, an $(\epsilon,\phi)$-expander decomposition of a graph $G=(V,E)$
is a partition of the edges $E=E_{1}\cup E_2 \cup \cdots\cup E_{x} \cup \Er$ meeting the following conditions.
\begin{itemize}
    \item  The subgraphs $G[E_1], G[E_2], \ldots, G[E_x]$ induced by the clusters are vertex-disjoint.
    \item  $\Phie(G[E_x]) \geq \phi$, for each $1 \leq i \leq x$.
    \item  The number of remaining edges is at most $\epsilon$ fraction of the total number of edges, i.e., $|\Er| \leq \epsilon|E|$.
\end{itemize}
In other words, an $(\epsilon,\phi)$-expander decomposition of a graph $G=(V,E)$
is a removal of at most $\epsilon$ fraction of its edges in such a way that each  remaining connected component has conductance at least $\phi$.
%Note that this   expander decomposition is a little non-standard in that it is \emph{edge-based} in the sense that each high-conductance component is an edge-induced subgraph. A variant of expander decomposition that appears more commonly in literature is the \emph{vertex-based} one where it is required that each high-conductance component is a vertex-induced subgraph. 
%However, in many of the applications of expander decomposition, this distinction does not matter.
It is
 well known  that for {\em any} graph, and an $(\epsilon,\phi)$-expander decomposition  exists for any $0 < \epsilon < 1$ and $\phi = \Omega(\epsilon/\log n)$~\cite{GoldreichR1999,KannanVV04,spielman2004nearly}, and this bound is \emph{tight}. After removing any constant fraction of the edges in a hypercube, some remaining
 component must have conductance at most $O(1/\log n)$~\cite{AlevALG18}.
 %Moshkovitz and Shapira~\cite{MoshkovitzS18} showed that this bound is essentially tight.
 %In particular, removing any constant fraction $\epsilon$ of the edges, the remaining
 %components have conductance at most $O((\log\log n)^2/\log n)$.
  %The expander decomposition has a wide range of applications, and it has been applied to solving linear systems~\cite{SpielmanT14}, unique games~\cite{AroraBS2015,Trevisan08,RaghavendraS2010},
 %minimum cut~\cite{KawarabayashiT15}, and dynamic algorithms~\cite{NanongkaiSW17}. 
 Expander decompositions have a wide range of applications, and it has been applied to solving linear systems~\cite{spielman2004nearly},
% ~\cite{SpielmanT14}, 
 unique games~\cite{AroraBS2015,Trevisan08,RaghavendraS2010},
 minimum cut~\cite{KawarabayashiT15}, and dynamic algorithms~\cite{NanongkaiSW17}.

\paragraph{Distributed expander decompositions.} 
 Recently, Chang, Pettie, and Zhang~\cite{ChangPZ19} applied expander decompositions to the field of distributed computing, and they showed that a variant of expander decomposition can be computed efficiently in $\CONGEST$. 
 Using this decomposition, they showed that triangle detection and enumeration can be solved in
$\tilde{O}(n^{1/2})$ rounds.\footnote{The $\tilde{O}(\cdot)$ notation hides any factor polylogarithmic in $n$.}
The previous upper
bounds for  triangle detection and enumeration  were 
$\tilde{O}(n^{2/3})$ and $\tilde{O}(n^{3/4})$, respectively,
due to Izumi and Le~Gall~\cite{IzumiL17}.
Later, Chang and Saranurak~\cite{changS19} improved the expander decomposition algorithm of~\cite{ChangPZ19}. For any $0 < \epsilon < 1$, and for any positive integer $k$,
an $(\epsilon,\phi)$-expander decomposition of a graph $G=(V,E)$ with $\phi=(\epsilon/\log n)^{2^{O(k)}}$  can be constructed in $O\left(n^{2/k}\cdot\poly\left(\phi^{-1},\log n\right)\right)$ rounds with high probability.
As a consequence, triangle detection and enumeration can be solved in
$\tilde{O}(n^{1/3})$ rounds, matching the $\tilde{\Omega}(n^{1/3})$ lower bound of Izumi and Le~Gall~\cite{IzumiL17} and Pandurangan, Robinson, and Scquizzato~\cite{PanduranganRS18}.

The triangle finding algorithms of~\cite{ChangPZ19,changS19} are based on the following generic framework.
Construct an expander decomposition, and then the  routing algorithms of~\cite{GhaffariKS17,GhaffariL2018} enables us to simulate 
some known $\CLIQUE$ algorithms with small overhead on the  high-conductance subgraphs $G[E_1], G[E_2], \ldots, G[E_x]$, and finally the remaining edges $\Er$ can be handled using recursive calls.  

After~\cite{ChangPZ19,changS19}, 
 other applications of distributed expander 
decompositions have been found.
Daga et~al.~\cite{Daga2019distributed} applied  distributed expander decompositions to obtain the first sublinear-round algorithm for \emph{exact} edge connectivity.

Eden et~al.~\cite{EFFKO19} showed that distributed expander decompositions are also useful for various distributed subgraph finding problems beyond triangles. 
%Using expander decompositions, they proved the following results. 
For \emph{any} $k$-vertex subgraph $H$,  whether a copy of $H$ exists can be detected in $n^{2 - \Omega(1/k)}$ rounds, %nearly
matching the $n^{2 - O(1/k)}$ lower bound of Fischer et~al.~\cite{FischerGO18}.
%\thatchaphol{I removed the sentence about the lower bound for cycle detection. If I remember correctly, cycle detection did not use expander decomposition.}
There is a constant $\delta \in (0, 1/2)$ such that for any $k$, any $\Omega(n^{(1/2)+\delta})$ lower bound on $2k$-cycle detection would imply a new circuit lower bound.

 Eden et~al.~\cite{EFFKO19} also showed that all 4-cliques can be enumerated in $\tilde{O}(n^{5/6})$ rounds, and all 5-cliques in $\tilde{O}(n^{21/22})$ rounds.  Censor-Hillel, Le~Gall, and Leitersdorf~\cite{CensorLL20}  later improved this result, showing that all $k$-cliques can be enumerated  in $\tilde{O}(n^{2/3})$ rounds for $k = 4$ and $\tilde{O}(n^{1-2/(k+2)})$ rounds for $k \geq 5$, getting closer to the $\tilde{\Omega}(n^{1 - 2/k})$ lower bound of Fischer et al.~\cite{FischerGO18}.

\subsection{Our Contribution}
The main contribution of this paper is to offer the first efficient \emph{deterministic} distributed algorithms for both  expander decomposition and routing in the $\CONGEST$ model.

\begin{restatable}[Deterministic expander decomposition]{theorem}{Rdet}\label{thm-det-expander-decomp} Let $0 < \epsilon < 1$ be a parameter. An $(\epsilon,\phi)$-expander decomposition of a graph $G=(V,E)$ with  $\phi =  \poly(\epsilon) 2^{-O(\sqrt{\log n \log \log n})}$ can be computed in $\poly(\epsilon^{-1}) 2^{O(\sqrt{\log n \log \log n})}$ rounds deterministically.
\end{restatable}

More generally, there is a tradeoff of parameters in \cref{thm-det-expander-decomp}. For any $1 > \gamma \geq \sqrt{\log \log n / \log n}$, there is a deterministic expander decomposition algorithm with round complexity $\epsilon^{-O(1)}  \cdot n^{O(\gamma)}$ with parameter $\phi = \epsilon^{O(1)}  \log^{-O(1/\gamma)}n$.
For example, an  $(\epsilon,\phi)$-expander decomposition of a graph $G=(V,E)$ with  $\phi =  \poly(\epsilon /\log n)$ can be computed deterministically in $\poly\left(\epsilon^{-1}\right)n^{0.001}$ rounds. 
%\thatchaphol{Does this mean that $\eps$ must be much lass than one in the above theorem?}

\begin{restatable}[Deterministic routing on expanders]{theorem}{Rrouting}\label{thm-routing-general}
Let $G=(V,E)$ be a graph with $\Phie(G) = \phi$, where vertex $v \in V$ is a source and a destination of $O(L) \cdot \deg(v)$ messages.
Then there is a deterministic algorithm that routes all messages to their destination in $O(L) \cdot \poly\left(\phi^{-1}\right) \cdot  2^{O\left( \log^{2/3} n \log^{1/3} \log n \right)}$ rounds. 
\end{restatable}

These results open up the possibility of  \emph{derandomization} of randomized distributed algorithms that are based on these techniques~\cite{GhaffariKS17,GhaffariL2018,ChangPZ19,changS19,Daga2019distributed,EFFKO19,izumiLM20,ChatterjeePN20,SuVdsic19,CensorLL20}.
%some randomized distributed algorithms that are based on these tools. We obtain a $n^{o(1)}$-round deterministic algorithm for minimum spanning tree on graphs with conductance $\phi = n^{-o(1)}$, and an
% $n^{1 - \omega^{-1} + o(1)} < 
%$O(n^{0.58})$-round deterministic algorithm for triangle detection on any graphs.
%\thatchaphol{Move the formal definition of triangle problems to \Cref{sect-derand}.}

We show that  triangle detection and counting can be solved deterministically in $n^{1 - \frac{1}{\omega} + o(1)} < O(n^{0.58})$ rounds, and triangle enumeration can be solved deterministically  in $n^{\frac{2}{3} +o(1)}$ rounds, by derandomizing the algorithm of~\cite{ChangPZ19} using \cref{thm-det-expander-decomp,thm-routing-general}. 
See \Cref{sect-derand} for a formal definition of these problems.
To the best of our knowledge, before this work, the only known deterministic algorithm for triangle detection is the trivial $O(n)$-round algorithm that asks each vertex $v$ to send $\{ \ID(u) \ | \ u \in N(v) \}$ to all its neighbors. 

\begin{restatable}[Triangle finding]{theorem}{Rtriangle}\label{thm-triangle}
Triangle detection and counting can be solved deterministically in $n^{1 - \frac{1}{\omega} + o(1)} < O(n^{0.58})$ rounds, and triangle enumeration can be solved deterministically  in $n^{\frac{2}{3} +o(1)}$ rounds.
\end{restatable}

We show that a minimum spanning tree can be constructed in $\poly\left(\phi^{-1}\right) \cdot n^{o(1)}$ rounds deterministically on graphs $G$ with conductance $\Phie(G) = \phi$, by derandomizing the algorithm of~\cite{GhaffariKS17} using \cref{thm-routing-general}. To the best of our knowledge, before this work, there is no known deterministic algorithm that can take advantage of the fact that the underlying graph has high conductance. Even for the case of $\Phie(G) = \Omega(1)$, the state-of-the-art deterministic algorithm is the well-known one that costs $O(D + \sqrt{n}\log^\ast n)$ rounds, where $D$ is the diameter of the graph~\cite{KUTTEN199840}.

\begin{restatable}[Minimum spanning trees]{theorem}{Rmst}\label{thm-mst}
A minimum spanning tree of a graph $G$ with $\Phie(G) = \phi$ can be constructed deterministically in $\poly\left( \phi^{-1}\right) \cdot n^{o(1)}$ rounds.
\end{restatable}

The techniques used in our deterministic expander decomposition algorithms also enables us to obtain an improved expander decomposition algorithm in the randomized setting. More specifically,   we can afford to have %more careful calculation shows that we can have 
$\phi = \Omega\left(\epsilon^{3} \log^{-10} n \right)$ in \cref{thm-rand-expander-decomp}.

\begin{restatable}[Randomized expander decomposition]{theorem}{Rrand}\label{thm-rand-expander-decomp} Let $0 < \epsilon < 1$ be a parameter. An $(\epsilon,\phi)$-expander decomposition of a graph $G=(V,E)$ with  $\phi = 1/ \poly(\epsilon^{-1}, \log n)$ can be computed in $\poly(\epsilon^{-1}, \log n)$ rounds with high probability.
\end{restatable} 

 This is the \emph{first} polylogarithmic-round distributed algorithm for expander decomposition with $\phi = 1/\poly \log n$. % For comparison, 
 The previous algorithm~\cite{changS19} needs $n^{\Omega(1)}$ rounds if $\phi = 1 / \poly\log n$.
In fact, %the parameters in 
\Cref{thm-rand-expander-decomp} attains better round complexity than the algorithm of~\cite{changS19} in all regimes of parameters. 
Specifically, for any given positive integer $k$, the  algorithm of~\cite{changS19} computes an expander decomposition with $\phi = 1/\poly(\epsilon^{-1},\log n)^{2^{O(k)}}$ % \ll 1/\poly(\epsilon^{-1},\log n)$
in $\poly(\epsilon^{-1},\log n)^{2^{O(k)}}\cdot n^{2/k}$ % \gg \poly(\epsilon^{-1},\log n)$ 
rounds, where the exponent $2^{O(k)}$ is enormous even for small constant $k$.
We note, however, that the expander decomposition of~\cite{changS19} has an additional guarantee that each cluster $G[E_i]$ in the decomposition is a \emph{vertex-induced} subgraph. 
% The output of our randomized and deterministic expander decomposition algorithms do not have this property.
%, which we do not guarantee,  that each expander is a subgraph induced by a \emph{vertex} set.%\thatchaphol{Nevertheless, we do not need this property for the applications. For example, using \Cref{thm-rand-expander-decomp} implies gives a new algorithm for triangle enumeration in $O(n^{1/3} \log^x n)$ rounds, while the previous algorithm by \cite{changS19} needs $O(n^{1/3} \log^{xxx} n)$ rounds.}

\subsection{Preliminaries}\label{sect-prelim}

We  present the graph terminologies     used in this paper.

\paragraph{Basic graph notation.} 
The parameter $n$ always denotes the number of vertices, and the parameter $\Delta$ always denotes the maximum degree. We write ``with high probability'' to denote a success probability of $1 - 1/\poly(n)$.
For any graph $H=(V',E')$, we write $H[S]$ to denote the subgraph of $H$ induced by $S$, where $S$ can be an edge set $S \subseteq E'$ or a vertex set $S \subseteq V'$.

For each vertex $v$, denote $N(v)$ as the set of neighbors of $v$. We write $\dist(u,v)$ to denote the distance between $u$ and $v$. 
%We also write $N^k(v) = \{ u \in V \ | \ \dist(u,v) \leq k\}$. Note that $N^1(v) = N(v) \cup \{v\}$. 
%The terms $\dist(u,v)$, $N(v)$, and $N^k(v)$ 
For a subset $S \subseteq V$, denote by $E(S)$ the set of all edges with both endpoints in $S$. Similarly, $E(S_1, S_2)$ is the set of all edges between $S_1$ and $S_2$. These terms depend on the underlying graph $G$, which appears subscripted if not clear from context.

%, and we write $G\{S\}$ to denote the graph resulting from adding $\deg_V(v) - \deg_S(v)$ self loops to each vertex $v$ in $G[S]$.  (As in~\cite{SpielmanS08}, self-loops are regarded as contributing one to the degree, not two, hence every $v\in S$ has the same degree
%in both $G$ and $G\{S\}$.)
%Observe that we always have
%\[\Phi(G\{S\}) \leq \Phi(G[S]).\]

\paragraph{Embeddings.}
An \emph{embedding} of a graph $H=(V',E')$ into a vertex set $S \subseteq V$ of $G=(V,E)$ with \emph{congestion} $c$ and \emph{dilation} $d$ consists of the following.
\begin{itemize}
    \item A bijective mapping $f$ between the vertices of $V'$ and $S$.
    \item A set of paths $\PP = \{ P_e \ | \ e \in E' \}$, where  path $P_e$ is a path in $G$ whose two ends are  $f(u)$ and $f(v)$.
    \item Each path $\PP$ has length at most $d$, and each edge $e \in E$ appears in at most $c$ paths in $\PP$. 
\end{itemize}

%there is  such that  
%If $V = S$, then we must have $\Phiv(G) \geq \Phiv(H)/c$~\cite{KRV09}.

\paragraph{Steiner trees.} A \emph{Steiner tree} $T$ for a vertex set $S$ is a tree whose leaf vertices are exactly $S$. In this paper we often need to deal with graphs $H=(V',E')$ that have high diameter or are disconnected, but they will be supplied with a Steiner tree $T'$ for $V'$ with small diameter so that the vertices in $V'$ are able to communicate efficiently. Throughout this paper, unless otherwise stated, $D$ denotes the diameter of the Steiner tree of the current graph, not the diameter of the current graph.

\paragraph{Conductance.} Consider a graph $G = (V,E)$. For a vertex subset $S$, we write $\vol(S)$ to denote $\sum_{v \in S} \deg(v)$. 
%By default, the degree is with respect to the original graph $G$. We write $\bar{S} = V \setminus S$, 
Let $\partial(S) = E(S, V \setminus {S})$ be the set of edges $e = \{u,v\}$ with $u \in S$ and $v \in V\setminus S$.  The  \emph{conductance} of a cut $S$ is defined as \[\Phie(S) = \frac{|\partial(S)|}{\min\{\vol(S), \vol(V \setminus S)\}}.\]
For the special case of $S = \emptyset$ and $S = V$, we set $\Phie(S) = 0$ .
The \emph{conductance} of a graph $G$ is  \[\Phie(G) = \min_{S \subseteq V  \ \text{s.t.} \  S \neq \emptyset \ \text{and} \  S \neq V} \Phie(S).\] 
In other words, $\Phie(G)$ is the minimum value of $\Phie(S)$ over all non-trivial cuts $S \subseteq V$.

%We say that $G$ is an $\phi$-expander if $\Phie(G) \geq \phi$. We say that a cut $S$ is $\phi$-sparse if $\Phie(S) \leq \phi$, and we say that $S$ is $\beta$-balanced if $\vol(S) \geq \beta \vol(V)$. Note that $\vol(V) = 2 |E|$. 

\paragraph{Sparsity.} The  \emph{sparsity} of a cut $S$ is defined as \[\Phiv(S) = \frac{|\partial(S)|}{\min\{ |S|, |V \setminus S|\}}.\]
For the special case of $S = \emptyset$ and $S = V$, we set $\Phiv(S) = 0$ .
The \emph{sparsity} of a graph $G$ is  \[\Phiv(G) = \min_{S \subseteq V  \ \text{s.t.} \  S \neq \emptyset \ \text{and} \  S \neq V} \Phiv(S).\] 
In other words, sparsity is a variant of conductance where the volume of a vertex set is measured by its cardinally. Note that sparsity is also commonly known as \emph{edge expansion}. Sparsity and conductance differs by a factor of at most $\Delta$.

%We say that $G$ is an $\psi$-expander if $\Phiv(G) \geq \psi$. We say that a cut $S$ is $\psi$-sparse if $\Phiv(S) \leq \psi$, and we say that $S$ is $\beta$-balanced if $|S| \geq \beta |V|$.  

\paragraph{Inner and outer sparsity.}
Consider a partition $\VV = \{V_1, V_2, \ldots, V_x\}$ of $V$ for a graph $G=(V,E)$.
We say that a cut $S$ \emph{respects} $\VV$ if for each $V_i \in \VV$, either $V_i \subseteq S$ or $V_i \subseteq V \setminus S$. The \emph{outer sparsity} $\vPhiOut{G}{\VV}$ of  $G$ is the minimum of $\Phiv_G(S)$ over all cuts $S$ that respects $\VV$.  The \emph{inner sparsity} $\vPhiIn{G}{\VV}$ of  $G$ is the minimum of $\Phiv(G[V_i])$ over all $V_i \in \VV$. 
%The inner sparsity   $\vPhiIn{G}{\VV}$ and outer sparsity  $\vPhiOut{G}{\VV}$ are defined analogously.

\paragraph{Expander split graphs.}
The \emph{expander split} $\Gsp=(\Vsp, \Esp)$ of $G=(V,E)$ is constructed as follows.
\begin{itemize}
    \item For each $v \in V$, construct a $\deg(v)$-vertex expander graph $X_v$ with $\Delta(X_v) = \Theta(1)$ and $\Phi(X_v) = \Theta(1)$.
    \item For each vertex $v \in V$,  consider an arbitrary ranking of its incident edges. Denote $r_v(e)$ the rank of $e$ at $v$.
    For each edge $e=\{u,v\} \in E$, add an edge linking the $r_u(e)$th vertex of $X_u$ and the $r_v(e)$th vertex of $X_v$.
\end{itemize}

The concept of expander split graphs and the inner and outer sparsity are from~\cite{chuzhoy2019deterministic}.
Note that in the distributed setting, $\Gsp$ can be simulated in $G$ with no added cost. See \cref{sect-conductance} for properties of expander split graphs. In particular, $\Phiv(\Gsp)$ and $\Phie(G)$ are within a constant factor of each other.

\subsection{Technical Overview}
Throughout the paper, we say that a cut $C \subseteq V$ is \emph{balanced} to informally indicate that $\min\{|V \setminus C|, |C|\}$ or $\min\{ \vol(V \setminus C), \vol(C)\}$ is high, and we say that a graph $G$ is \emph{well-connected} or is an \emph{expander} to  informally indicate that $\Phiv(G)$ or $\Phie(G)$ is high.

The main ingredients underlying our distributed expander decomposition algorithms are efficient randomized and deterministic distributed algorithms solving the following task $\mathcal{P}$. See \cref{lem-rand-cutmatch-main-bounded-deg,lem-det-cutmatch-main} for the precise specifications of the task $\mathcal{P}$ that we use. 
\begin{description}
\item[Input:] A bounded-degree graph $G=(V,E)$ and two parameters $0 < \phiemb < \phicut < 1$.
\item[Output:] Two subsets $W \subseteq V$ and $C \subseteq V$ satisfying the following conditions.
\begin{description}
    \item[Expander:] The induced subgraph $G[W]$ has  $\Phiv(G[W]) \geq \phiemb$.
    \item[Cut:] The cut $C$ satisfies  $0 \leq |C| \leq |V|/2$ and $\Phiv(C) \leq \phicut$.
    \item[Balance:] Either one of the following  is met.
    \begin{itemize}
        \item $|C| = \Omega(1) \cdot |V|$, i.e., $C$ is a \emph{balanced} cut.
        \item $|V \setminus (C \cup W)|$ is \emph{small}.  
    \end{itemize}
\end{description}
\end{description}
Note that the requirement $\Phiv(G[W]) \geq \phiemb > 0$ implies that $G[W]$ must be a connected subgraph of $G$ if $W \neq \emptyset$, and it is possible that  $G[C]$ is disconnected.

\paragraph{Randomized sequential algorithms.} We first review existing sequential algorithms solving  $\mathcal{P}$.
The  task $\mathcal{P}$ can be solved using a technique called the \emph{cut-matching game}, which was
  first introduced by Khandekar,   Rao,  and   Vazirani~\cite{KRV09}. 
The goal of a cut-matching game on $G= (V,E)$ is to either find a \emph{small-congestion embedding} of a well-connected graph $H$ into $V$, certifying that $G$ is also well-connected, or to find a sparse cut $C \subseteq V$.  The game proceeds in iterations.
Roughly speaking, in each iteration, the cut player uses some strategy to produce two disjoint subsets $S\subseteq V$ and $T \subseteq V$ with $|S| \leq |T|$, and then the matching player tries to embed a \emph{matching} between $S$ and $T$ that \emph{saturates all vertices} in $S$ with \emph{small congestion}  to the underlying graph $G$. If the matching player fails to do so because such a small-congestion matching embedding does not exist, then there must be a sparse cut $C$ in  $G$ separating the unmatched vertices in $S$ and the unmatched vertices in $T$.  If the matching player is always successful, then the strategy for the cut player guarantees that after $\poly \log n$ iterations, the union of all  matchings found by the matching player forms a graph $H$ with sparsity $\Phiv(H) = \Omega(1)$, certifying that $G$ itself is well-connected.

The above algorithm does not solve $\mathcal{P}$ yet. To facilitate discussion, let $G$ be a graph that contains a well-connected subgraph $G[W]$ with $|W| = (1-\beta)|V|$ and a sparse cut $C$ with $|C| = \Omega(\beta) \cdot |V|$, for some parameter $0 < \beta < 1$. Intuitively, the presence of the sparse cut $C$ implies the non-existence of a well-connected subgraph $G[W']$ of $G$ with  $|W| = (1- o(\beta))|V|$, and so $G$ itself is not well-connected. If we apply the above cut-matching game algorithm on $G$, then it will return us some sparse cut $C'$, as the matching player will at some point fail to embed a large enough matching. However, there is no guarantee on the balance of the cut $C'$, and $|C'|$ can be arbitrarily small.

R\"{a}cke,  Shah, and T\"{a}ubig~\cite{RST14} considered a variant of the  cut-matching game that deals with this issue. The main difference is that in the RST cut-matching game, we continue the cut-matching game on the remaining part of the graph even if a sparse cut $C_i$ is found in an iteration $i$, unless the union of all sparse cuts $C_1 \cup C_2 \cup \cdots \cup C_i$ found so far already has size $\Omega(|V|)$.  
Hence there are two possible outcomes of the RST cut-matching game. 
If the RST cut-matching game stops early, then the output is a sparse cut $C$ of size  $\Omega(|V|)$, which is good, as we can set $W =\emptyset$ to satisfy the requirements for $\mathcal{P}$. 
Otherwise, the output is a sparse cut $C$ together with a small-congestion embedding of the union of all matchings  $H = M_1 \cup M_2 \cup \cdots \cup M_\tau$ that the matching player found. 
If $C = \emptyset$, then as discussed earlier, the small-congestion embedding of $H$ certifies that $G$ itself is well-connected, then we can set $W = V$  to satisfy the requirements for $\mathcal{P}$. 

If $C \neq \emptyset$, then the embedding of $H$ does not guarantee anything about $\Phiv(G)$.
To deal with this issue,
Saranurak and Wang~\cite{SaranurakW19} showed that $G[V \setminus C]$ is actually {nearly} an expander in the following sense.
A subgraph $G[U]$ of $G$ induced by $U \subset V$ is said to be \emph{nearly} a $\phi$-expander if for all $S \subseteq U$ with $0 < \vol(S) \leq \vol(U) / 2$, we have $|\partial_{G}(S)| \geq \phi \vol(S)$. Note that if we  change the requirement from $|\partial_{G}(S)| \geq \phi \vol(S)$ to  $|\partial_{G[U]}(S)| \geq \phi \vol(S)$, then $G[U]$ would have been a $\phi$-expander.
Moreover, given such a subgraph $G[U]$, they presented an efficient sequential \emph{expander trimming} algorithm that removes a small fraction $U' \subseteq U$ of the vertices in $U$ in such a way that $\Phie(G[U \setminus U']) = \Omega(\phi)$. Applying this expander trimming algorithm to $G[V \setminus C]$, we can obtain a well-connected subgraph $G[W]$, and $C$ and $W$ together satisfy the requirement for the task $\mathcal{P}$. 

\paragraph{Deterministic sequential algorithms.} 
The above sequential algorithm for $\mathcal{P}$ is \emph{randomized} because the strategy of the cut player in the cut-matching game of~\cite{KRV09,RST14} is inherently randomized. Recently,  
 Chuzhoy et al.~\cite{chuzhoy2019deterministic} designed an efficient \emph{deterministic} sequential algorithm for  $\mathcal{P}$ by
considering a different cut-matching by   Khandekar et al.~\cite{khandekar2007cut}, where the cut player can be implemented recursively and deterministically.

In the KKOV cut-matching game~\cite{khandekar2007cut}, the strategy of the cut player in iteration $i$ is to simply find a sparse cut $S \subseteq V$ that is \emph{as balanced as possible} in the graph $H_{i-1} = M_1 \cup M_2 \cup \cdots \cup M_{i-1}$ formed by the union of all matchings found so far, and set $T = V \setminus S$.
Observe that the output cut $C$ of the task $\mathcal{P}$ is also \emph{approximately} as balanced as possible, as the output well-connected subgraph $G[W]$ implies that we cannot find a cut $C'$ that is significantly more balanced than $C$ and significantly sparser than $C$ at the same time. This gives rise to a possibility of solving $\mathcal{P}$ \emph{recursively}.
To realize this idea, the approach taken in~\cite{chuzhoy2019deterministic} is to decompose the current vertex set $V$ into $k$ subsets $\VV = \{V_1, V_2, \ldots, V_k\}$ of equal size, and then run the KKOV cut-matching game simultaneously on each part of $\VV$, and so the cut player can be implemented by solving $\mathcal{P}$ recursively on instances of size $O(n/k)$.

Similar to the case of~\cite{KRV09,RST14}, during the process, it is possible that in some iteration, the matching for some part $V_i \in \VV$ returned by the matching player does not saturate all of $S$. The approach taken in~\cite{chuzhoy2019deterministic} is to consider some threshold $0 < \beta < 1$. If adding at most $\beta |V|$ \emph{fake edges} to the graph is enough to let the cut-matching games in all parts to continue, then we add them to $G$. Otherwise, we can obtain a sparse cut $C$ of size $\Omega(\beta) \cdot |V| \leq |C| \leq |V|/2$, which is already sufficiently balanced.
Suppose we are always able to add a small number of fake edges to let the cut-matching games in all parts continue in each iteration, then in the end we obtain a small-congestion \emph{simultaneous embedding} of an expander $H_i$ to each part $V_i \in \VV$ in the \emph{augmented graph} $G'$ which is formed by adding a small number of fake edges to the original graph $G$.
Now, select $C$ to be a sparse cut of $G$ that \emph{respects} $\VV$  and is as balanced as possible. If $C$ is already sufficiently balanced, then we can return $C$. Otherwise, $|C|$ is small, and so we can add a small amount of fake edges to $G'$ to make it a well-connected graph $G''$. 

Now we face a situation that is similar to the case of~\cite{KRV09,RST14} discussed earlier. That is, we have obtained something that is almost an expander, and we just need to turn it into an expander. 

Specifically, we are given a graph $G$ and a small set of fake edges $E^\ast$ such that the graph resulting from adding all these fake edges to $G$ is well-connected. Saranurak and Wang~\cite{SaranurakW19} showed that in such a situation, there is an efficient deterministic   \emph{expander pruning} algorithm that is able to find a well-connected subgraph $G[W]$ of $G$ such that $\vol(V \setminus W)$ is small.

\paragraph{Distributed algorithms.}
In order to apply these approaches to the distributed setting, we need to deal with the following challenges.
\begin{itemize}
    \item We need an efficient distributed algorithm to implement the matching player in cut-matching games.
    \item There is no known efficient distributed algorithms for expander trimming and  expander pruning. Moreover, we cannot afford to add fake edges to the distributed network $G$.
\end{itemize}

To implement the matching player efficiently, it suffices to be able to solve the following problem efficiently. Given two disjoint sets $S \subseteq V$ and $T \subseteq V$, find a \emph{maximal} set of vertex-disjoint $S$-$T$ paths of length at most $d$. In the randomized setting, this problem can be solved in $\poly(d, \log n)$ rounds with high probability using the augmenting path finding algorithm in the distributed approximate matching algorithm of Lotker, Patt-Shamir, and Pettie~\cite{LotkerPP08}.

In the deterministic setting, we are not aware of an algorithm that can solve this problem in $\poly(d, \log n)$ rounds.\footnote{Cohen~\cite{Cohen95} showed that the $(1-\epsilon)$-approximate maximum flow problem on a directed acyclic graph with depth $r$ can be solved in PRAM deterministically in $\poly(r, \epsilon^{-1}, \log n)$ time with $O(|E|/r)$ processors, and this algorithm can be adapted to the $\CONGEST$ model with the same round complexity $\poly(r, \epsilon^{-1}, \log n)$  if \emph{fractional flow} is allowed. The part of Cohen's algorithm that rounds a fractional flow into an integral flow does not seem to have an efficient implementation in $\CONGEST$, and hence we are unable to use this algorithm.} However, if we allow a small number of \emph{leftover} vertices, then we can solve the problem efficiently using the approach of Goldberg, Plotkin, and Vaidya~\cite{GPV93}. Specifically, given a parameter $0 < \beta < 1$, in  $\poly\left(d, \beta^{-1}, \log n\right)$ rounds we can find a set $B$ of size $|B|  \leq \beta |V \setminus T|$ and a set of vertex-disjoint $S$-$T$ paths $\{P_1, P_2, \ldots, P_k\}$ of length at most $d$ in such a way that any $S$-$T$ path must vertex-intersect either $B$ or some path in $\{P_1, P_2, \ldots, P_k\}$. We  will revise the analysis of the KKOV cut-matching game to show that it still works well if the matching returns by the matching player only saturates a constant fraction of $S$, and so we are allowed to use flow algorithms that have leftover vertices.

The more crucial challenge is the lack of efficient distributed algorithms for expander trimming and  expander pruning. To get around this issue, we will consider a different approach of extracting a well-connected subgraph from the embedding of matchings in the cut-matching game. In the randomized setting that is based on the cut-matching game of~\cite{KRV09,RST14,SaranurakW19}, we will show that we can identify a large subset $U \subseteq V \setminus C$ of \emph{good vertices} so that if we start from $u \in U$ a certain random walk  defined the by the matchings $M_1, M_2, \ldots, M_{\tau}$, then the probability that the walk ends up in $v \in U$ is $\Theta(1/n)$, for all $u,v\in U$. Hence we can use random walks to embed a small-degree random graph to $U$. If we take $W$ to be the set of all vertices involved in the embedding, then $G[W]$ is well-connected because the embedding has small congestion and small dilation. Here we need the embedding to have \emph{small dilation} in order to bound the sparsity of $G[W]$. This is different from the previous works~\cite{KRV09,RST14,SaranurakW19} where we only need the embedding to have small congestion.
%\thatchaphol{Add a sentence that we don't need small dilation here for guaranteeing the conductance. Small dilation is important for us later in the routing algorithm} 
Note that $W$ not only includes   vertices in $V \setminus C$, but it might include vertices in $C$ as well. In contrast, the subset $W$ obtained using expander trimming~\cite{SaranurakW19} is guaranteed to be within $V \setminus C$.

 In the deterministic setting that is based on the approach  of~\cite{chuzhoy2019deterministic,khandekar2007cut}, we cannot afford to add fake edges in a distributed network. 
 %\thatchaphol{We really should say that the high-level idea is to combine the ``non-stop'' cut-matching game of \cite{RST14} with the recursive implementation of \cite{khandekar2007cut} shown in \cite{chuzhoy2019deterministic}. This is why we explain both KKOV and RST in the beginning.
 %Maybe it is worth repeating two times here, and somewhere in the beginning.}
 To get around this issue, we will carry out the  simultaneous execution of KKOV cut-matching games using the ``non-stop'' style of~\cite{RST14} in the sense that we still continue the cut-matching games on the remaining parts of $\VV$ even if the cut-matching games in some parts of $\VV$ have already failed. 
 Specifically, during the simultaneous execution of cut-matching games, %for each part in $\VV$, 
 we maintain a cut $C$ in such a way that if the cut-matching game in some part $V_i$ fails, then $C$ includes a constant fraction of vertices in $V_i$. Therefore, in the end, either constant fraction of the cut-matching games are successful, or $\Omega(1) \cdot |V| \leq |C| \leq |V|/2$. i.e., $C$ is a sufficiently balanced sparse cut.
 For each part $V_i$ that is successful, it is guaranteed to have a subset $U_i \subseteq V_i$ that is embedded a well-connected graph $H_i$ and $|U_i| \geq (2/3) |V_i|$. Moreover, the overall simultaneous embedding has small congestion and dilation. Now the union of $U_i$ over all successful parts $V_i$ constitutes a constant fraction of vertices in $V$. We apply the deterministic flow algorithm based on~\cite{GPV93} described earlier to enlarge and combine these expander embeddings. During this process, it is possible that a sparse cut is found, and also because of the nature of the approach of~\cite{GPV93}, there might be a small number of leftover vertices that we cannot handle, but we can show that our algorithm always end up with a sparse cut $C$ with $|C| \leq |V|/2$ and a well-connected subgraph $G[W]$ in one of the following situations.
 \begin{itemize}
     \item $|V|/100 \leq |C|$ and $W = \emptyset$, i.e., $C$ is already a balanced cut.
     \item $V = C \cup W$.
     \item $C = \emptyset$ but $|V \setminus W|$ is small, i.e., there is only a small number of leftover vertices that we cannot handle.
 \end{itemize}
This output is already good enough for our purpose.

\paragraph{Expander routing.} 
At a very high level, our deterministic expander routing algorithm follows a similar approach of the randomized routing algorithm of Ghaffari, Kuhn, and Su~\cite{GhaffariKS17}. Consider a distributed network $G=(V,E)$ with high conductance, where each vertex $v$ is a source and a sink for at most $\deg(v)$ messages. The GKS routing algorithm works as follows. First, simulate a $2|E|$-vertex $O(\log n)$-degree random graph $G_0$ to $G$, where each vertex $v$ in $G$ is responsible for $\deg(v)$ vertices in $G_0$, and the edges in $G_0$ are constructed by performing lazy random walks. Next, partition the vertices of $G_0$ into $k$ parts $\VV = \{V_1, V_2, \ldots, V_k\}$ of equal size,  and simultaneously embed an $O(\log n)$-degree random graph $H_i$ to each part $V_i$, with small congestion and dilation, also using lazy random walks. The routing is performed recursively by first routing all the messages between parts in $\VV = \{V_1, V_2, \ldots, V_k\}$, and then recursively route the messages inside each part to their destinations.

Recall that an intermediate product of our deterministic distributed expander decomposition algorithm is a simultaneous embedding of a high-conductance graph $H_i$ into $U_i \subseteq V_i$ with $|U_i| \geq (2/3)|V_i|$ for the parts  $V_i \in \VV = \{V_1, V_2, \ldots, V_k\}$ where the cut-matching game does not return a balanced sparse cut. If the underlying graph is guaranteed to be well-connected, then every cut-matching game cannot return a balanced sparse cut. Therefore, we have a simultaneous expander embedding for all parts in $\VV$.

In order to apply the recursive approach of~\cite{GhaffariKS17}  using this simultaneous expander embedding, we  need to handle the leftover vertices in $V_i \setminus U_i$, as we can only do recursive calls on expanders, and the nature of our approach is that for each part $V_i$, we can only embed an expander on a subset $U_i \subseteq V_i$, and there are always some remaining vertices. By increasing the round complexity, it is possible to reduce the size of $V_i \setminus U_i$, but we cannot afford to make it an empty set.
To deal with these leftover vertices, for each $v \in V_i \setminus U_i$, we will find another vertex $v^\star \in U_i$ that serves as the \emph{representative} of $v$ in all subsequent recursive calls. For each leftover vertex $v$ and its representative $v^\star$, we will establish a communication link between them. In other words, we embed a matching between $V_i \setminus U_i$ and $U_i$ that saturates all vertices in $V_i \setminus U_i$, for each $1 \leq i \leq k$. 
%\thatchaphol{Can we summarize this as trying to embed a matching from $\bigcup_i (V_i \setminus U_i)$ to $\bigcup_i U_i$?}

The communication links between $v$ and $v^\star$ and also the communication links between different parts $U_i$ and $U_j$ are based on  deterministic flow algorithms using the approach of~\cite{GPV93}.
The issue that there are always some unmatched source vertices in $S$ in the approach of~\cite{GPV93} can be resolved for the case of $T = V \setminus S$ by increasing the congestion and round complexity from $\poly(\log n)$ to $2^{O(\sqrt{\log n})}$, see \cref{lem-cut-match-det-noleftover}.  Note that we are allowed to have super-polylogarithmic congestion here, as this congestion has no effect on the expander embedding, and so it has no effect on the cost of simulation for recursive calls.

%\thatchaphol{Need to find the place to cite Cohen's paper. Explain a bit why hers does not work in distributed setting. FIXED: Discussed in a footnote.}  

%We also develop a deterministic \emph{load-balancing} algorithm to balance the load at each vertex, see \cref{sect-load-bal}.\thatchaphol{Change this. We use the load balancing technique by \cite{GhoshLMMPRRTZ99}. To our best knowledge, this is the first application of load balancing to graph algorithms in $\CONGEST$.} This ensures that no vertex receives too many messages during the routing algorithm.

%\begin{itemize}
%    \item We need an efficient deterministic load-balancing algorithms to avoid some vertices to receive too many messages.
 %   \item We need an efficient deterministic algorithms to send messages between different parts in $\VV$.
 %   \item We 
%\end{itemize}

\subsection{Organization}

In \cref{sect-expander-decomposition-main}, we present  our main algorithm for distributed expander decomposition in both the randomized and the deterministic models, using a certain balanced sparse cut algorithm as a black box. In \cref{sect-rand-expander-decomp}, we present our randomized sparse cut algorithm. In \cref{sect-det-expander-decomp}, we present our deterministic sparse cut algorithm. In \cref{sect-routing}, we present our algorithm for deterministic routing in high-conductance graphs. In \cref{sect-derand}, we show two applications of our results in derandomizing distributed graph algorithms.

\paragraph{Technical lemmas.} In  \cref{sect-steiner}, we provide all the basic communication primitives of Steiner trees needed in this paper.  In \cref{sect-low-diam-decomp}, we review prior work on low-diameter decompositions.  In  \cref{sect-conductance},   we provide tools for analyzing the conductance and sparsity of graphs and cuts. In \cref{sect-flow}, we provide distributed algorithms for variants of maximal flow problems and show how  sparse cuts can be obtained if we cannot find a desired  solution for the given flow problems.  In \cref{sect-sparse-cut-tools}, we provide tools  for sparse cut computation, including a technique that allows us to avoid reusing Steiner trees in recursive calls in our deterministic sparse cut algorithm. %In \cref{sect-load-bal}, we present a deterministic load-balancing algorithm that is used in our routing algorithm to make sure that no vertex receives too many messages. 
In \cref{sect-potential}, we present the analysis of potential functions in cut-matching games.
%\thatchaphol{TO discuss: I think we should put these intro appendix only when it follows from previous works in a straight forward manner. I think \cref{sect-flow} should be in the main body?  (NOT FIXED)}

\section{Distributed Expander Decomposition}\label{sect-expander-decomposition-main}

The goal of this section is to present the our main algorithm for distributed expander decomposition, which uses the  following task $\ecutemb{\ephicut}{\ephiemb}{\beta}$ as a subroutine.

\begin{definition}[Conductance-based cut or expander]\label{def-sparse-cut-e}
Let $G=(V,E)$ be any graph.  The task \[\ecutemb{\ephicut}{\ephiemb}{\beta}\] asks for either one of the following. 
\begin{description}
    \item[Expander:] A subset  $E^\ast \subseteq E$ such that the induced subgraph $G[E^\ast]$ has  $\Phie(G[E^\ast]) \geq \ephiemb$ and $|E \setminus E^\ast| \leq \beta|E|$.
    \item[Cut:] A cut  $C \subseteq V$ satisfying $\Omega(\beta \vol(V)) \leq \vol(C) \leq \vol(V)/2$ and $\Phie(C) \leq \ephicut$.
\end{description}
\end{definition}

The task defined in \cref{def-sparse-cut-e} on $G=(V,E)$ can be reduced to the following task on the expander split graph $\Gsp=(\Vsp, \Esp)$.

\begin{definition}[Sparsity-based cut or expander]\label{def-sparse-cut-v}
Let $G=(V,E)$ be any graph.  The task \[\vcutemb{\phicut}{\phiemb}{\beta}\] asks for either one of the following. 
\begin{description}
    \item[Expander:] A subset  $W \subseteq V$ such that the induced subgraph $G[W]$ has  $\Phiv(G[W]) \geq \phiemb$ and $|V \setminus W| \leq \beta|V|$.
    \item[Cut:] A cut  $C \subseteq V$   satisfying  $\Omega(\beta|V|) \leq |C| \leq |V|/2$ and $\Phiv(C) \leq \phicut$.
\end{description}
\end{definition}

\cref{lem-rand-cutmatch-main-bounded-deg,lem-det-cutmatch-main} imply the following two theorems. Note that the round complexity of \cref{thm-rand-cut-expander} does not depend on $\beta$.

\begin{theorem}[Randomized cut or expander]\label{thm-rand-cut-expander}
The task $\vcutemb{\phicut}{\phiemb}{\beta}$ can be solved with high probability in $O(D) \cdot \poly(\phicut^{-1}, \log n)$ on bounded-degree graphs, with $\phiemb = 1 / \poly(\phicut^{-1}, \log n)$.
The task $\ecutemb{\ephicut}{\ephiemb}{\beta}$ can be solved with high probability in $O(D) \cdot \poly(\ephicut^{-1}, \log n)$ on   general graphs, with $\ephiemb = 1 / \poly(\ephicut^{-1}, \log n)$.
\end{theorem}
\begin{proof}
The result for $\vcutemb{\phicut}{\phiemb}{\beta}$ follows from \cref{lem-rand-cutmatch-main-bounded-deg} immediately.
Run the algorithm of \cref{lem-rand-cutmatch-main-bounded-deg} with parameters $\psi = \phicut$, and
let $W$ and $C$ be the output result.
If $|V \setminus W| \leq \beta|V|$, then $W$ satisfies the requirement of $\vcutemb{\phicut}{\phiemb}{\beta}$ with $\phiemb = \Omega(\phicut^{-3} \log^{-10} n)$. Otherwise, $|V \setminus W| > \beta|V|$, and so $C$ satisfies the requirement of $\vcutemb{\phicut}{\phiemb}{\beta}$, as $|V \setminus W|/8 \leq |C| \leq |V|/2$.

To extend this result to the conductance-based setting on general graphs $G=(V, E)$, we use \cref{lem-expander-split-high-conductance,lem-expander-split-cut}. Specifically, we take the split graph $\Gsp=(\Vsp, \Esp)$, and then run the algorithm $\vcutemb{\phicut}{\phiemb}{\beta}$ on $\Gsp$, and let $C'$ or $W'$ be the result. 

Suppose the output is a cut  $C' \subseteq \Vsp$ with  $\Omega(\beta|\Vsp|) \leq |C'| \leq |\Vsp|/2$ and $\Phiv_{\Gsp}(C') \leq \phicut$. Then the $O(D)$-round algorithm of \cref{lem-expander-split-cut} can turn $C'$ into a cut $C \subseteq V$ in $G$ with $\Omega(\beta \vol(V)) \leq \vol(C) \leq \vol(V)/2$ and $\Phie_{G}(C) = O(\phicut)$. 
%only lose at most a constant factor in the conductance / and balance.

Suppose the output is a subset  $W' \subseteq \Vsp$ with $\Phiv(\Gsp[W']) \geq \phiemb$ and $|\Vsp \setminus W'| \leq \beta|\Vsp|$.
Then \cref{lem-expander-split-high-conductance} shows that in zero rounds, we can turn $W'$ into a subset $E^\ast \subseteq E$ in $G$ with  $\Phie(G[E^\ast]) = \Omega(\phiemb)$ and $|E \setminus E^\ast| \leq \beta|E|$. 

%That is, we only lose at most a constant factor in the conductance, and there is no loss in balance.
To summarize, there exist two universal constants $K_1$ and $K_2$ such that $\Phie_{G}(C) \leq K_1 \phicut$  and  $\Phie(G[E^\ast]) \geq \phiemb / K_2$. Therefore, to solve $\ecutemb{\ephicut}{\ephiemb}{\beta}$ on $G$, we just need to solve $\vcutemb{\phicut}{\phiemb}{\beta}$ on $\Gsp$ with parameters  $\phicut = \ephicut  / K_1$ and $\phiemb = K_2 \ephiemb$.
\end{proof}

\begin{theorem}[Deterministic cut or expander]\label{thm-det-cut-expander}
The task $\vcutemb{\phicut}{\phiemb}{\beta}$ can be solved deterministically in $O(D) + (\beta \phicut)^{-O(1)} \cdot 2^{O(\sqrt{\log n \log \log n})}$ rounds on bounded-degree graphs, with $\phiemb = \phicut^{O(1)} \cdot 2^{-O(\sqrt{\log n \log \log n})}$.
The task $\ecutemb{\ephicut}{\ephiemb}{\beta}$ can be solved deterministically in $O(D) + (\beta \ephicut)^{-O(1)} \cdot 2^{O(\sqrt{\log n \log \log n})}$ rounds  on   general graphs, with $\ephiemb = \ephicut^{O(1)} \cdot 2^{-O(\sqrt{\log n \log \log n})}$. %\thatchaphol{I think this should be $\ephiemb \ge \ephicut^{O(1)} \cdot 2^{\Omega(\sqrt{\log n \log \log n})}$}
\end{theorem}
\begin{proof}
The result for $\vcutemb{\phicut}{\phiemb}{\beta}$ follows from \cref{lem-det-cutmatch-main} immediately.
Run the algorithm of \cref{lem-det-cutmatch-main} with parameters $\psi = \phicut$ and $\betaleft = \beta/2$, and
let $W$ and $C$ be the output result.
If $|V \setminus W| \leq \beta|V|$, then $W$ satisfies the requirement of $\vcutemb{\phicut}{\phiemb}{\beta}$ with $\phiemb = \phicut^{O(1)} \cdot 2^{-O(\sqrt{\log n \log \log n})}$. Otherwise, $|V \setminus W| > \beta|V|$, and so we have $|C| \geq (\beta/2)|V|$  due to the requirement $|V \setminus (C \cup W)| \leq \betaleft|V| = (\beta/2)|V|$ in \cref{def-det-sparse-cut}. Therefore, 
$C$ satisfies the requirement of $\vcutemb{\phicut}{\phiemb}{\beta}$, as $(\beta/2)|V| \leq |C| \leq |V|/2$. 
%\thatchaphol{It is not clear why this follows from \cref{lem-det-cutmatch-main}. In \cref{lem-det-cutmatch-main}, most parameters depends on $\psi$, but in this theorem we can choose $\beta$. Also, it might be a good idea to explicitly talk about how $\betacut,\betaleft$ is related to $\beta$.} 
The extension to the conductance-based setting on general graphs is the same as the proof of \cref{thm-rand-cut-expander}.
\end{proof}

\paragraph{Low-diameter decompositions.} 
The round complexity of \cref{thm-rand-cut-expander,thm-det-cut-expander} depends on the Steiner tree diameter $D$, and hence we can afford to apply these algorithms only when $D$ is small.
To cope with this issue, our expander decomposition algorithm employs  the following graph decomposition of Rozho\v{n} and Ghaffari~\cite{RozhonG20}.

Given a parameter $0 < \beta < 1$, the deterministic algorithm of \cref{thm-low-diam-decomp}  decomposes the vertex set $V$ into clusters $V = V_1 \cup V_2 \cup \cdots \cup V_x$ in such a way that 
the number of inter-cluster edges is at most $\beta|E|$, and each cluster $V_i$ is associated with a Steiner tree $T_i$ with diameter    $O(\beta^{-1} \log^3 n)$. Each edge $e \in E$ belongs to at most $O(\log n)$  Steiner trees.
The round complexity of \cref{thm-low-diam-decomp}  is $O(\beta^{-2} \log^6 n)$.

\cref{thm-low-diam-decomp-mod} is a variant of \cref{thm-low-diam-decomp} with a higher round complexity $O(D + \beta^{-2} \log^6 n)$ and a better bound on the  number of inter-cluster edges  $\beta \vol(V \setminus V_{i^\ast})$, where $V_{i^\ast}$ is a cluster with the highest volume.

\subsection{Main Algorithm}

Given an algorithm  of $\ecutemb{\ephicut}{\ephiemb}{\beta}$ as a black box, we can construct an $(\epsilon,\phi)$-expander decomposition of a graph $G=(V,E)$ as follows. We use the following parameters
\[
    \ephicut = \Theta(\epsilon), \ \ 
    \ephiemb = \phi, \ \ 
    \beta_1 = \Theta(\epsilon), \ \
    \beta_2 = \Theta(\epsilon \log^{-1} n), \ \  \text{and} \ \
    \beta_3 = \Theta(\epsilon).
\]
%Our expander decomposition algorithm begins with the low-diameter decomposition algorithm of \cref{thm-low-diam-decomp} with parameter $\beta_1$. This algorithm partitions the vertex set $V$ into clusters  $\VV = \{V_1, V_2, \ldots, V_x\}$, where each part is associated with a Steiner tree $T_i$ of diameter $O(\beta_1^{-1} \log^3 n) = O(\epsilon^{-1} \log^3 n)$, and each edge $e \in E$ belongs to at most $O(\log n)$ Steiner trees. The round complexity of this part is $O(\beta_1^{-2} \log^6 n) = O(\epsilon^{-2} \log^6 n)$.

%All inter-cluster edges are removed.
%Each cluster $V_i \in \VV$ executes the following procedure on the graph $G' = G[V_i]$, in parallel.

%\paragraph{Expander decomposition algorithm.} 
%Using \cref{lem-rand-main}, an expander decomposition can be constructed as follows.
\begin{enumerate}
    \item \label{test1} Compute a low-diameter decomposition $\VV = \{V_1, V_2, \ldots, V_x\}$ of the vertex set $V$.
    If we are at the top level of the recursion, we use
     \cref{thm-low-diam-decomp} with  parameter $\beta_1$.
     Otherwise, we use \cref{thm-low-diam-decomp-mod}  with  parameter $\beta_2$.
    All inter-cluster edges are removed.
    %, and the number of edges removed is at most $\beta_1|E|$.\thatchaphol{Can we have the bound $\le \beta_1 \vol(V\setminus V_1)$ where $V_1$ is the part with biggest volume? If so, then all cuts edges can be charged to smaller side, and the total number of cut edges will be improved to $O(\beta_2 + \beta_2 + \ephicut)|E|\log n$. Then, we can have $\ephicut = \tilde\Omega(\epsilon)$.}
    \item \label{test2} For each $V_i \in \VV$, run an algorithm  $\mathcal{A}$ of $\ecutemb{\ephicut}{\ephiemb}{\beta_3}$  on $G' = (V', E') =  G[V_i]$ in parallel. %Let $\beta_2$ be some parameter. 
    There are two cases.
    \begin{enumerate}
        \item \label{test2a} If the output of  $\mathcal{A}$ is an edge set   $E^\ast \subseteq E'$ with $|E^\ast| \geq (1-{\beta_3}) |E'|$ and $\Phie(G'[E^\ast]) \geq \ephiemb$, then we remove all edges of $G'$ that are not in  $E^\ast$. 
        \item \label{test2b} If the output of  $\mathcal{A}$   is a vertex set $C \subseteq V'$ with $\Omega(\beta_3) \cdot \vol_{G'}(V')\leq \vol_{G'}(C) \leq \vol_{G'}(V')/2$ and $\Phie_{G'}(C) \leq \ephicut$, then we remove the edge set $E(C, V' \setminus C)$ and recurse on $G'[C]$ and $G'[V' \setminus C]$.
    \end{enumerate}
\end{enumerate}

%We set $\phi = \ephiemb$.
At the end of the algorithm, the set of remaining edges induces connected components with conductance at least $\ephiemb = \phi$, as required.

%The purpose of \cref{thm-low-diam-decomp,thm-low-diam-decomp-mod} is to partition the current graph into clusters associated with small-diameter Steiner trees, as the round complexity of $\mathcal{A}$ depends on the Steiner tree diameter. 
Recall that \cref{thm-low-diam-decomp-mod} offers a better upper bound on the number of inter-cluster edges than that of \cref{thm-low-diam-decomp}, but it comes at the cost of having an additional $O(D)$ term in the round complexity. To avoid a linear dependence on the graph diameter in the overall round complexity, we have to use \cref{thm-low-diam-decomp} in the top level of recursion.

Whenever \cref{thm-low-diam-decomp} is applied in our algorithm, the current graph under consideration must be $G'[C]$ or $G'[V' \setminus C]$, where $G'=G[V_i]$ for some cluster $V_i$ in the low-diameter decomposition $\VV = \{V_1, V_2, \ldots, V_x\}$ in the previous level of recursion. When we run the algorithm of \cref{thm-low-diam-decomp} on $G'[C]$ and $G'[V' \setminus C]$, we use the Steiner tree $T_i$ associated with $V_i$, which is guaranteed to have a small diameter $D' = O\left(\max\{\beta_1^{-1}, \beta_2^{-1}\} \log^3 n\right) = O(\epsilon^{-1} \log^4 n)$. 

\paragraph{The number of removed edges.}
We argue that the number of removed edges is at most $\epsilon|E|$.  
\begin{enumerate}
    \item  The low-diameter decomposition at the top layer of recursion uses \cref{thm-low-diam-decomp} with  parameter $\beta_1$, and so the number of edges removed due to this  decomposition  is at most $\beta_1 |E|$.
    \item The rest of the low-diameter decompositions use  \cref{thm-low-diam-decomp-mod}  with  parameter $\beta_2$. Suppose that the decomposition of the current vertex set $V$ is $\VV = \{V_1, V_2, \ldots, V_x\}$, and $V_{i^\ast}$ is a cluster with the highest volume. Then the number of edges removed due to this decomposition  is at most $\beta_2  \vol(V \setminus V_{i^\ast})$. We charge a cost of $2 \beta_
   2$ to each edge $e$ not incident to $V_{i^\ast}$, and a cost of $\beta_2$ to each edge $e \in E(V_{i^\ast}, V \setminus V_{i^\ast})$.
   It is clear that throughout the recursion, each edge $e$ is charged for at most $O(\log n)$ times, because we must have $\vol(V_i) \leq \vol(V)/2$ if $V_i \neq V_{i^\ast}$. Therefore, the total number of edges removed due to the low-diameter decompositions based on \cref{thm-low-diam-decomp-mod} is at most $O(\beta_2 |E| \log n)$.
   \item The total number of edges removed due to Step~2(a) is at most $\beta_3|E|$. 
   \item The total number of edges removed due to Step~2(b) is at most $\ephicut \vol(V) = 2\ephicut|E|$.
\end{enumerate}
To summarize, the total number of removed edges is  \[O\left(\beta_1 + \beta_2  \log n +  \beta_3 + \ephicut \right) \cdot |E|.\]
This number can be made at most $\epsilon |E|$ as long as $\beta_1 = \Theta(\epsilon)$,  
$\beta_2 = \Theta(\epsilon \log^{-1} n)$, 
$\beta_3 = \Theta(\epsilon)$, and
$\ephicut = \Theta(\epsilon)$ are chosen to be sufficiently small.

%It is straightforward to see that the depth of recursion is $O(\beta_2^{-1} \log n)$, as $C$ satisfies $\Omega(\beta_2) \cdot \vol_{G'}(V')\leq \vol_{G'}(C) \leq \vol_{G'}(V')/2$. In each iteration, at most $\beta_1$ fraction of the edges are removed due to the low-diameter decomposition, and at most $\ephicut$ fraction of the edges  are removed due to the output $C$ of algorithm $\mathcal{A}$. Finally, in the last iteration, at most $\beta_2$ fraction of the edges can be removed due to the output $E^\ast$  of algorithm $\mathcal{A}$.
%Therefore, the total number of removed edges is  \[O\left(\beta_2 + (\beta_1 + \ephicut) \cdot (\beta_2^{-1} \log n)\right) \cdot |E|.\]
%We choose $\beta_1 = \ephicut$ and $\beta_2 = \sqrt{\ephicut \log n}$ so that the number of removed edges is $O(\sqrt{\ephicut \log n})$.
%Then we do a change of variable $\ephicut = \Omega(\epsilon^2 \log^{-1} n)$ to make the fraction of the removed edges at most $\epsilon$, satisfying the requirement of an $(\epsilon,\phi)$-expander decomposition.

\paragraph{Round complexity.} We now analyze the round complexity of our algorithm.  It is straightforward to see that the depth of recursion is $t = O(\beta_3^{-1} \log n) = O(\epsilon^{-1} \log n)$, as the cut $C$ computed in Step~\ref{test2} satisfies $\Omega(\beta_3) \cdot \vol_{G'}(V')\leq \vol_{G'}(C) \leq \vol_{G'}(V')/2$.
\begin{enumerate}
    \item The round complexity of \cref{thm-low-diam-decomp} is $O(\beta_1^{-2} \log^6 n) = O(\epsilon^{-2} \log^6 n)$. Note that the algorithm of \cref{thm-low-diam-decomp} is only applied once.
    \item The round complexity of \cref{thm-low-diam-decomp-mod} is $O(D' + \beta_2^{-2} \log^6 n) = O(\epsilon^{-2} \log^8 n)$, as we have $D' = O(\epsilon^{-1} \log^4 n)$.  The algorithm of \cref{thm-low-diam-decomp-mod} is   applied in  $t-1 = O(\epsilon^{-1} \log n)$ levels of recursions, and there is a congestion of $c = O(\log n)$ when the algorithm is applied, as they use the Steiner trees in the low-diameter decomposition in the previous level of recursion. Therefore, the total round complexity associated with  \cref{thm-low-diam-decomp-mod} is  $O(\epsilon^{-2} \log^8 n) \cdot  O(\epsilon^{-1} \log n) \cdot O(\log n) = O(\epsilon^{-3} \log^{10} n)$.
\item Denote $T$ as the round complexity of $\mathcal{A}$ with $D' = O(\epsilon^{-1} \log^4 n)$. Since $\mathcal{A}$ is applied in  $t = O(\epsilon^{-1} \log n)$ levels of recursions with a congestion of $c = O(\log n)$, the total round complexity associated with $\mathcal{A}$ is $T \cdot  O(\epsilon^{-1} \log n) \cdot O(\log n) = O(T \epsilon^{-1} \log^{2} n)$.
\end{enumerate}
We can express the round complexity of our expander decomposition algorithm as 
\[ O(\epsilon^{-2} \log^6 n) + O(\epsilon^{-3} \log^{10} n) + O(T \epsilon^{-1} \log^{2} n) = O(\epsilon^{-3} \log^{10} n) + O(T \epsilon^{-1} \log^{2} n).\]

%where $O\left(\sqrt{\ephicut^{-1} \log n}\right) = O(\beta_2^{-1} \log n)$ is the number of iterations, $O(\beta_1^{-2} \log^6 n) = O(\ephicut^{-2} \log^6 n)$ is the round complexity of the low-diameter decomposition algorithm of \cref{thm-low-diam-decomp}, and $T$ is the round complexity of $\mathcal{A}$ with the effective Steiner tree diameter $D' = cD$, where $c = O(\log n)$ is the congestion of the Steiner trees resulting from \cref{thm-low-diam-decomp}, and $D = O(\beta_1^{-1} \log^3 n) = O(\ephicut^{-1} \log^3 n)$ is the diameter of the Steiner trees resulting from \cref{thm-low-diam-decomp}.

\paragraph{Randomized setting.}
In \cref{thm-rand-cut-expander} we use \[\phi = \ephiemb = 1/ \poly(\ephicut^{-1}, \log n) = 1/ \poly(\epsilon^{-1}, \log n),\] and hence the randomized round complexity of $\mathcal{A}$ is \[T = O(D') \cdot \poly(\epsilon^{-1}, \log n )  
= \poly(\epsilon^{-1}, \log n).\] We conclude the following theorem.

\Rrand*

Here the conductance parameter $\phi$ can be made as large as $\phi = \ephiemb  = \Omega\left(\epsilon^{3} \log^{-10} n \right)$, as we have $\ephicut = \Theta(\epsilon)$  and  $\ephiemb = \Omega\left(\ephicut^3 \log^{-10} n\right)$ in view of \cref{thm-rand-cut-expander,lem-rand-cutmatch-main-bounded-deg}.
%\thatchaphol{As we should be able to set $\ephicut = \Omega(\epsilon \log^{-1} n)$, we can have $\phi = \tilde\Omega(\epsilon^3)$.}

\paragraph{Deterministic setting.}
%For the deterministic algorithm, we all parameters $\beta_1,\beta_2,\ephicut$ in the same way in the randomized algorithm above, except that, from
In \cref{thm-det-cut-expander}, we use 
\[\phi = \ephiemb = \poly(\ephicut) 2^{-O(\sqrt{\log n \log \log n})} = \epsilon^{O(1)} 2^{-O(\sqrt{\log n \log \log n})},\]
%In \cref{thm-det-cut-expander} we have $\ephiemb = \poly(\epsilon) 2^{-\Omega(\sqrt{\log n \log \log n})}$, and 
and hence the deterministic round complexity of $\mathcal{A}$ is \[ T = O(D')+\epsilon^{-O(1)}2^{O(\sqrt{\log n \log \log n})} = \epsilon^{-O(1)} 2^{O(\sqrt{\log n \log \log n})}.\]
%as we note that both $1/\ephicut$ and $1/\beta_2$ are $\poly(\epsilon^{-1}, \log n)$, and the $\poly(\log n)$ part can be absorbed into $2^{O(\sqrt{\log n \log \log n})}$.
We conclude the following theorem.

\Rdet*

More generally, taking into consideration the tradeoff of parameters in \cref{lem-det-cutmatch-main,thm-solve-recursion}, for any $1 \geq \gamma \geq \sqrt{\log \log n / \log n}$, there is a deterministic expander decomposition algorithm with round complexity $\epsilon^{-O(1)}  \cdot n^{O(\gamma)}$ with parameter $\phi = \epsilon^{O(1)}  \log^{-O(1/\gamma)}n$.

\section{Randomized Sparse Cut Computation}\label{sect-rand-expander-decomp}

The goal of this section is to prove \cref{lem-rand-cutmatch-main-bounded-deg}. Intuitively, in \cref{lem-rand-cutmatch-main-bounded-deg}, the set $W$ \emph{certifies} that $C$ is \emph{nearly most balanced} in the following sense. If $|C| \leq |V|/20$, then for any cut $C'$ with $9 |C| \leq |C'| \leq |V|/2$, we have \[\Phiv(C') = \frac{|\partial(C')|}{|C'|}
\geq  \frac{|\partial(C')|}{9 |C' \cap W|}
\geq  \frac{|\partial_{G[W]}(C' \cap W)|}{9 |C' \cap W|}
\geq \frac{\Phiv(G[W])}{9} = \Omega(\psi^{3} \log^{-10}n).\]
In other words, there does not exist a cut $C'$ that is significantly more balanced than $C$ and significantly sparser than $C$ at the same time. 

\begin{theorem}[Randomized sparse cut computation]\label{lem-rand-cutmatch-main-bounded-deg}
%For $\Delta = O(1)$, $\beta = 1/3$, and $f_n(\phi) = \Omega(\psi^{3} \log^{-10}n)$, we have
%\[
%\Trand(n, D, \Delta, \psi, f, \beta) \leq O\left(D + \psi^{-5} \log^{11} n \right).
%\]
%Specifically, 
Let $G=(V,E)$ be a bounded-degree graph, and let $0 < \psi < 1$ be any parameter. There is a randomized algorithm with round complexity
\[O\left(\psi^{-2} \log^6 n \left(D + \psi^{-3} \log^4 n\right) \right)\]
that finds $W \subseteq V$ and $C \subseteq V$ meeting the following requirements with high probability.
\begin{description}
    \item[Expander:] The induced subgraph $G[W]$ has  $\Phiv(G[W])  = \Omega(\psi^{3} \log^{-10}n)$.
    \item[Cut:] The cut $C$ satisfies  $|V \setminus W|/8 \leq |C| \leq |V|/2$ and $\Phiv(C) \leq \psi$.
\end{description}
\end{theorem}

We note that the theorem above allows $W = \emptyset$ or $C = \emptyset$.  If $W = \emptyset$, then $|C| = \Omega(|V|)$, i.e., $C$ is a $\Theta(1)$-balanced sparse cut. 
If $C = \emptyset$, then $W = V$, i.e., $G$ is an $\Omega(\psi^{3} \log^{-10}n)$-expander.

\begin{proof}
%We will prove that for the case $\beta = 1/8$ and $f_n(\phi) = \Omega(\psi^{3} \log^{-10}n)$, we can obtain the round complexity of $O\left(\psi^{-2} \log^6 n \left(D + \psi^{-3} \log^4 n\right) \right)$. This is sufficient to prove the lemma, as we can use the reductions in \cref{lem-diam-reduce,lem-Bal-improve} to improve the balance parameter $\beta$ and to make the dependence on $D$ additive, at the cost of worsening $f$ by a constant factor.
Apply the algorithm of \cref{lem-rand-cut-embed} with the same parameter $\psi$. If the cut $C \subseteq V$ returned by the algorithm already has $|C| \geq |V|/8$, then output this cut $C$ with $W = \emptyset$. 
Otherwise, the algorithm of \cref{lem-rand-cut-embed} must return both $H$ and $C$.
Run the algorithm of \cref{lem-expander-from-random-walk} with these $H$ and $C$, and parameters $c = O(\psi^{-2} \log^4 n)$, $d = O(\psi^{-1} \log n)$, and $k = O(\log^2 n)$. Since $|C| < |V|/8$, the algorithm of \cref{lem-expander-from-random-walk} must return a subset $W \subseteq V$ with $|C| \geq |V \setminus W| / 4 >  |V \setminus W| / 8$ and $\Phiv(G[W]) = \Omega(c^{-1} d^{-1} k^{-2} \log^{-1} n) = \Omega(\psi^{3} \log^{-10}n)$, as required. The round complexity of the algorithm of \cref{lem-rand-cut-embed} is
\[
O\left(\psi^{-2} \log^6 n \left(D + \psi^{-3} \log^4 n\right) \right),
\]
which dominates the round complexity of \cref{lem-expander-from-random-walk}, which is
\[O(D + ck \log^2 n + dk) = O(D + \psi^{-2} \log^8 n).\qedhere\] 
\end{proof}

It remains to prove  \Cref{lem-rand-cut-embed} and \Cref{lem-expander-from-random-walk}.
In~\cref{sect-RST}, we give a distributed implementation of a version of the cut-matching game~\cite{RST14,SaranurakW19} and prove \cref{lem-rand-cut-embed}. In \cref{sect-expander-emb-rand}, we show how to recover a large subgraph $G[W]$ of high $\Phiv(G[W])$ from the embedding of the matchings in $H = (M_1, M_2, \ldots, M_k)$ found during the cut-matching game, and prove \cref{lem-expander-from-random-walk}.
This lemma is new and  is crucial for bypassing the expander trimming technique from \cite{SaranurakW19} whose distributed implementation is not known.

\subsection{Cut-matching Game}\label{sect-RST}

We assume $V = \{v_1, v_2, \ldots, v_n\}$. Given a matching $M$ of the vertex set $V$, define the doubly stochastic matrix $F_M \in \mathbb{R}^{n \times n}$~\cite{KRV09} as
\begin{equation*}
    F_M[i,j] = \begin{cases}
             1 &   \text{if ($i=j$ and $v_i$ is not matched in $M$),}\\
             
               1/2               & \text{if ($i=j$  and $v_i$ is matched in $M$) or ($i \neq j$ and $\{v_i, v_j\} \in M$),}\\
               0               & \text{if ($i \neq j$ and $\{v_i, v_j\} \notin M$).}
           \end{cases}
\end{equation*}
 Suppose $H = (M_1, M_2, \ldots,  M_k)$ be a  sequence of $k$ matchings. We write 
$F_H = F_{M_k} \cdot F_{M_{k-1}} \cdots F_{M_1}$. For the special case of $H = \emptyset$, $F_H$ is defined as the $n \times n$ identity matrix.
Intuitively, the matrix $F_H$ represents the following random walk $(u_0, u_1, \ldots, u_k)$. Start at $u_0$. Then, for $i = 1, 2, \ldots, k$, do the following. If $u_{i-1}$ is matched to $w$ in the matching $M_i$, then there is a probability of $1/2$ that $u_{i} = w$; otherwise $u_{i} = u_{i-1}$. We call such a random walk an $F_H$-random walk. Note that the ordering of the matchings $H = (M_1, M_2, \ldots,  M_k)$  affects the random walk.
It is straightforward to verify that the probability that an $F_H$-random walk starting at $u_0 = v_j$ ending at $u_k = v_i$ equals $F_H[i, j]$. For convenience, we also write $$p_H(j \rightsquigarrow i) = F_H[i, j]$$ to denote this random walk probability. 
We define $F[i] \in \mathbb{R}^n$ to be the length-$n$ vector corresponding to the $i$th row of $F$. That is, $F[i][j] = F[i, j]$. Since $F$ is doubly stochastic, we have
\[
\sum_{1 \leq j \leq n} p_H(j \rightsquigarrow i) = 
\sum_{1 \leq j \leq n} F[i][j]
= 1.
\]
The goal of \cref{sect-RST} is to prove \cref{lem-rand-cut-embed}.
 
\begin{lemma}[Randomized cut or embedding]\label{lem-rand-cut-embed}
Let $G=(V,E)$ be a bounded-degree graph. Let $0 < \psi < 1$ be any parameter. There is a randomized algorithm with round complexity
\[
O\left(\psi^{-2} \log^6 n \left(D + \psi^{-3} \log^4 n\right) \right)
\]
that achieves the following task with high probability.
\begin{description}
    \item[Cut:] The algorithm is required to output a cut $C$ satisfying  $0 \leq |C| \leq |V|/2$ and $\Phiv(C) \leq \psi$.
    \item[Embedding:] If $|C| < (3/11)|V|$, the algorithm is required to find a sequence of matchings $H = (M_1, M_2, \ldots,  M_k)$ with $k= O(\log^2 n)$, where each $M_i$ is a matching that is embedded into $V$ with congestion $c = O(\psi^{-2} \log^4 n)$ and dilation $d = O(\psi^{-1} \log n)$. Furthermore, 
    $$\forall_{v_i, v_j \in V\setminus C}, \ \left|p_H(i \rightsquigarrow j) - \frac{1}{|V \setminus C|}\sum_{v_l \in V \setminus C} p_H(i \rightsquigarrow l) \right| \leq \frac{1}{4n}.$$
\end{description}
\end{lemma}

\cref{lem-rand-cut-embed} is proved using the \emph{cut-matching game} of~\cite{RST14,SaranurakW19}. The cut-matching game proceeds in iterations $i = 1, 2, \ldots, \tau = \Theta(\log^2 n)$.
\begin{description}
\item[Active vertices:] At the beginning of iteration $i$, there is a set of \emph{active vertices} $A_i$. Initially we set $A_1 = V$. For $i > 1$, we set $A_{i} = A_{i-1} \setminus C_{i-1}$, where $C_{i-1}$ is a cut found during the $(i-1)$th iteration.
\item[Cut player:] In iteration $i$, the \emph{cut player} finds two disjoint subsets $A^l \subseteq A_i$ and $A^r \subseteq A_i$ satisfying $|A^l| \leq |A_i|/8$ and $|A^r| \geq |A_i|/2$.
\item[Matching player:] Given $(A^l, A^r)$, the matching player finds a cut $C_i$ and a matching $M_i$ with its embedding $\PP_i$ satisfying the following conditions.
\begin{description}
\item[Match:] $M_i$  is a matching between $A^l$ and $A^r$,  $\PP_i$ is a set of $A^l$-$A^r$ paths in the subgraph $G[A_i]$ that embeds $M_i$ with congestion $c = O(\psi^{-2} \log^4 n)$ and dilation $d = O(\psi^{-1} \log n)$.
\item[Cut:] $C_i \subseteq A_i$ is a cut of the subgraph   $G[A_i]$ with $\Phiv_{G[A_i]}(C_i) \leq \psi' = (3/8)\psi$. Furthermore, all the unmatched vertices in $A^l$ belong to $C_i$, and all the unmatched vertices in $A^r$ belongs to $A_i \setminus C_i$.
\end{description}
\item[Terminating condition 1:] If $|C_1 \cup C_2 \cup \cdots C_i| \geq (3/11)|V|$, the cut-matching game is terminated, and the output $C$ is the one of $C_1 \cup C_2 \cup \cdots \cup C_i$ and $V \setminus (C_1 \cup C_2 \cup \cdots \cup C_i)$ that has the smaller size. 
\item[Terminating condition 2:] If $i = \tau$ is the last iteration, then output $C = C_1 \cup C_2 \cup \cdots \cup C_\tau$ and $H =  (M_{1}, M_{2}, \ldots,  M_\tau)$, together with the embedding $\PP_j$ of  $M_j$, for each $1 \leq j \leq \tau$. 
\end{description}

Although there are still some missing details of the algorithm, we have enough information to show that  the  cut $C$ returned by the algorithm  satisfies all the requirements stated in \cref{lem-rand-cut-embed}.

\begin{lemma}[Property of the output cut $C$]\label{lem-rand-cut-match-bal-cut}
The cut $C$ returned by the algorithm satisfies  $0 \leq |C| \leq (1/2)|V|$ and $\Phiv(C) \leq \psi$.
Furthermore, if the algorithm 
terminates with a cut $C$ only, then the cut $C$ additionally satisfies  $|C| \geq (3/11)|V|$.
\end{lemma}
\begin{proof}
We first consider the case the algorithm is terminated because $i = \tau$ is the last iteration. In this case, the output cut $C$ is $C = C_1 \cup C_2 \cup \cdots \cup C_i$, and it satisfies $|C| = |C_1 \cup C_2 \cup \cdots \cup C_i| < (3/11)|V| < |V|/2$. The following upper bound of $|\partial(C)|$ shows that $\Phiv(C) \leq \psi' = (3/8)\psi < \psi$.
\begin{align*}
|\partial(C)| = 
|\partial(C_1 \cup C_2 \cup \cdots \cup C_i)| &\leq \sum_{1 \leq j \leq i} |\partial_{G[A_i]}(C_i)|\\
&\leq \sum_{1 \leq j \leq i} \psi'|C_i|\\
& = \psi' |C_1 \cup C_2 \cup \cdots \cup C_i|.
\end{align*}
For the rest of the proof, we consider the case the algorithm is terminated because $|C_1 \cup C_2 \cup \cdots \cup C_i| \geq (3/11)|V|$. Since the algorithm does not terminate in iteration $i-1$, we have $|C_1 \cup C_2 \cup \cdots \cup C_{i-1}| < (3/11)|V|$, and so $|A_i| \geq (8/11)|V|$. Recall that the cut $C_i$ of $G[A_i]$ satisfies that all vertices in $A^r$ unmatched in $M_i$ must be in $A_i \setminus C_i$. Since $|M_i| \leq |A^l| \leq |A_i|/8$ and $|A^r| \geq |A_i|/2$, we have $|A_i \setminus C_i| \geq
|A^l| - |M_i| \geq |A_i|/2 - |A_i|/8 =
(3/8)|A_i| \geq (3/11)|V|$. Therefore, we not only have $|C_1 \cup C_2 \cup \cdots \cup C_i| \geq (3/11)|V|$ but also have $|V \setminus (C_1 \cup C_2 \cup \cdots C_i)| = |A_i \setminus C_i| \geq (3/11)|V|$.
Since the output $C$ is the one of $C_1 \cup C_2 \cup \cdots \cup C_i$ and $V \setminus (C_1 \cup C_2 \cup \cdots \cup C_i)$ that has the smaller size, we conclude that \[\frac{3}{11}|V| \leq |C| \leq \frac{1}{2}|V|.\]
It remains to show that $\Phiv(C) \leq \psi$.
In view of the above discussion, we have $|V \setminus (C_1 \cup C_2 \cup \cdots \cup C_i)| \geq (3/8) |C_1 \cup C_2 \cup \cdots \cup C_i|$. Combining this with the upper bound $|\partial(C)| \leq \psi' |C_1 \cup C_2 \cup \cdots \cup C_i|$ above, we have
\begin{align*}
|\partial(C)| &\leq \psi' |C_1 \cup C_2 \cup \cdots \cup C_i| \\
&\leq (8/3)\psi' |V \setminus (C_1 \cup C_2 \cup \cdots \cup C_i)|\\
&= \psi |V \setminus (C_1 \cup C_2 \cup \cdots \cup C_i)|.
\end{align*}
Therefore, 
\[\Phiv(C) = \frac{|\partial(C)|} {\min\left\{|C_1 \cup C_2 \cup \cdots \cup C_i|, |V \setminus (C_1 \cup C_2 \cup \cdots \cup C_i)|\right\}} \leq \psi,
\]
as required.
\end{proof}

To analyze the properties of the sequence of the matchings $H$ returned by the algorithm, we define the following notations for each $1 \leq i \leq \tau$.
\[
H_i = (M_1, M_2, \ldots, M_i), \ \ \ 
\text{and} \ \ \
F_i = F_{H_i} =  F_{M_i} \cdot F_{M_{i-1}} \cdots F_{M_1}.
\]
For the special case of $i= 0$, we have $H_0 = \emptyset$, and $F_0$ is the $n \times n$ identity matrix. 
%For each $1 \leq i  \leq n$, we write $e_i \in \mathbb{R}^{n}$ to denote the $i$th standard basis vector: $e_i[j] = 1$ if $j=i$, and  $e_i[j] = 0$ otherwise. Note that $F_i \cdot e_j$ is the length-$n$ vector corresponding to the probability distribution of the $F_{H_i}$-random walk starting at $v_j$.
We consider the following potential function~\cite{KRV09,RST14,SaranurakW19}.
\[\Pi(i) = \sum_{v_j \in A_{i+1}} \| F_i[j] - \mu_i \|^2, \ \ \text{where} \ \ \mu_i = \frac{1}{|A_{i+1}|} \cdot\sum_{v_j \in A_{i+1}} F_i[j].\]

%Intuitively, if $\Pi(i)$ is small, then the probability distributions of the $F_{H_i}$-random walk starting at $v_j$ among all $v_j \in A_i$ are close to their average.
%The matrix $F_t$ represents a random walk using the matchings found in the first $t$ iterations. Roughly speaking, the potential function $\Pi(t)$ measures the distance between $F_t$ and the uniform distribution, 
Remember that $F_i[j]$ is the $j$th row of $F_i$, which is the vector \[\left(p_{H_i}(1 \rightsquigarrow j), p_{H_i}(2 \rightsquigarrow j), \ldots, p_{H_i}(n \rightsquigarrow j)  \right)^\top\] of the probabilities of $F_i$-random walks ending in $v_j$, and so $\mu_i$ is the average of  such  vectors among all $v_j \in A_{i+1}$. 
Note that $\Pi(i)$  only takes into account the probability mass that remains in $A_{i+1}$, and so $\Pi(i)$ being very small does not imply that a $F_{H_i}$-random walk is close to the uniform distribution.

\begin{lemma}[Property of the output sequence of matchings $H$]
Let $H = H_\tau = (M_1, M_2, \ldots,  M_{\tau})$,   $C = C_1 \cup C_2 \cup \cdots \cup C_\tau$ and $A_{\tau+1} = V \setminus (C_1 \cup C_2 \cup \cdots \cup C_\tau)$  in the final iteration $i = \tau$. 
\[
\text{If} \ \ \Pi(\tau) \leq \frac{1}{64 n^2}, \ \ \ \text{then} \ \ 
\forall_{v_i, v_j \in V\setminus C}, \ \left|p_H(i \rightsquigarrow j) - \frac{1}{|V \setminus C|}\sum_{v_l \in V \setminus C} p_H(i \rightsquigarrow l) \right| \leq \frac{1}{4n}.
\]
\end{lemma}
\begin{proof}
%Since $H$ is the reversal of $H_{\tau}$, we have 
%\[p_H(i \rightsquigarrow j) = p_{H_{\tau}}(j \rightsquigarrow i).\]
%More formally, this is because 
%\[
%F_{H} = M_1 \cdot M_2 \cdots M_{\tau}
%= \left( M_{\tau}^\top \cdot M_{\tau - 1}^\top \cdot M_{1}^\top \right)^\top
%= \left( M_{\tau} \cdot M_{\tau - 1}\cdot M_{1} \right)^\top
%= F_{H_\tau}^\top.
%\]
%Therefore, to prove this lemma, it suffices to show that
%\[
%\text{If} \ \ \Pi(\tau) \leq \frac{1}{64 n}, \ \ \ \text{then} \ \ 
%\forall_{v_i, v_j \in A_{\tau+1}}, \ \left|p_{H_\tau}(j \rightsquigarrow i) - \frac{1}{|A_{\tau+1}|}\sum_{v_l \in A_{\tau+1}} p_{H_\tau}(l \rightsquigarrow i) \right| \leq \frac{1}{4n}.
%\]
Consider any $v_i, v_j \in A_{\tau+1}$.
Recall that $$p_{H_\tau}(i \rightsquigarrow j) = F_\tau[j][i]$$ and $$\frac{1}{|A_{\tau+1}|}\sum_{v_l \in A_{\tau+1}} p_{H_\tau}(i \rightsquigarrow l)  = \mu_{\tau}[i].$$
Now suppose $|F_\tau[j][i] - \mu_\tau[i]| > 1/(4n)$. Then 
\[\Pi(\tau) = \sum_{v_i \in A_{\tau+1}} \| F_\tau[j] - \mu_\tau \|^2 \geq \left(F_\tau[j][i] - \mu_\tau[i]\right)^2  > \frac{1}{64n^2},\]
which is a contradiction.
\end{proof}

To make $H = H_\tau$  satisfy the requirement of \cref{lem-rand-cut-embed}, it suffices to have $\Pi(\tau) \leq 1/(64n^2)$. 
%The proof of the following lemma is left to \cref{sect-rand-potential}.

\begin{lemma}[Number of iterations]\label{lem-rand-cutmatch-iterations}
Suppose $\tau \geq K \log^2 n$, for some sufficiently large constant $K$. Then we have $\Pi(\tau) \leq 1/(64n^2)$ with high probability.
\end{lemma}
\begin{proof}
Initially, $H_0$ is the identity matrix and $\mu_0=(1/n, 1/n, \ldots, 1/n)^\top$, and so  \[\Pi(0) = (n^2 -n) \cdot \left(0 - \frac{1}{n}\right)^2 + n \cdot \left(1 - \frac{1}{n}\right)^2 = n-1,\] and so it suffices to show that the potential decreases by a factor of $1 - \Omega(1 / \log n)$ in expectation in each iteration $i$. 
As the analysis is the same as in \cite{RST14,SaranurakW19},
we defer the analysis of the potential drop to \cref{lem-potential-drop} in \cref{sect-rand-potential}.
\end{proof}

 %For the graph $H = H_\tau$ (with $C = C_1 \cup C_2 \cup \cdots \cup C_\tau$ and $A_{\tau+1} = V \setminus (C_1 \cup C_2 \cup \cdots \cup C_\tau)$) to satisfy the following property stated in \cref{lem-rand-cut-embed},
%    $$\forall_{v_i \in V}, \  \forall_{v_j \in A_{\tau+1}}, \ \left|p_H(i \rightsquigarrow j) - \frac{1}{| A_{\tau+1} |}\sum_{v_l \in A_{\tau+1}} p_H(i \rightsquigarrow l) \right| \leq \frac{1}{4n},$$
%it suffices that $\Pi(\tau) \leq 1/(64n^2)$. Note that initially, $\Pi(0) = n^2 \cdot (1 - 1/n)^2$, and so it suffices to show that the potential decreases by a factor of $1 - \Omega(1 / \log n)$ in expectation, in each iteration, as $\tau = \Theta(\log^2 n)$. We defer the analysis of the potential drop to \cref{lem-potential-drop} in \cref{sect-rand-potential}.

%In particular, $\Pi(t)$ being very small does not have any implication to the conductance of $H_t$. In contrast, if $A_t = V$, then $\Pi(t) = O(n^{-2})$ guarantees that $H_t$ has conductance $\Omega(1)$~\cite[Lemma 3.1]{KRV09}.

%In the $i$th iteration of the cut-matching game, the cut player construct two disjoint vertex subsets $A^l$ and $A^r$ with $|A^l| \leq |A^r|$, and it asks the matching player to embed a matching $M_i$ between them.

%It is possible that some vertices in $A^l$ are not matched. In this case, the matching player is required to also output a $O(\phi)$-sparse cut $C_t$ containing all these vertices.
 %We update the set of active vertices by $A_{i+1} = A_i \setminus C_t$.

\paragraph{Cut player.} We describe the algorithm~\cite{RST14} for the cut player to choose $A^l$ and $A^r$ in   iteration $i$ of the cut-matching game. 
\begin{enumerate}
    \item Select $r \in \mathbb{R}^{n}$ to be a random unit vector, where $r[j]$ is associated with the vertex $v_j \in V$.
    \item Calculate $u = F_{i-1} \cdot r \in \mathbb{R}^{n}$, and $\bar{u} = \sum_{v_j \in A_i} u[j] / |A_i|$.
    \item Define the two sets $L = \{v_j \in A_i \ | \ u[j] < \bar{u} \}$ and $R = A_i \setminus L$. For the $S = L$ or $S = R$,  define the two numbers $P_S = \sum_{v_j \in S}(u[j] - \bar{u})^2$ and $\ell_S = \sum_{v_j \in S} |u[j] - \bar{u}|$. The two sets  $A^l$ and $A^r$  are chosen as follows.
    \begin{enumerate}
        \item If $|L| \leq |R|$ and $P_L \geq P_A / 20$, then $A^l$ is the $|A_i|/8$ vertices in $A_i$ with the smallest $u$-value, and $A^r = R$.
        \item If $|L| \leq |R|$ and $P_L < P_A / 20$, then $A^l$ is chosen as the $\max\{ |R'|, |A_i|/8 \}$ vertices in $R' = \{v_j \in A_i \ | \ u[j] \geq \bar{u} + 6\ell_R / |A_i| \}$ with the largest $u$-value, and $A^r = \{v_j \in A_i \ | \ u[j] \leq \bar{u} + 4\ell_R / |A_i| \}$.
        \item If $|L| > |R|$ and $P_R \geq P_A / 20$, then $A^l$ is the $|A_i|/8$ vertices in $A_i$ with the largest $u$-value, and $A^r = L$.        
        \item If $|L| > |R|$ and $P_R < P_A / 20$, then $A^l$ is chosen as the $\max\{ |L'|, |A_i|/8 \}$ vertices in $L' = \{v_j \in A_i \ | \ u[j] \leq \bar{u} - 6\ell_L / |A_i| \}$ with the smallest $u$-value, and $A^r = \{v_j \in A_i \ | \ u[j] \geq \bar{u} - 4\ell_L / |A_i| \}$.
    \end{enumerate}
\end{enumerate}
This algorithm is taken from the proof of~\cite[Lemma~3.3]{RST14}.
The reason that $A^l$ and $A^r$ are selected this way is to fulfill the requirements in \cref{lem-rst-property}.
The above procedure can be implemented in the distributed setting efficiently. Recall that $c$ and $d$ are the congestion and dilation of the embedding of each matching $M_j$.

\begin{lemma}[Round complexity of the cut player] \label{lem-cut-pl-impl}
The algorithm for the cut player in iteration $i$ costs $O(D \log n + cdi)$ rounds with high probability.
\end{lemma}
\begin{proof}
For sampling the unit vector $r$, we can simply let each vertex $v_j \in V$ samples from the standard normal distribution, and then re-scale the sampled values of all vertices to make it a unit vector. This can be implemented in $O(D)$ rounds via a straightforward information gathering along the Steiner tree (\cref{lem-basic}). %\thatchaphol{Say the name of the technique a bit}

The calculation of $u = F_i \cdot r$ can be done in $O(cdi)$ rounds by executing the matrix-vector multiplication directly by walking along the matchings $M_1, M_2, \ldots, M_{i-1}$. 

The calculation of $A^l$ and $A^r$ can be done in $O(D \log n)$ rounds with high probability, using a binary search (\cref{lem-rand-bsearch}). %\thatchaphol{Say the name of the technique a bit}
\end{proof}

\paragraph{Matching player.} The task of the matching player can be solved with high probability using \cref{lem-matching-player-rand} with parameters $\psi'$, $S = A^l$, and $T = A^r$. %\thatchaphol{Say the name of the technique a bit} 
The algorithm of \cref{lem-matching-player-rand} either finds a cut $C_i$ or finds a matching $M_i$ with its embedding $\PP_i$ meeting our requirements.
The round complexity of the algorithm is  $O(\psi^{-2} \log^4 n (D + \psi^{-3} \log^4 n) )$ rounds.
%with high probability. 

Note that the round complexity for algorithm of the cut player $O(cdi + D\log n) = O(\psi^{-3} \log^7 n + D \log n)$ (\cref{lem-cut-pl-impl}) is dominated by the round complexity for the algorithm of the matching player. Since there are $\tau = O(\log^2 n)$ iterations, we conclude that the overall round complexity of the cut-matching game is 
\[O(\tau \cdot \psi^{-2} \log^4 n (D + \psi^{-3} \log^4 n) ) 
= O(\psi^{-2} \log^6 n (D + \psi^{-3} \log^4 n) ).\]

\subsection{Extracting a Well-connected Subgraph}\label{sect-expander-emb-rand}

The goal of this section is to prove \cref{lem-expander-from-random-walk}, which allows us to transform the embeddings for $H = (M_1, M_2, \ldots,  M_k)$ in to a subgraph $G[W]$ with high $\Phiv(G[W])$, for the case $|C| < (1/8)|V|$ in the outcome of the cut-matching game.

\begin{lemma}[Well-connected subgraph via random walk]\label{lem-expander-from-random-walk}
Consider a bounded-degree graph $G = (V,E)$. Suppose we are given a subset $C \subseteq V$ with $|C| < (1/8)|V|$ and a sequence of matchings $H = (M_1, M_2, \ldots,  M_k)$, where each $M_i$ is a matching that is embedded into $V$ with congestion $c$ and dilation $d$. Furthermore, 
    $$\forall_{v_i, v_j \in V\setminus C}, \ \left|p_H(i \rightsquigarrow j) - \frac{1}{|V \setminus C|}\sum_{v_l \in V \setminus C} p_H(i \rightsquigarrow l) \right| \leq \frac{1}{4n}.$$
Then there is a randomized algorithm with round complexity $O(D + ck \log^2 n + dk)$ that, with high probability, finds a subset $W \subseteq V$ satisfying
\begin{align*}
   &|W| \geq |V| - 4|C| \ \ \  \text{and} \ \ \
   \Phiv(G[W]) = \Omega(c^{-1}d^{-1}k^{-2} \log^{-1} n).
\end{align*}
\end{lemma}

Note that the above algorithm guarantees that $W \subseteq V\setminus C$ and works for general graphs, then it would give the same qualitative guarantee as in the expander trimming algorithm by \cite{SaranurakW19} (with worse quantitative guarantees). Here, we relax this requirement.
For the rest of this section, suppose we are given $C$ and $H = (M_1, M_2, \ldots,  M_k)$ described in \cref{lem-expander-from-random-walk}.
 We select the set of \emph{good vertices} $U \subseteq V$ as follows. Each vertex $v \in V \setminus C$ initiates $\Theta(\log n)$ $F_H$-random walks, and then it calculates the fraction of walks that end up in $V \setminus C$. If the fraction is at least $1/2$, then add $v$ to $U$. \cref{lem-property-good-v-chernoff} is a straightforward consequence of a Chernoff bound.

\begin{lemma}\label{lem-property-good-v-chernoff}
With high probability, the following is true for all vertices  $v_i \in V \setminus C$.
\begin{itemize}
    \item If $\sum_{v_l \in V \setminus C} p_H(i \rightsquigarrow l) \geq 2/3$, then $v_i \in U$. 
    \item  If $\sum_{v_l \in V \setminus C} p_H(i \rightsquigarrow l) \leq 1/3$, then $v_i \notin U$. 
\end{itemize}
\end{lemma}
\begin{proof}
Suppose each vertex $v_i \in V \setminus C$ initiates $K \log n$ $F_H$-random walks. Let $X$ be the number of walks starting from $v_i$ ending in $V\setminus C$. The criterion for $v_i$ to be included in $U$ is   $X \geq (K/2) \log n$. If $\sum_{v_l \in V \setminus C} p_H(i \rightsquigarrow l) \geq 2/3$, then $\mu = \Expect[X] \geq (2K/3) \log n$, and so the probability that $v_i \notin U$ is  
\[
\Prob\left[X <  (K/2) \log n\right]
= \Prob\left[X <  \left(1 + \frac{1}{3}\right)\mu \right]
\leq \exp\left(- \frac{1}{2} \cdot \frac{1}{3^2} \cdot \mu \right)
= n^{-\Omega(K)},
\]
by a Chernoff bound $\Prob[X \leq (1+\delta)\mu] \leq \exp(-\delta^2 \mu)$. The case of $\sum_{v_l \in V \setminus C} p_H(i \rightsquigarrow l) \leq 1/3$ is similar.
\end{proof}

The following lemma gives a lower bound on the number of good vertices, and it also shows that the random walk probability between any two good vertices must be $\Theta(1/n)$.

\begin{lemma}\label{lem-property-good-v}
With high probability, the set $U$ satisfies 
\begin{align*}
  &|U| \geq |V| - 4 |C| \geq \frac{|V|}{2} \ \ \ \text{and} \ \ \
  \frac{1}{12n} \leq p_H(i \rightsquigarrow j) \leq \frac{85}{84n} \ \ \text{for all $v_i, v_j \in U$}.
\end{align*}
\end{lemma}
\begin{proof}
Observe that 
\[\sum_{v_i \in V, \; v_j\in C} p_H(i \rightsquigarrow j) =
\sum_{v_i \in V, \; v_j\in C} p_H(j \rightsquigarrow i)
= |C|,
\]
and so the number of vertices $v_i \in V \setminus C$ with at least $1/3$ probability of $F_H$-random walk landing in $C$ must be at most $3|C|$. Therefore, \cref{lem-property-good-v-chernoff} implies $|U| \geq |V \setminus C| - 3|C| = |V| - 4|C|$. 

For the rest of the proof, we calculate $p_H(i \rightsquigarrow j)$ for any $v_i, v_j \in U$. 
Recall that in the assumption in \cref{lem-expander-from-random-walk}, we have
    $$\left|p_H(i \rightsquigarrow j) - \frac{1}{|V \setminus C|}\sum_{v_l \in V \setminus C} p_H(i \rightsquigarrow l) \right| \leq \frac{1}{4n}.$$
We also have the following two bounds, where the first one is due to the assumption that $|C| < (1/8) |V|$  in \cref{lem-expander-from-random-walk}, and the second one is due to \cref{lem-property-good-v-chernoff}.
\begin{align*}
    & \frac{7n}{8} \leq |V \setminus C| \leq n \ \ \ \text{and} \ \ \
    \frac{1}{3} \leq \sum_{v_l \in V \setminus C} p_H(i \rightsquigarrow l) \leq \frac{2}{3} \ \ \text{for all $v_i \in U$.}
\end{align*}
Now it is clear that for any $v_i, v_j \in U$, we have 
\[
\frac{1}{12n} = \frac{\frac{1}{3}}{n} - \frac{1}{4n}
\leq  p_H(i \rightsquigarrow j) \leq \frac{\frac{2}{3}}{\frac{7n}{8}} + \frac{1}{4n} = \frac{85}{84n},
\]
as required.
\end{proof}

We construct a graph $R$ on the vertex set $U$ by repeating $\Theta(n \log n)$ times the following procedure. Note that it is possible that $R$ has multi-edges, and we do not coalesce  these multi-edges into single edges.

\begin{enumerate}
    \item Choose a vertex $v_i \in U$ uniformly at random.
    \item Do an $F_H$-random walk starting at $v_i$. Let $v_j$ be the end vertex of the random walk. 
    \item If $v_j \in U$, then add the edge $\{v_i, v_j\}$ to $R$.
\end{enumerate}

This construction gives an embedding of $R$ into $U$ in the underlying graph $G$. We can bound the congestion and dilation of the embedding as follows.

\begin{lemma}\label{lem-property-random-graph-0}
The embedding of $R$ into $U$ has congestion $c' = O(ck \log n)$ and dilation $d' = O(dk)$ with high probability.
\end{lemma}
\begin{proof}
The dilation upper bound $d' = O(dk)$ follows from the fact that $H = (M_1, M_2, \ldots,  M_k)$, and the given embedding of $M_i$ into $V$ has dilation $d$.

To bound the congestion $c'$, we focus on a specific edge $e \in E$ and a specific matching $M_i$ in $H$. Let $M^\ast \subseteq M_i$ be the subset of $M_i$ whose embedding involve the edge $e$. By the assumption given in \cref{lem-expander-from-random-walk}, we know that $|M^\ast| \leq c$. 
Let $S^\ast \subseteq V$ be the set of vertices incident to $M^\ast$, and we have $|S^\ast| 
\leq 2|M^\ast| = 2c$.

Consider the matrix $F = M_{i-1} \cdot M_{i-2} \cdots M_1$, and a vector $u \in \mathbb{R}^n$ defined by $u[j] = 1/|U|$ if $v_j \in U$, and $v_j = 0$ otherwise. 
Let $u' = F \cdot u$. It is clear that $u'[j]$ is the probability that a $F_H$-random walk starting at a uniformly random vertex in $U$ is at $v_j$ right after the $(i-1)$th transition.
Therefore, the probability that the embedding of a $F_H$-random walk starting at a uniformly random vertex in $U$ uses the edge $e$ to embed the $i$th transition is at most %\thatchaphol{I added ``at most'' and change the sum to sum over all $v_j \in S^*$ instead of $S$.}
\[
\frac{1}{2} \cdot \sum_{v_j \in S^*} u'[j] \leq 
\frac{1}{2} \cdot  |S^\ast| \cdot \max_{1 \leq j \leq n} u'[j] 
\leq \frac{1}{2} \cdot  |S^\ast| \cdot \frac{1}{|U|}
\leq \frac{c}{4n}.
\]
Here we use the fact that $|U| \geq |V|/2$ in \cref{lem-property-good-v} and the fact that $\max_{1 \leq j \leq n} u'[j]  \leq \max_{1 \leq j \leq n} u[j] = 1/|U|$, as $u' = F \cdot u$ for a doubly stochastic matrix F.

Suppose the total number of $F_H$-random walks we initiate in the construction of $R$ is $K n \log n$. Let $X$ be the number of $F_H$-random walks using the edge $e$ to embed their $i$th transition. Then $\mu = \Expect[X] \leq \frac{c}{4n} \cdot K n \log n = (cK/4) \log n$. By a Chernoff bound, we have $\Prob[X > (cK/3) \log n] = n^{-\Omega(cK)}$. That is, with high probability, the edge $e$ is used at most  $(cK/3) \log n$ times to embed the $i$th transition of the $K n \log n$ $F_H$-random walks. Therefore, the congestion of the overall embedding of $R$ is at most $k \cdot K (cK/3) \log n = O(ck \log n)$ with high probability.
\end{proof}

\begin{lemma}\label{lem-property-random-graph}
The graph $R$ has sparsity $\Phiv(R) = \Omega(\log n)$ and maximum degree $\Delta_R = O(\log n)$ with high probability.
\end{lemma}
\begin{proof}
Suppose the total number of $F_H$-random walks we initiate in the construction of $R$ is $K n \log n$.
By \cref{lem-property-good-v}, $p_H(i \rightsquigarrow j) \leq (85/84)(1/n)$ for each $v_i, v_j \in U$, and so the expected degree of each vertex $v_i \in U$ in $R$ is at most $(K n \log n) \cdot (85/84)(1/n) = (85K/84) \log n$. By a Chernoff bound, it can be shown that with probability at least $1 - n^{-\Omega(K)}$, the degree of $v_i$ is at most $2K \log n$. Therefore,  $\Delta_R = O(\log n)$ with high probability.

Let $1 \leq s \leq |U|/2$. Fix any $s$-vertex subset $S$ of $U$. 
By \cref{lem-property-good-v}, $p_H(i \rightsquigarrow j) \geq (1/12)(1/n)$ for each $v_i, v_j \in U$, and so
the expected number of edges in $R$ connecting $S$ and $U \setminus S$ is at least 
\[(K n \log n) 
\cdot |S| \cdot |U \setminus S| 
\cdot \frac{1}{|U|} 
\cdot \frac{1}{12n} 
= \frac{K \log n}{12} \cdot \frac{|U \setminus S| |S|}{|U|} 
\geq \frac{K \log n}{12} \cdot \frac{(|V|/2) |S|}{|V|} 
= \frac{K \log n}{24} \cdot |S|.
\] 
By a Chernoff bound, the number of edges in $R$ connecting $S$ and $U \setminus S$ is at least $(K/24) |S| \log n$ with probability $1 - n^{-\Omega(Ks)}$.
Now we take a union bound over all size-$s$ subsets $S \subseteq U$, and a union bound over all $1 \leq s \leq |U|/2$, the probability that there exists a set $S^\ast \subseteq U$ with $0 < |S^\ast| \leq |U|/2$ and $|\partial_{R}(S^\ast)| < (K/30) |S^\ast| \log n$ is at most
\[
\sum_{1 \leq s \leq |U|/2} %{|U| \choose s}
{\genfrac(){0pt}{1}{|U|}{s}}
n^{-\Omega(Ks)} 
\leq \sum_{1 \leq s \leq |U|/2} n^s \cdot n^{-\Omega(Ks)}.
= \sum_{1 \leq s \leq |U|/2} n^{-\Omega(Ks)}
= n^{-\Omega(K)}.
\]
Therefore, with high probability, we have $\Phiv(R) \geq (K/30) \log n = \Omega(\log n)$.
\end{proof}

 We define $W$ as the set of all vertices involved in the embedding of $R$.
 By \cref{lem-expander-emb} with parameters $c' = ck \log n$, $d' = dk$, $\Delta_R = O(\log n)$ and $\psi' = \Omega(\log n)$ (in view of \cref{lem-property-random-graph-0,lem-property-random-graph}), we obtain the following lemma.

\begin{lemma}\label{lem-conductance-rand-embed}
The subgraph $G[W]$ has  $\Phiv(G[W]) = \Omega(c^{-1} d^{-1} k^{-2} \log^{-1} n)$ with high probability.
\end{lemma}

What remains to do is to bound the   round complexity of finding the subset $W$. The algorithm for constructing $W$ consists of  simulating $O(\log n)$ $F_H$-random walks of $F_H$ from each $v \in V \setminus C$. The random walk simulation costs $O(c' \log n + d') = O(ck \log^2 n + dk)$ rounds with high probability using the routing algorithm of~\cite{ghaffari2015near,leighton1994packet}. %\thatchaphol{Refer to Lemma. Why this is not $\tilde O(ck + dk)$}
In the construction of $R$, we need to let each vertex $v \in U$ to know how many $F_H$-random walks it needs to initiate. This can be done in $O(D)$ rounds using \cref{lem-uniform-sample}. Hence the overall round complexity for \cref{lem-expander-from-random-walk} is $O(D + ck \log^2 n + dk)$.

\section{Deterministic Sparse Cut Computation}\label{sect-det-expander-decomp}

We consider a variant of the sparse cut problem that allows a small set of leftover vertices. That is, whenever $|W| \geq (1-\betaleft)|V|$ for a given threshold $\betaleft$, we do not need to find any sparse cut.

\begin{definition}[Balanced sparse cut with leftover]\label{def-det-sparse-cut}
Let $G=(V,E)$ be a graph of maximum degree $\Delta$. Let $0 < \phicut < 1$ and $0 < \phiemb < 1$ be any parameter. The task \[\detbalspcut{\phicut}{\phiemb}{\betacut}{\betaleft}\] asks for two subsets  $W \subseteq V$ and $C \subseteq V$  meeting the following requirements.
\begin{description}
    \item[Expander:] The induced subgraph $G[W]$ has  $\Phiv(G[W]) \geq \phiemb$.
    \item[Cut:] The cut $C$ satisfies  $0 \leq |C| \leq |V|/2$ and $\Phiv(C) \leq \phicut$.
    \item[Balance:] Either one of the following  is met.
    \begin{itemize}
        \item $|C| \geq \betacut|V|$ and $W = \emptyset$.
        \item $|V \setminus (C \cup W)| \leq \betaleft|V|$.
    \end{itemize}
\end{description}
\end{definition}

We write $\Tcut(n, \Delta, D,\phicut, \phiemb,  \betacut, \betaleft)$ to denote the deterministic round complexity for solving the above task. We note again that $D$ is the diameter of a Steiner tree that spans $G$, not the diameter of $G$.
%\thatchaphol{Note again that $D$ is the diameter of a Steiner tree that spans $G$.}

The goal of this section is to prove \cref{lem-det-cutmatch-main}. Similar to its analogous result \cref{lem-rand-cutmatch-main-bounded-deg} in the randomized setting, in \cref{lem-det-cutmatch-main} the set $W$ \emph{certifies} that $C$ is \emph{approximately nearly most balanced}, but only for the case $|C| = \Omega(\betaleft) \cdot |V|$. 

\begin{theorem}[Deterministic sparse cut computation]\label{lem-det-cutmatch-main}
Let $G=(V,E)$ be a bounded-degree graph, and let $0 < \psi < 1$ be any parameter.
For  $\betacut = 1/3$, % $\betaleft = \poly(\psi) 2^{-\Omega(\sqrt{\log n \log \log n})}$,
$\phicut = \psi$, and $\phiemb = \poly(\psi) 2^{-O(\sqrt{\log n \log \log n})}$,
The task $\detbalspcut{\phicut}{\phiemb}{\betacut}{\betaleft}$ can be solved deterministically in 
$O(D) + \poly(\psi^{-1}, \betaleft^{-1}) \cdot 2^{O(\sqrt{\log n \log \log n})}$ rounds.
\end{theorem}

Note that we can have the round complexity $O(D) + \poly(\psi^{-1}) \cdot 2^{O(\sqrt{\log n \log \log n})}$ with $\betaleft$ being as small as $\poly(\psi) 2^{-\Omega(\sqrt{\log n \log \log n})}$ in \cref{lem-det-cutmatch-main}.

\cref{lem-det-cutmatch-main} is proved using \cref{thm-solve-recursion}. There is a tradeoff between $\phiemb$ and the round complexity in \cref{thm-solve-recursion}. Specifically, for any $1 \geq \epsilon \geq \sqrt{\log \log n / \log n}$, there is a deterministic algorithm with round complexity $O(D) + \poly(\psi^{-1}, \betaleft^{-1}) \cdot n^{O(\epsilon)}$ with  $\phiemb = \poly(\psi) \cdot \log^{-O(1/\epsilon)}n$.

%The task of the sparse cut problem in \cref{lem-det-cutmatch-main} is parameterized by a balance parameter $\beta$ and a function $f$. We write $\balspcut{f}{\beta}$ to denote this task, and write $\Tdet(n, D, \Delta, \phi, f, \beta)$ and $\Trand(n, D, \Delta, \phi, f, \beta)$ to denote its deterministic and its randomized round complexities. We assume that $0 < \beta \leq 1/3$ always, and the function $f$ must satisfy that $f_n(\psi) \leq f_{n'}(\psi)$ whenever $n \geq n'$. That is, $f_n(\psi)$ is non-increasing in the parameter $n$.

\subsection{KKOV Cut-matching Game}\label{sect-KKOV-single}

In this section, we consider the   cut-matching game of Khandekar,  Khot,  Orecchia and  Vishnoi~\cite{khandekar2007cut} with some slight modifications.  The algorithm $\mathcal{A}$ used in the cut-matching game can be implemented using $\detbalspcut{\phicut}{\phiemb}{\betacut}{\betaleft}$ with $\phicut = 1/2$, $\betaleft = 1/12$ and $\betacut = 1/3$.
In particular, with these choices of parameters, the output $(C,W)$ of $\mathcal{A}$ satisfies either $|C| \geq |V|/4$ or $|W| \geq (2/3)|V|$. 

There is no underlying graph $G=(V,E)$ in \cref{def-kkov}. This definition simply describes the rules for the cut player and the matching player to construct a graph $H^\ast$ on the vertex set $V$. That is, the cut player always uses the output $C$ of $\mathcal{A}$ to produce an instance $(S^i, T^i)$ of a bipartite matching problem, and then the matching player is asked to find a matching that saturates at least half of $S^i$. The construction terminates when  $\mathcal{A}$ outputs $W$. 

\begin{definition}[KKOV cut-matching game]\label{def-kkov}
Let $\mathcal{A}$ be an algorithm that, given an input graph $G'=(V',E')$, returns either one of the following.
\begin{itemize}
    \item A subset $C' \subseteq V'$ with $\Phiv(C') \leq 1/2$ and $(1/4)|V'| \leq |C'| \leq (1/2)|V'|$.
    \item A subset $W' \subseteq V'$ with $\Phiv(G[W']) \geq \phiemb$ and $|W'| \geq (2/3)|V'|$.
\end{itemize}
Given a vertex set $V=\{v_1, v_2, \ldots, v_n\}$, a KKOV cut-matching game consists of a sequence of pairs of vertex subsets $(S^{1},T^{1}), (S^{2},T^{2}), \ldots (S^{\tau},T^{\tau})$ a sequence of matchings $M^1, M^2, \ldots, M^{\tau-1}$ and a graph $H^\ast$  are constructed by the following iterative steps for $i = 1, 2, \ldots$ until the termination condition is met.
\begin{description}
    \item[Graphs:] Initially, $H_0=(V, \emptyset)$ is the graph on the vertex set $V$ with zero edges. For each $j \geq 1$, define $H^j=(V, M^1 \cup M^2 \cup \cdots \cup M^j)$ as the graph on the vertex set $V$ that is the union of $M^1, M^2, \ldots, M^{j}$. 
    \item[Cut player:] For each iteration $i$, % $1 \leq i \leq \tau$, 
    the cut player applies the algorithm $\mathcal{A}$ on $H_{i-1}$. There are two possibilities.
    \begin{itemize}
        \item  If the output cut $C^i \subset V$ has $\Phiv_{H_{i-1}}(C^i) \leq 1/2$ and $|C^i| \geq (1/4)|V|$, then $S^i = C^i$ and $T^i = V \setminus C^i$. %In this case, $1 \leq i \leq \tau-1$.
        \item  Otherwise, the output subgraph $H^{i-1}[W^i]$ has $\Phiv(H^{i-1}[W^i]) \geq \phiemb$ and $|W^i| \geq (2/3)|V|$. Set  $H^\ast = H^{i-1}[W^i]$, set $\tau = i$, and the construction is terminated.
    \end{itemize}
    \item[Matching player:] For each iteration $i$, after the cut player computes $S^i$ and $T^i$, the matching player finds an arbitrary $M^i$   between $S^i$ and $T^i$ with size $|M^i| \geq |S^i|/2$. 
\end{description}
\end{definition}

%Note that $M^i$ is defined for $1 \leq i \leq \tau-1$, and $H^i$ is defined for  $0 \leq i \leq \tau-1$.

%(1) we  use $\cutemb$ to implement the cut player, and (2) we do not require the matching player to match all of the source vertices, and a $1/2$ fraction is good, (3) the balance parameter $\beta$ is now a variable instead of a fixed constant.
We summarize some basic properties of the KKOV cut-matching game that are direct consequences of \cref{def-kkov}.
\begin{itemize}
\item  The flow instance $(S^i, T^i)$ for the $i$th iteration satisfies $$(1/4)|V| \leq |S^i| \leq (1/2)|V|$$ and $$T^i = V \setminus S^i.$$  
\item The graph $H^i$ for the $i$th iteration has maximum degree $\Delta(H^i) \leq i \leq \tau-1$, because it is a union of $i$ matchings.
\item The final graph $H^\ast = H_{\tau-1}[W^{\tau}]$ is an induced subgraph of $H^{\tau-1}$ on the vertex set $W^{\tau}$ of size  $$|V(H^\ast)| = |W^{\tau}| \geq (2/3)|V|.$$ The graph $H^\ast$ has maximum degree   $\Delta(H^\ast) \leq \tau-1$, and it has sparsity $$\Phiv[H^\ast] \geq \phiemb.$$
\end{itemize}

%$0 < \beta \leq 1/3$ and $0 < \psi < \delta \beta \log^{-1} n$ for some universal constant $\delta$. 

%The cut-matching game works as follows.
%Initially, there are $n$ isolated vertices $V$. In each iteration $i$, the cut player computes a bipartition $(S_i,T_i)$ of $V$, with $|S_i| \leq |T_i|$. The matching player outputs an arbitrary matching $M_i$ between $S_i$ and $T_i$ such that at least half of the vertices in $S_i$ are matched. Denote $H_i = M_1 \cup M_2 \cup \cdots \cup M_i$, and $H_0$ is the graph with $n$ vertices but no edge.

%In the KKOV cut-matching game, the output $(S_i,T_i)$ of the cut player is computed by finding any cut $C = S_i$ of $H_{i-1}$ satisfying the above conductance and balance requirement, and set $T_i = V \setminus C$.  

%Intuitively, such a cut $C$ will be computed by applying $\cutemb(\phicut,\phiemb,\betabal)$ on $H_{i-1}$ with $\betabal = \beta$ and $\phicut = \phi$. 

%The main goal of \cref{sect-KKOV-single} is to

We will show that any KKOV cut-matching game must have
$\tau = O( \log n)$, and so the final graph $H^\ast$ has maximum degree  $\Delta(H^\ast) = O(\log n)$ and sparsity $\Phiv(H^\ast) =  \phiemb$.

%continue for at most $t = O(\beta^{-1} \log n)$ iterations.
%That is, there must be some $t = O(\beta^{-1} \log n)$ such that the cut player is unable to find the required cut $C$ on $G = H_t$.
%Specifically, if the algorithm of the cut player is  $\cutemb(\phicut,\phiemb,\betabal)$, then the cut player must find a subset $W \subseteq V$ such that $G\{W\}$ has conductance $\Phiv(G\{W\}) \geq \phiemb$, and the $|V \setminus W| \leq 2\betabal |V|$. That is, we obtain a large high-conductance subgraph $G\{W\}$ of $G$ with maximum degree $t = O(\beta^{-1} \log n)$.

\paragraph{Potential function.}  Similar to the cut-matching game in the randomized setting, the sequence of matchings $(M^1, M^2, \ldots, M^{i})$ gives rise to a doubly stochastic matrix 
\[F = F_{M^{i}} \cdot F_{M^{i}} \cdots F_{M^{1}},\] and we write \[p(i \rightsquigarrow j) = F[j,i]\] as the transition probability from $v_i$ to $v_j$ after applying an $F$-random walk. Recall that the random walk works as follows. From $j = 1, 2, \ldots, i$. In iteration $j$, if you are at a vertex $u$ that belongs to an edge $\{u,v\}\in M_j$, then there is a $1/2$ probability that you move to $v$. As in~\cite{khandekar2007cut},  define the potential function 
\[\Pi(i) = \sum_{1\leq j \leq n, \, 1 \leq l \leq n} p(j \rightsquigarrow l) \log \frac{1}{p(j \rightsquigarrow l)}.\]
This potential function measures the entropy of the current $F$-random walk, and it is different from the potential function used in~\cite{KRV09,RST14,SaranurakW19}.
It is clear that $\Pi(t) \leq n \log n$, and $\Pi(i)$ is maximized when $p(j \rightsquigarrow l) = 1/n$ for all $j, l$. We  have $\Pi(0)= 0$ initially. Recall that when there is no matching, $F$ is the identity matrix, and we have $p(j \rightsquigarrow l) = 1$ if $j=l$, and $p(j \rightsquigarrow l) = 0$ otherwise.

\begin{lemma}[Potential increase]\label{lem-potential-reduction}
We have $\Pi(i) - \Pi(i-1) = \Omega(n)$ for each iteration $1 \leq i \leq \tau-1$.
\end{lemma}

The proof of \cref{lem-potential-reduction} is left to \cref{sect-det-potential}.
The following lemma is an immediate consequence of \cref{lem-potential-reduction}.

\begin{lemma}[Number of iterations]\label{lem-iteraions-KKOV}
The KKOV cut-matching game has  $\tau = O(\log n)$ iterations, and so the final graph $H^\ast$ has maximum degree  $\Delta(H^\ast) = O(\log n)$.
\end{lemma}

\subsection{Simultaneous Embedding of Multiple Expanders}\label{sect-simult-emb-expander}

We consider the following task $\cutmatch{\phicut}{\phiemb}$.

\begin{definition}[Simultaneous embedding of multiple expanders]
    Given an input graph $G=(V,E)$ with maximum degree $\Delta$, together with a partition $\VV= \{V_1, V_2, \ldots, V_k\}$ of $V$ with $|V_i| = \Theta(n/k)$ for each $1 \leq i \leq k$, the task $\cutmatch{\phicut}{\phiemb}$ partitions $\VV$ into $\VV = \VVcut \cup \VVemb$ and outputs the following sparse cut and expander embedding. 
\begin{description}
\item[Embedding:]  A simultaneous embedding of a graph $H_i$  into $U_i \subseteq V_i$ for each $V_i \in \VVemb$.
The simultaneous embedding has congestion $c = \poly(\Delta \phicut^{-1}, \log n)$ and dilation $d = \poly(\Delta \phicut^{-1}, \log n)$. 
Furthermore, the following holds for each $V_i \in \VVemb$.
\begin{itemize}
    \item $\Delta(H_i) = O(\log n)$.
    \item $\Phiv(H_i) \geq \phiemb$.
    \item $|U_i| \geq (2/3) |V_i|$.
\end{itemize}
\item[Cut:]  A cut $C \subseteq V$ with $\Phiv(C) \leq \phicut$ and $0 \leq |C| \leq |V|/2$ satisfying either one of the following.
\begin{itemize}
    \item $|C| \geq |V|/3$.
    \item $C$ contains at least $(1/8)|V_i|$ vertices from each $V_i \in \VVcut$.
\end{itemize}
\end{description}
\end{definition} 

%In view of the discussion in \cref{sect-KKOV-single}, in order to solve $\cutmatch$, it suffices to run $O(\log n)$ iterations of the KKOV cut-matching game (with constant $\beta$). 

The goal of \cref{sect-simult-emb-expander}  is to prove \cref{lem-simul-emb-det}.
The proof of  \cref{lem-simul-emb-det} relies on the following auxiliary lemma, which shows that we can make the dependence on the Steiner tree diameter $D$ additive. This is crucial since the $k$ parallel recursive calls in the simultaneous expander embeddings share the same Steiner tree in the underlying graph. The proof of \cref{lem-diam-reduction-main} is left to \cref{sect-diam-reduce}.

\begin{lemma}[Diameter reduction]
\label{lem-diam-reduction-main}
There is an $O\left(D + \psi^{-2} \Delta^{2} \log^6 n\right)$-round deterministic algorithm $\mathcal{A}_1$ and an $O(D)$-round deterministic algorithm $\mathcal{A}_2$ that allow us to solve \[\detbalspcut{\phicut}{\phiemb}{\betacut}{\betaleft}\] as follows.
\begin{enumerate}
    \item Run the algorithm $\mathcal{A}_1$, which produces a subgraph with a Steiner tree of diameter $D' = O(\psi^{-1}  \Delta \log^3 n)$. 
    \item Solve the task $\detbalspcut{\phicut/2}{\phiemb}{\betacut}{\betaleft}$ on this subgraph.
    \item Run the algorithm $\mathcal{A}_2$, which gives us a solution to $\detbalspcut{\phicut}{\phiemb}{\betacut}{\betaleft}$  by combining the output results in previous steps.
\end{enumerate}
\end{lemma}

%\begin{align*}
%   &\Tcut\left(n, \Delta, D, \phicut, \phiemb, \betacut, \betaleft \right)\\ 
%    &\leq O\left(D + \phicut^{-2}\Delta^2 \log^6 n\right)
%    + \Tcut\left(n, \Delta, O\left(\psi^{-1}\Delta \log^3 n\right),\frac{\phicut}{2}, \phiemb,  \betacut, \betaleft\right)
%\end{align*}

In \cref{lem-simul-emb-det}, recall that $\Tcut(n', \Delta', D', \phicut', \phiemb',  \betacut', \betaleft')$ is the deterministic round complexity for solving the problem $\detbalspcut{\phicut'}{\phiemb'}{\betacut'}{\betaleft'}$  defined in \Cref{def-det-sparse-cut}.

\begin{lemma}[Deterministic embedding of multiple expanders] \label{lem-simul-emb-det} Consider a  graph $G=(V,E)$ with maximum degree $\Delta$ and a parameter $0 < \phicut < \Delta$. Given a partition $\VV= \{V_1, V_2, \ldots, V_k\}$ of $V$ with $|V_i| = \Theta(n/k)$ for each $1 \leq i \leq k$, the task $\cutmatch{\phicut}{\phiemb}$  can be solved deterministically with round complexity
\[
\poly\left(D, k, \Delta \phicut^{-1}, \log n\right) +   \poly\left(\Delta \phicut^{-1}, \log n\right) \cdot \Tcut\left(n', \Delta', D', \phicut', \phiemb', \betacut', \betaleft' \right),
\]
where
\begin{align*} 
 n' &= O\left(\frac{n}{k}\right), & \Delta' &= O(\log n),  
 &  D' &= O(\log^4 n),  \\ 
 \phicut' &= \frac{1}{4}, &\phiemb' &= \phiemb, 
 & \betacut' &= \frac{1}{3}, 
 & \betaleft' &= \frac{1}{12}. 
\end{align*}
\end{lemma}
\begin{proof}
The algorithm implements the  cut-matching game on each part of $\VV= \{V_1, V_2, \ldots, V_k\}$ simultaneously by using the flow algorithm of \cref{lem-cut-match-det-multi} to implement the matching player and using an algorithm for $\detbalspcut{\phicut'}{\phiemb'}{\betacut'}{\betaleft'}$ to implement the cut player. Throughout the process, a cut is maintained. Whenever the flow algorithm of \cref{lem-cut-match-det-multi} fails to find a large enough matching for some part $V_i$ to fulfill the task of the matching player in the cut-matching game, then \cref{lem-cut-match-det-multi} guarantees that we can either enlarge the current cut to include at least $1/8$ fraction of vertices of $V_i$ or obtain a balanced sparse cut.

We set up the notations for the algorithm. The algorithm has $\tau = O(\log n)$ iterations. At the beginning of each iteration $i$, we maintain a partition of $\VV$ into three sets.
\begin{description}
    \item[Active parts:] $\VV_{i-1}^{\mathsf{active}} \subseteq \VV$ is the set of active parts where the cut-matching game has not terminated yet.
    \item[Finished parts:] $\VV_{i-1}^{\mathsf{emb}} \subseteq \VV$ is the set of  parts where the cut-matching game has terminated.
    \item[Failed parts] $\VV_{i-1}^{\mathsf{cut}} \subseteq \VV$ is the set of  parts where the cut-matching game cannot continue because at some point the matching player failed to find a large enough matching. 
\end{description}
We also maintain a cut $C_{i-1} \subseteq V$. The induction hypothesis guarantees the following.
\begin{itemize}
    \item $|C_{i-1}| < |V|/3$, i.e., it is not balanced yet.
    \item $C_{i-1}$ contains at least $(1/8)|V_j|$ vertices from each $V_j   \in \VV_{i-1}^{\mathsf{cut}}$, i.e., it includes enough vertices from the parts that we cannot embed an expander into.
    \item $\Phiv(C_{i-1}) \leq \phicut/2$, i.e., the sparsity requirement is met.
\end{itemize}

For each active part $V_j \in \VV_{i-1}^{\mathsf{active}}$, it is guaranteed that we have already finished the first $i-1$ steps of the cut-matching game, but it is not done yet.
We write $M_j^r$ to denote the $r$th matching of the cut-matching game, for each $V_j \in \VV_{i-1}^{\mathsf{active}}$, for each $1 \leq r \leq i-1$.
The induction hypothesis guarantees that for each $1 \leq r \leq i-1$, we have a simultaneous embedding of $\bigcup_{V_j \in \VV_{i-1}}^{\mathsf{active}} M_j^r$ with congestion $c = O(\Delta^2 \phicut^{-2} \log n \log k)$ and dilation $d = O(\Delta  \phicut^{-1} \log n)$.

Suppose that all the above requirements are met at the beginning of iteration $i$. We design the algorithm for iteration $i$ as follows.

\paragraph{Cut player.}
We implement the  cut player by solving $\detbalspcut{\widehat{\phicut}}{\widehat{\phiemb}}{\widehat{\betacut}}{\widehat{\betaleft}}$ with 
$\widehat{\phicut} = 1/2$, $\widehat{\phiemb} = \phiemb$, $\widehat{\betacut} = 1/3$, and $\widehat{\betaleft} = 1/12$
on the graph $H_j^{i-1} = M_j^1 \cup M_j^2 \cup \cdots \cup M_j^{i-1}$, in parallel for all $V_j \in \VV_{i-1}^{\mathsf{active}}$, using \cref{lem-diam-reduction-main}. The number of vertices and the maximum degree for each part are $n' = O(n/k)$ and $\Delta' = i-1 = O(\log n)$.

The parallel simulation suffers from congestion $(i-1)c$ and dilation $d$. Moreover, for the parts $\mathcal{A}_1$ and $\mathcal{A}_2$ in \cref{lem-diam-reduction-main}, we have to share the underlying Steiner tree $T$, which causes a congestion of $k$. Taking into these overheads into consideration, the total round complexity for the parts $\mathcal{A}_1$ and $\mathcal{A}_2$ in \cref{lem-diam-reduction-main} is
\[
c(i-1) \cdot d \cdot k \cdot O\left(D + {\widehat{\phicut}}^{-2} {\Delta'}^{2} \log^6 n\right)
= \poly\left(D, k, \Delta \phicut^{-1}, \log n\right).
\]
%, and so the effective $D$-parameter for each $H_j^{i-1}$ is $|\VV_{i-1}^{\mathsf{active}}| \cdot D \leq kD$. \thatchaphol{Refer to some lemma. Explain the factor $kD$.}

Because $\phicut' = \widehat{\phicut}/2$, $\phiemb' = \widehat{\phiemb}$, $\betacut' = \widehat{\betacut}$, and $\betaleft' = \widehat{\betaleft}$, the part of solving $\detbalspcut{\widehat{\phicut}/2}{\widehat{\phiemb}}{\widehat{\betacut}}{\widehat{\betaleft}}$ with $D' = O\left({\widehat{\phicut}}^{-1} \Delta' \log^3 n\right) = O(\log^4 n)$ in \cref{lem-diam-reduction-main} has round complexity
\begin{align*}
   &c(i-1) \cdot d \cdot \Tcut\left(n', \Delta', D', \phicut', \phiemb', \betacut', \betaleft' \right)\\
   &=\poly(\Delta \phicut^{-1}, \log n) \cdot \Tcut\left(n', \Delta', D', \phicut', \phiemb', \betacut', \betaleft' \right).
\end{align*} 

%To summarize, the parameters for the algorithm of the cut player are $n' = O(n/k)$, $\Delta' = O(\log n)$, $D' = kD$, $\phicut' = 1/2$, $\phiemb' = \phiemb$, $\betacut' = 1/3$, and $\betaleft' = 1/12$. Overall, the round complexity of one iteration of the cut player is
%\begin{align*}
%   &c(i-1) \cdot d \cdot \Tcut\left(n', \Delta', D', \phicut', \phiemb', \betacut', \betaleft' \right)\\
%   &=\poly(\Delta \phicut^{-1}, \log n) \cdot \Tcut\left(n', \Delta', D', \phicut', \phiemb', \betacut', \betaleft' \right).
%\end{align*} 
After that, the cut player either return    $C_j^{i}$ or return $W_j^{i}$ for each $V_j \in \VV_{i-1}^{\mathsf{active}}$. If $W_j^{i}$ is returned, we add $V_j$ to $\VV_{i}^{\mathsf{emb}}$, as the cut-matching game on this part has finished.

\paragraph{Matching player.} We let $\VV^\ast = \VV_{i-1}^{\mathsf{active}} \setminus \VV_{i}^{\mathsf{emb}}$ be the set of active parts where the cut player returns $C_j^{i}$. As in \cref{def-kkov}, set  $S_j^{i} = C_j^{i}$ and $T_j^{i} = V_j \setminus S_j^{i}$. Now the task of the matching player is to find a matching $M_j^i$ of size at least $|S_j^{i}|/2$ between $S_j^{i}$ and $T_j^{i}$ to fulfill the requirement in \cref{def-kkov}, for each $V_j \in \VV^\ast$. Furthermore, we want to have a simultaneous  embedding of the union of $M_j^i$ over all $V_j \in \VV^\ast$ with congestion $c$ and dilation $d$. 

We apply the algorithm of \cref{lem-cut-match-det-multi} with parameters $\psi = \phicut$, $\beta = \Omega(1/k)$ and $\Cin = C_{i-1}$.
The algorithm has round complexity
\[
O\left(Dk \Delta^{2} \psi^{-2} \log n \log \frac{1}{\beta} + k \Delta^6 \psi^{-6} \beta^{-1} \log^2\Delta \log^2 n \log  \frac{1}{\beta} \right) =  \poly(D, k, \Delta \phicut^{-1}, \log n).
\]

We set $C_i = \Cout$ resulting from the algorithm of \cref{lem-cut-match-det-multi}. One possibility is that we already have $|V|/3 \leq |C_i| \leq |V|/2$ and $\Phiv(C_i) \leq \phicut$, and we are done in this case by outputting $\VVemb = \emptyset$, $\VVcut = \VV$, and $C = C_i$.

The other possibility is that $\Phiv(C_i) \leq \phicut/2$, $|C_i| < |V|/3$, $C_{i-1} \subseteq C_i$, and for each $V_j \in \VV^\ast$, if less than $|S_j^{i}|/2$ vertices in $S_j^{i}$ are matched by the matching player, then $C_i$ contains at least $|S_j^{i}|/2 \geq |V_j|/8$ vertices in  $S_j^{i} \cup T_j^{i} = V_j$. If this is the case, $V_j$ is added to $\VV_i^{\mathsf{cut}}$.
All parts in $\VV^\ast \setminus \VV_i^{\mathsf{cut}}$ are added to $\VV_{i}^{\mathsf{active}}$.

\paragraph{Summary.} 
After $\tau = O(\log n)$ iterations, we must have $\VV_{\tau}^{\mathsf{active}} = \emptyset$ (\cref{lem-iteraions-KKOV}). Now we can output $\VVcut = \VV_{\tau}^{\mathsf{cut}}$, $\VVemb = \VV_{\tau}^{\mathsf{emb}}$ and
  $C = C_\tau$.
By induction hypothesis, each part $V_j \in \VVcut$ must have at least $|V_j|/8$ of its vertices in $C$.
Also by  induction hypothesis, $\Phiv(C) \leq \phicut/2 < \phicut$ and $0 \leq |C| < |V|/3$.
Therefore, $C$ is a valid output cut.

For each $V_j \in \VVemb$, the induction hypothesis implies that they have successfully finished the entire cut-matching game, and therefore it has found an induced subgraph $$H^\ast = H_j^{\tau}[W_j^\tau] \subseteq H_j^{\tau} = M_j^1 \cup M_j^2 \cup \cdots \cup M_j^\tau,$$ and it has sparsity $\Phiv(H^\ast) \geq \phiemb$. We set $U_j = W_j^\tau$ and $H_j = H^\ast$, and from the specification of the cut-matching game, we have $|U_j| \geq (2/3)|V_j|$. Since there are $\tau = O(\log n)$ iterations, this  graph $H_j$ is a union of $O(\log n)$ matchings, and so it has maximum degree $O(\log n)$.

Since all matchings in one iteration can be embedded simultaneously with congestion $c$ and dilation $d$, the well-connected graphs $H_j$ for all $V_j \in \VVemb$ can be embedded simultaneously with congestion $c \tau = O(c \log n) = \poly(\Delta \phicut^{-1}, \log n)$ and dilation $d = \poly(\Delta \phicut^{-1}, \log n)$ simultaneously, as required.
\end{proof}

\subsection{Combining Well-connected Subgraphs}\label{sect-merge-subgraphs}

The goal of \cref{sect-merge-subgraphs} is to prove \cref{lem-merge-subgraphs}. Note that the output of \cref{lem-merge-subgraphs} solves the task
\begin{align*}
    &\detbalspcut{\phicut'}{\phiemb'}{\betacut'}{\betaleft'} \ \ \ \  \text{with parameters}\\
    & \  \phicut' = \phicut, \ \
\phiemb' = \frac{ \phiemb} {\poly(\Delta \phicut^{-1}, \log n)}, \ \
\betacut' = \frac{1}{100}, \ \
\betaleft' =  O(\phicut\Delta^{-1}k^{-1}).
\end{align*}

\begin{lemma}[Combining well-connected subgraphs]\label{lem-merge-subgraphs}
Consider a graph $G=(V,E)$  with a partition $\VV= \{V_1, V_2, \ldots, V_k\}$ of $V$ with $\floor{n/k}\leq |V_i| \leq \ceil{n/k}$ for each $1 \leq i \leq k$. Given a solution of the task $\cutmatch{\phicut}{\phiemb}$,
 there is a deterministic algorithm with round complexity
\[
\poly(D, k, \Delta \phicut^{-1}, \log n)
\]
that outputs $C^\ast \subseteq V$ and $W^\ast \subseteq V$ satisfying 
\[0 \leq |C^\ast| \leq \frac{|V|}{2} \ \ \ \text{and} \ \ \  \Phiv(C^\ast) \leq \phicut \ \ \ \text{and} \ \ \ \Phiv(G[W^\ast]) \geq \frac{ \phiemb} {\poly(\Delta \phicut^{-1}, \log n)}.\]
Furthermore, at least one of the following is met.
    \begin{itemize}
        \item $|C^\ast| \geq \frac{|V|}{100}$ and $W^\ast = \emptyset$
        \item $V = C^\ast \cup W^\ast$.
        \item $C^\ast = \emptyset$ and $|V \setminus W^\ast| = O(\Delta^{-1}k^{-1} \phicut) \cdot |V|$.
    \end{itemize}
\end{lemma}

%For the rest of \cref{sect-merge-subgraphs}, 
To prove  \cref{lem-merge-subgraphs}, recall the specification of the output of $\cutmatch{\phicut}{\phiemb}$. The partition $\VV= \{V_1, V_2, \ldots, V_k\}$ of $V$  is further partitioned into $\VV = \VVcut \cup \VVemb$, and we have the following simultaneous embedding and sparse cut.
\begin{description}
\item[Embedding:]  A simultaneous embedding of a graph $H_i$ into $U_i \subseteq V_i$ for each $V_i \in \VVemb$.
The simultaneous embedding has congestion $c = O(\poly(\Delta \phicut^{-1}, \log n))$ and dilation $d = O(\poly(\Delta \phicut^{-1}, \log n))$. 
Furthermore, the following holds for each $V_i \in \VVemb$.
\begin{itemize}
    \item $\Delta(H_i) = O(\log n)$.
    \item $\Phiv(H_i) \geq \phiemb$.
    \item $|U_i| \geq (2/3) |V_i|$.
\end{itemize}
\item[Cut:]  A cut $C \subseteq V$ with $\Phiv(C) \leq \phicut$ and $0 \leq |C| \leq |V|/2$ satisfying either one of the following.
\begin{itemize}
    \item $|C| \geq |V|/3$.
    \item $C$ contains at least $(1/8)|V_i|$ vertices from each $V_i \in \VVcut$.
\end{itemize}
\end{description}

\paragraph{The easy case.}
  If the cut $C$   already has $|C| \geq |V|/100$, then we are done solving \cref{lem-merge-subgraphs} already by outputting $C^\ast = C$ and $W^\ast = \emptyset$. 
  
  For the rest of the proof, we focus on the case $|C| < |V|/100$, and so we must have $|\VVcut| < k/10$, since \[n/100 > |C| \geq \frac{1}{8} \cdot \left\lfloor{\frac{n}{k}}\right\rfloor  \cdot |\VVcut|, \ \ \text{and it implies} \ \ 
  |\VVcut| < \left(\frac{8}{100} +o(1)\right)k < \frac{k}{10}.
  \]
  Therefore, we have $|\VVemb| > (9/10)k$. From now on, we re-order $\VV$ in such a way that $\VVemb = \{V_1, V_2, \ldots, V_{k'}\}$ consists of the first $k' = |\VVemb| > (9/10)k$ parts of $\VV$.  
Since $|U_i| \geq (2/3) |V_i|$ for each $1 \leq i \leq k'$, the set
$U_1 \cup U_2 \cup \cdots \cup U_{k'}$ has size \[
|U_1 \cup U_2 \cup \cdots \cup U_{k'}| \geq \frac{2}{3} \cdot \left\lfloor{\frac{n}{k}}\right\rfloor \cdot \frac{9k}{10} = \left(\frac{3}{5} - o(1)\right)n > \frac{11n}{20}.\] For the rest of the proof, we abandon the cut $C$, and we will try combine these expander embeddings of $H_1, H_2, \ldots, H_{k'}$ to $U_1, U_2, \ldots, U_{k'}$ together to obtain a well-connected subgraph $G[W^\ast]$. 

\paragraph{Expanding the expander embeddings.} 
Given that $|U_1 \cup U_2 \cup \cdots \cup U_{k'}| > (11/20)|V|$, \cref{lem-enlarge} shows that in $\poly(D, k, \Delta \phicut^{-1}, \log n)$ rounds, we can enlarge the current simultaneous expander embedding to $U_i^\ast \supseteq U_i$ and $H_i^\ast \supseteq H_i$ for each $1 \leq i \leq k'$.

\begin{lemma}\label{lem-enlarge}
Suppose $T = U_1 \cup U_2 \cdots \cup U_{k'}$ has $|T| > (11/20)|V|$.
In $\poly(D, k, \Delta \phicut^{-1}, \log n)$ rounds, we can find a simultaneous embedding of $H_i^\ast$ into $U_i^\ast$ and a cut $C_1$ satisfying the following conditions.
\begin{description}
\item[Embedding:]
The sets $U_1^\ast, U_2^\ast, \ldots, U_{k'}^\ast$ are disjoint, $U_i^\ast \supseteq U_i$ and $H_i^\ast \supseteq H_i$ for each $1 \leq i \leq k'$. The simultaneous embedding has congestion and dilation $\poly(\Delta \phicut^{-1}, \log n)$.  Furthermore, for each $1 \leq i \leq k'$, we have \[\Phiv(H_i^\ast) = \Omega(1) \cdot \Phiv(H_i) \ \ \ \text{and} \ \ \   \Delta(H_i^\ast) \leq \Delta(H_i) + 1   \ \ \  \text{and} \ \ \   |U_i^\ast| \leq 2 \cdot |U_i|.\]
\item[Cut:] The cut $C_1 \subseteq V$ has $0 \leq |C_1| \leq |V|/2$ and $\Phiv(C_1) \leq \phicut/4$, and it satisfies either one of the following.
\begin{itemize}
    \item $|C_1| \geq (1/10)|V|$.
    \item $C_1$ contains all vertices not in any of $U_1^\ast, U_2^\ast, \ldots, U_{k'}^\ast$.
    \item $C_1 = \emptyset$, and $B = V \setminus (U_1^\ast, U_2^\ast, \ldots, U_{k'}^\ast)$ has size $|B| < \frac{\phicut}{4 \Delta} \cdot \floor{n/k} = O(\Delta^{-1} k^{-1} \phicut ) \cdot |V|$.
\end{itemize}
\end{description}
\end{lemma}
\begin{proof}
Apply \cref{lem-cut-match-det} to $(S,T)$ with $S = V\setminus T$ to find a matching between $V\setminus T$ and $T$.
 with $\beta = O(\phicut\Delta^{-1}k^{-1})$ being sufficiently small, and $\psi = \phicut/4$, to embed a matching between $S$ and $T$. The embedding has the required congestion $O(\poly(\Delta \phicut^{-1}, \log n))$ and dilation $O(\poly(\Delta \phicut^{-1}, \log n))$. The round complexity is  $\poly(D, k, \Delta \phicut^{-1}, \log n)$.

\paragraph{Leftover.}
By the specification of \cref{lem-cut-match-det}, there are two possibilities about the unmatched vertices in $S$. One is that almost all of $S$ are matched except at most $\beta|V| = O(\phicut\Delta^{-1}k^{-1}) \cdot |V|$ of them. We set $B$ to be these vertices. Selecting $\beta$ to be small enough, we can make $|B| < \frac{\phicut}{4 \Delta} \cdot \floor{n/k} = O(\phicut\Delta^{-1}k^{-1}) \cdot |V|$.

\paragraph{Cut.}
The other possibility is that we found a cut $\tilde{C}$ with $\Phiv(C) \leq \psi = \phicut/4$ such that 
all the unmatched vertices in $S$ are included in $\tilde{C}$, and all the unmatched vertices in $T$ are in $V \setminus \tilde{C}$.
In particular, 
 we must have $|V \setminus \tilde{C}| \geq |T| - |S| > |V|/10$. 
 
  If we also have $|\tilde{C}| > |V|/10$, then we can simply set $C_1$ to be the one of $\tilde{C}$ and $V \setminus \tilde{C}$ that has the smaller size, and this is a valid output as $|V|/10 \leq |C_1| \leq |V|/2$ and $\Phiv(C_1) \leq \phicut/4$.
Otherwise, we have $|\tilde{C}| < |V|/10 < |V|/2$. As $\tilde{C}$ contains all unmatched vertices in $S$, and so choosing $C_1 = \tilde{C}$  is also  a valid output.

\paragraph{Embedding.}
For each $1 \leq i \leq k'$, set $U_i^\ast$ to be the set $U_i$ together with the vertices matched to $U_i$, and set $H_i^\ast$ to be $H_i$ together with the edges in the matching that are incident to $U_i$. It is clear that this only affects the sparsity by a constant factor, increases the degree of each vertex by at most 1, and the size of $U_i^\ast$ is at most twice the size of $U_i$.
\end{proof}

\paragraph{Combining the expander embeddings.}
Let $C_1$ and $B$ be the sets specified in the algorithm of \cref{lem-enlarge}. We assume $|C_1| < |V|/100$, since otherwise we are done already. 
Define \[G' = (V',E') = G[U_1^\ast \cup U_2^\ast \cup \cdots \cup U_{k'}^\ast]\]
as the subgraph induced by the vertex set $V' = U_1^\ast \cup U_2^\ast \cup \cdots \cup U_{k'}^\ast$.
We must have $|V'| \geq (99/100)|V|$, since otherwise  $|C_1| > |V|/100$.

We select a cut $C_2$ in $G'$ as any cut  respecting the decomposition $\UU^\ast = \{U_1^\ast, U_2^\ast, \ldots, U_{k'}^\ast\}$ with the maximum possible size among those cuts satisfying either one of the following condition.
\begin{itemize}
    \item $|C_2| \geq |V'|/3$ and   $\Phiv_{G'}(C_2) \leq \phicut/4$.
    \item $0 \leq |C_2| < |V'|/3$ and   $\Phiv_{G'}(C_2) \leq \phicut/8$.
\end{itemize}
 
Define $G''= (V'', E'')$ as the subgraph of $G'$ resulting from removing the parts of $\UU^\ast$ that are in $C_2$.

\begin{lemma}\label{lem-conductance-combined graph}
If $|C_2| < |V'|/3$, then $\vPhiOut{G''}{\UU^\ast} > \phicut/8$.
\end{lemma}
\begin{proof}
Suppose there is a cut $C_3$ in $G''$ respecting the decomposition $\UU^\ast$ with $\Phiv_{G''}(C_3) \leq \phicut/8$.
Without loss of generality, assume $|C_3| \leq |V''|/2 = |V' \setminus C_2| / 2$.
By \cref{lem-cut-combine}, we have the following two cases.
\begin{itemize}
    \item If $|C_2 \cup C_3| \leq |V'|/2$, then $C = C_2 \cup C_3$ satisfies $\Phiv_{G'}(C) \leq \phicut/8$ and $|C| \leq |V'|/2$. 
    \item If $|C_2 \cup C_3| > |V'|/2$, then $C = V' \setminus (C_2 \cup C_3)$ satisfies $\Phiv_{G'}(C) \leq \phicut/4$ and $|V'|/3 \leq |C| \leq |V'|/2$.
\end{itemize}
In any case, we obtain a $\UU^\ast$-respecting cut $C$ in $G'$  violating the maximality of our choice of $C_2$ in its definition, as $|V'|/2 \geq |C| > |C_2|$ and $C$ also meets the criterion for selecting $C_2$.
\end{proof}

\begin{lemma}\label{lem-cut-case-1}
If $C_1 \neq \emptyset$, then $C_3 = C_1 \cup C_2$ satisfies $\Phiv(C_3) \leq \phicut$, $0 <  |C_3| \leq (2/3)|V|$, and $C_3$ includes all vertices not in $V' = U_1^\ast \cup U_2^\ast \cup \cdots \cup U_{k'}^\ast$.
\end{lemma}
\begin{proof}
First of all, $|C_3| \leq |C_2| + |C_1| \leq |V'|/2 + |V|/100 \leq (1/2)|V| + (1/100)|V|< (2/3)|V|$. To calculate $\Phiv(C_3)$,
\begin{align*}
&\Phiv(C_3) = \min\left\{
\frac{|\partial(C_3)|}{|C_3|},
\frac{|\partial(C_3)|}{|V \setminus C_3|}
\right\}\\
&\leq \frac{2|\partial(C_3)|}{|C_3|} & |C_3| \leq (2/3)|V|\\
%&\leq \frac{2|\partial(C_1)| + 2|\partial(C_2)|}{|C_3|}\\
&\leq \frac{2|\partial(C_1)|}{|C_3|} + \frac{2|\partial(C_2)|}{|C_3|}\\
&\leq \frac{2|\partial(C_1)|}{|C_1|} + \frac{2|\partial(C_2)|}{|C_2|} \leq 2\cdot \frac{\phi}{4} + 2\cdot \frac{\phi}{4} = \phi.
\end{align*}
Finally, the fact that $C_3$ includes all vertices not in $V' = U_1^\ast \cup U_2^\ast \cup \cdots \cup U_{k'}^\ast$ is due to the requirements on $C_1$ in \cref{lem-enlarge} (note that we assume $|C_1| < |V|/100$).
\end{proof}

\begin{lemma}\label{lem-cut-case-2}
If $C_1 = \emptyset$ and $C_2 \neq \emptyset$, then $C_3 = B \cup C_2$ satisfies $\Phiv(C_3) \leq \phicut$, $0 < |C_3| \leq (2/3)|V|$, and $C_3$ includes all vertices not in $V' = U_1^\ast \cup U_2^\ast \cup \cdots \cup U_{k'}^\ast$.
\end{lemma} 
\begin{proof}
First of all, $|C_3| \leq  |C_2| + |B| \leq |V'|/2 +  o(1) \cdot |V| < (2/3)|V|$. To calculate $\Phiv(C_3)$,
\begin{align*}
&\Phiv(C_3) = \min\left\{
\frac{|\partial(C_3)|}{|C_3|},
\frac{|\partial(C_3)|}{|V \setminus C_3|}
\right\}\\
&\leq \frac{2|\partial(C_3)|}{|C_3|} & |C_3| \leq (2/3)|V|\\
%&\leq \frac{2|\partial(B)| + 2|\partial(C_2)|}{|C_3|}\\
&\leq \frac{2|\partial(B)|}{|C_3|} + \frac{2|\partial(C_2)|}{|C_3|}\\
&\leq \frac{2 \cdot \Delta \cdot \frac{\phicut}{4 \Delta} \cdot \left\lfloor \frac{n}{k}\right\rfloor }{|C_3|} + \frac{2|\partial(C_2)|}{|C_2|} \leq 2\cdot \frac{\phi}{4} + 2\cdot \frac{\phi}{4} = \phi. & |C_3| \geq |C_2| \geq \left\lfloor \frac{n}{k}\right\rfloor
\end{align*}
Finally, the fact that $C_3$ includes all vertices not in $V' = U_1^\ast \cup U_2^\ast \cup \cdots \cup U_{k'}^\ast$ is due to the requirements for the case $C_1 = \emptyset$ in \cref{lem-enlarge}.
\end{proof}

If $|C_2| < |V'|/3$, define $W^\ast$ as the set of all vertices involved in the embedding of $H_j^\ast$, for each $U_j^\ast \in \UU^\ast$ with $U_j^\ast \cap C_2 = \emptyset$. That is, we only consider the parts that are not covered in $C_2$.

\begin{lemma}\label{lem-property-final-w}
If $|C_2| < |V'|/3$, then
$\Phiv(G[W]) \geq  \phiemb / \poly(\Delta \phicut^{-1}, \log n)$.
\end{lemma}
\begin{proof}
Use \cref{lem-expander-emb-multi} with the following parameters.
\begin{itemize}
    \item $\UU$ consists of the parts  $U_j^\ast \in \UU^\ast$ with $U_j^\ast \cap C_2 = \emptyset$.
    \item The congestion $c$ and the dilation $d$ are $\poly(\Delta \phicut^{-1}, \log n)$.
    \item $\hat{\Delta} = O(\log n)$ is the upper bound for maximum degree of each graph $H_j^\ast$.
    \item $\phi_i = \Omega(\phiemb)$ is the lower bound for the sparsity of each graph $H_j^\ast$.
    \item $\phi_o = \phicut/8$ in view of \cref{lem-conductance-combined graph}.
\end{itemize}
Then we have\[\Phiv(G[W])
=\Omega(\hat{\Delta}^{-2} c^{-1} d^{-1} \phi_i \phi_o)
= \phicut \phiemb / \poly(\Delta \phicut^{-1}, \log n)
= \phiemb / \poly(\Delta\phicut^{-1}, \log n). \qedhere\]
\end{proof}

We now describe how we select the output $C^\ast$ and $W^\ast$. Note that \cref{lem-property-final-w} shows that $\Phiv(G[W^\ast])$ has the required sparsity bound.
\begin{itemize}
    \item If $|C_2| \geq |V'|/3$, then we select $C^\ast$ to be one of $C_3$ and $V \setminus C_3$ of the smaller size, where $C_3$ is defined in \cref{lem-cut-case-1,lem-cut-case-2}. It is straightforward to see that we have $|V|/100 < |V'|/3 \leq |C^\ast| \leq |V|/2$ and $\Phiv(C^\ast) \leq \phicut$, and so we can set $W^\ast = \emptyset$.
    \item If $|C_2| < |V'|/3$ and $C_1 \cup C_2 \neq \emptyset$, then the set $W^\ast$ is defined as above, and $C^\ast$ is selected as $C_3$ defined in \cref{lem-cut-case-1,lem-cut-case-2}.
    Since $|C_1| < |V|/100$ and $|B| < \frac{\phicut}{4 \Delta} \cdot \floor{n/k}$, we must have $|C^\ast| \leq |V|/2$. Note that all vertices outside of $W^\ast$ are covered in $C^\ast$.
    \item If $C_1 \cup C_2 = \emptyset$, then the set $W^\ast$ is defined as above, and $C^\ast = \emptyset$. All vertices outside of $W^\ast$ are covered in $B$, and it has size $|B| < \frac{\phicut}{4 \Delta} \cdot \floor{n/k} = O(k^{-1} \Delta^{-1} \phicut) \cdot |V|$.
\end{itemize}

For the distributed implementation of finding $C_2$, we can simply gather all the needed information to one vertex $v^\ast$, and compute the cut $C_2$ there. The vertex $v^\ast$ only need to know the following information.
\begin{itemize}
    \item The number of edges between  $U_i^\ast$ and $U_j^\ast$, for each $1 \leq i < j \leq k'$.
    \item The size $|U_i^\ast|$ of $U_i^\ast$, for each $1 \leq i \leq k'$.
\end{itemize}
 We can use \cref{lem-basic} to calculate these numbers in $O(D + k^2)$ rounds and have them sent to $v^\ast$.
 
\paragraph{A note on local computation time.} 
The above procedure for calculating $C_2$ is efficient in terms of round complexity but it is inefficient in that it requires a brute-force search overall possible cuts respecting the partition $\UU^\ast = \{U_1^\ast, U_2^\ast, \ldots, U_{k'}^\ast\}$. 
This requires $2^{O(k')}$ time. We note that this issue can be solved by applying an approximate balanced sparse cut algorithm of~\cite[Theorem 2.7]{gao2019deterministic}, which costs only $\poly(k')$ time, and the approximation only causes $\phiemb$ to decrease by a factor of at most $\poly\left(\Delta\phicut^{-1}, \log n\right)$, and so it does not affect the analysis in this paper.

%Since the number of leftover vertices is small (only $\betaleft = O(1/k)$ fraction), either taking the cut $C''$ or the expander $G\{W\}$ above fulfills the task of $\cutemb$ with parameters $\phicut' = \phicut$, $\phiemb' = \Omega(\phiemb \phicut /(c^2d \log n))$, and $\betabal' = \Omega(1/k)$. We can upper bound the round complexity of such task  $\cutemb$ by $O(\Dst + k^2) + O(\poly(\phicut^{-1}, k, \log n))$ plus the round complexity for solving $\cutmatch$.

\subsection{Round Complexity Analysis}

\cref{lem-bal-improve-main} is an auxiliary lemma showing that the balance parameter $\betacut$ can be improved to $1/3$ using $O\left(\betacut^{-1} \right)$ iterations of $\detbalspcut{\phicut/2}{\phiemb}{\betacut}{\betaleft}$. The proof of \cref{lem-bal-improve-main} is left to \cref{sect-Bal-improve}.

\begin{lemma}[Balance improvement]
\label{lem-bal-improve-main}
\begin{align*}
   &\Tcut\left(n, \Delta, D, \phicut, \phiemb, \frac{1}{3}, \betaleft \right)\\
    &\leq O\left(D \betacut^{-1} \right) 
    + O\left(\betacut^{-1} \right) \cdot \Tcut\left(n, \Delta, D, \frac{\phicut}{2}, \phiemb,  \betacut, \betaleft\right)
\end{align*}
\end{lemma}

In \cref{lem-simul-emb-det}, the recursive calls of $\detbalspcut{\phicut'}{\phiemb'}{\betacut'}{\betaleft'}$ have many of its parameters fixed. To simplify the analysis, we define 
\[
\Tcut^\star(n, \phiemb) = 
 %O\left(\betacut^{-1} \right) \cdot
 \Tcut\left(n, O(\log n), O\left(\log^4 n \right), \frac{1}{4}, \phiemb, \frac{1}{3}, \frac{1}{12}\right).
\]
Writing $\Temb(n,\Delta,D,k,\phicut,\phiemb)$ to denote the round complexity of \cref{lem-simul-emb-det},  we have
\begin{align*}
&\Temb\left(n,\Delta,D,k,\phicut,\phiemb\right)\\
%&\leq
%\poly\left(\Delta\phicut^{-1},D,k, \log n\right) +
%\poly\left(\Delta\phicut^{-1}, \log n\right) \cdot 
%\Tcut\left( O\left(\frac{n}{k}\right), O(\log n), kD, \frac{1}{2}, \phiemb, \frac{1}{3}, \frac{1}{12}\right) \\
%%%%%%%
&\leq
\poly\left(\Delta\phicut^{-1},D,k, \log n\right)% \\
%& \ \  
%& \ \ 
+
\poly\left(\Delta\phicut^{-1}, \log n\right)  \cdot 
%\left(
%O\left( kD + \log^8 n \right) +
\Tcut\left( O\left(\frac{n}{k}\right), O(\log n), O\left(\log^4 n \right), \frac{1}{4}, \phiemb, \frac{1}{3}, \frac{1}{12}\right) \\
%\right)\\
%%%%%%%
%&= 
%\poly\left(\Delta\phicut^{-1},D,k, \log n\right) +
%\poly\left(\Delta\phicut^{-1}, \log n\right) \cdot 
%\left(
%O\left( kD + \log^8 n \right) +
%\Tcut^\star\left( O\left(\frac{n}{k}, \phiemb \right) \right)
%\right)\\
%%%%%%%
&= 
\poly\left(\Delta\phicut^{-1},D,k, \log n\right) +
\poly\left(\Delta\phicut^{-1}, \log n\right) \cdot 
\Tcut^\star\left( O\left(\frac{n}{k}\right), \phiemb \right),
\end{align*}
where the  %first 
inequality is due to \cref{lem-simul-emb-det}. 
%, and the second inequality is due to \cref{lem-diam-reduction-main}.
Using \cref{lem-merge-subgraphs,lem-bal-improve-main}, we can bound $\Tcut$ recursively as follows.
\begin{align*}
   &\Tcut\left(n, \Delta, D, \phicut, \phiemb, \frac{1}{3}, O\left(\phicut\Delta^{-1}k^{-1}\right) \right)\\
%%%%%%%%%
    &\leq O\left(D  \right) 
    + O\left(1\right) \cdot \Tcut\left(n, \Delta, D, \frac{\phicut}{2}, \frac{1}{100},  \betacut, O\left(\phicut\Delta^{-1}k^{-1}\right) \right)\\
%%%%%%%%%
&\leq O\left(D + \phicut^{-2} \Delta^2 \log^6 n \right) 
    + O\left(1\right) \cdot \Tcut\left(n, \Delta, O\left(\phicut^{-1} \Delta \log^3 n \right) , \frac{\phicut}{4}, \frac{1}{100},  \betacut, O\left(\phicut\Delta^{-1}k^{-1}\right) \right)\\
%%%%%%%%%
&\leq O\left(D\right) 
+ \poly\left(\Delta\phicut^{-1}, k, \log n\right)\\ 
& \ \ + O\left(1\right) \cdot 
\Temb\left(n,\Delta,O\left(\phicut^{-1} \Delta \log^3 n \right),k, \frac{\phicut}{4},\phiemb \cdot \poly\left(\Delta\phicut^{-1}, \log n\right)\right)\\
%%%%%%%%%
&\leq O\left(D\right) 
+ \poly\left(\Delta\phicut^{-1}, k, \log n\right)\\
& \ \ +
\poly\left(\Delta\phicut^{-1}, \log n\right) \cdot 
\Tcut^\star\left( O\left(\frac{n}{k}\right), \phiemb \cdot \poly\left(\Delta\phicut^{-1}, \log n\right) \right),
\end{align*}
where the first three inequalities are due to 
\cref{lem-bal-improve-main},
\cref{lem-diam-reduction-main}, and 
\cref{lem-merge-subgraphs}, respectively. 

For any given parameter $0 < \betaleft < 1$, we can use $k = \max\left\{1, O\left(\phicut\Delta^{-1}\betaleft^{-1}\right) \right\}$ in the above inequality to ensure that $\betaleft = \Omega\left(\phicut\Delta^{-1}k^{-1}\right)$, and so
\begin{align*}
   &\Tcut\left(n, \Delta, D, \phicut, \phiemb, \frac{1}{3}, \betaleft \right)\\
%%%%%%%%%
&\leq O\left(D\right) 
+ \poly\left(\Delta\phicut^{-1}, \log n, \betaleft^{-1}\right)  +
\poly\left(\Delta\phicut^{-1}, \log n\right) \cdot 
\Tcut^\star\left( n, \phiemb \cdot \poly\left(\Delta\phicut^{-1}, \log n\right) \right).
\end{align*}

The recurrence relation becomes much simpler if we restrict our attention to $\Tcut^\star$.
\begin{align*}
\Tcut^\star\left( n, \phiemb \right)   
&\leq   \poly\left( k, \log n\right)  +
\poly\left( \log n\right) \cdot 
\Tcut^\star\left( O\left(\frac{n}{k}\right), \phiemb \cdot \poly\left( \log n\right) \right).    
\end{align*}
For the base case of $n = O(1)$, there must be a  constant $\psi_0$ such that $\Tcut^\star\left( n, \phiemb \right) = O(1)$ when $\phiemb \geq \psi_0$. Suppose we always fix $k = 2^{\epsilon \log n}$ in all recursive calls, where $0<\epsilon<1$, and the parameter $n$ is the one in the top level of recursion, i.e., we do not change $k$ in recursive calls because the parameter  $n$ changes. Then the depth of recursion in order to reduce the number of vertices from $n$ to $O(1)$ is $d = \epsilon^{-1}$. Therefore, we have
\[
\Tcut^\star\left( n, 2^{-O\left( \epsilon^{-1} \log \log n \right)} \right)
=
\Tcut^\star\left( n, \log^{-O( \epsilon^{-1} )} n \right)
\leq    2^{O(\epsilon \log n)} \cdot \log^{O(\epsilon^{-1})} n
=     2^{O\left(\epsilon \log n + \epsilon^{-1} \log \log n \right)}.
\]
Combining this with the previous calculation, we have the following theorem.  

\begin{theorem}[Round complexity analysis]\label{thm-solve-recursion}
 For any $0<\epsilon<1$, the following holds.
\begin{align*}
 &\Tcut\left(n, \Delta, D, \phicut, \frac{1}{ \left(\Delta  \phicut^{-1} \right)^{O(1)} 2^{O\left( \epsilon^{-1} \log \log n \right)}}, \frac{1}{3}, \betaleft \right)\\
 & \ = O\left(D\right) + 
 \poly\left(\Delta \phicut^{-1}, \log n, \betaleft^{-1} \right)
 +  \left(\Delta  \phicut^{-1} \right)^{O(1)} \cdot 2^{O\left(\epsilon \log n + \epsilon^{-1} \log \log n \right)} \\
 &\Temb\left(n,\Delta,D,k,\phicut,\frac{1}{  2^{O\left( \epsilon^{-1} \log \log n \right)}}\right)\\
 & \ =  
 \poly\left(D, k, \Delta \phicut^{-1}, \log n\right)
 +  \left(\Delta  \phicut^{-1} \right)^{O(1)} \cdot 2^{O\left(\epsilon \log n + \epsilon^{-1} \log \log n \right)}
\end{align*}
\end{theorem}

In particular, for bounded-degree graphs, with $\epsilon = \sqrt{\log \log n /  \log n}$, 
\[\detbalspcut{\phicut}{\phiemb}{\betacut}{\betaleft}\] can be solved in \[O(D) + \poly\left(\phicut^{-1}, \betaleft^{-1}\right) \cdot 2^{O\left(\sqrt{ \log n \log \log n} \right)}\] rounds with
\begin{align*}
    \phiemb &= \poly\left(\phicut\right) 2^{-O\left(\sqrt{ \log n \log \log n} \right)},
    %\ \ \ \text{and}\\
    %\betaleft &= \poly\left(\phicut\right) 2^{-\Omega\left(\sqrt{ \log n \log \log n} \right)},
\end{align*} and hence we conclude the proof of \cref{lem-det-cutmatch-main}.
\section{Expander Routing}\label{sect-routing}

In this section, we consider a routing task on a network $G = (V,E)$, where each vertex $v \in V$ is the source and the destination of at most $\deg(v)$ messages of $O(\log n)$ bits.
The address of the destination of a message is given by unique identifiers of vertices, and we assume that each vertex $v \in V$ is associated with a distinct $O(\log n)$-bit identifier $\ID(v)$.

Ghaffari, Kuhn, and Su~\cite{GhaffariKS17} showed that this problem can be solved in $\mix(G) \cdot 2^{O(\sqrt{\log n \log \log n})}$ rounds with high probability, where $\mix(G) = O(\phi^{-2} \log n)$ for any graph $G$ with conductance $\phi$. This round complexity was subsequently improved to $\mix(G) \cdot 2^{O(\sqrt{\log n})}$ by Ghaffari and Li~\cite{GhaffariL2018}. Both of these algorithms are randomized and rely heavily on random walks. The goal of this section is to give an efficient \emph{deterministic} algorithm for this routing task.

\paragraph{A note on unique identifiers.}
The first step of our algorithm is to take the expander split graph $\Gsp$ of $G$ and then simulate $\Gsp$ on $G$, and so we can restrict our attention to bounded-degree graphs. However, care has to be taken to handle the unique identifiers. Specifically, how do we reduce a given routing task on $G$ to a routing task on $\Gsp$?

First of all,  each $v$ in $G$ corresponds to $\deg(v)$ vertices $X_v \subseteq \Vsp$ in $\Gsp$, and we can simply assign them the identifiers $(\ID(v), 1), (\ID(v), 2), \ldots, (\ID(v), \deg(v))$. The vertex $v$ initially holds $\deg(v)$ messages, and $v$ can simply distribute these messages to its  $\deg(v)$ corresponding vertices in $\Gsp$ so that each of them holds at most one message.

One remaining issue is how we assign the destination address to each message. If randomness is allowed, one thing we could do is the following.
In $O(D + \log n)$ rounds, we can re-assign distinct $\ID$s to $V$ in such a way that each vertex $v \in V$ is able to calculate $\lfloor \log \deg(u) \rfloor$ given $\ID(u)$~\cite{ChangPZ19}. For each message $\ID(v) \rightsquigarrow \ID(u)$, the vertex $v$ can reset its destination from $\ID(u)$ to $(\ID(u), x)$ by sampling $x$ uniformly at random from $\{1, 2, \ldots,  \lfloor \log \deg(u) \rfloor \}$.
After this destination assignment, it can be shown by a Chernoff bound that each vertex in $\Gsp$ is the destination of at most $O(\log n)$ messages.

This strategy~\cite{ChangPZ19} does not work in the deterministic setting. To handle this issue, we will change the definition of the  routing problem to allow the destination of a message $m$ to be specified by a \emph{range of numbers} $[i, j]$ instead of a specific number $k$. It means that the message $m$ can be delivered to any vertex $u$ whose $\ID$ belongs to the range $[i,j]$.
For a message $m$ with destination range $[i, j]$, we say that its \emph{$u$-weight} is 0 if $\ID(u) \notin [i,j]$, and is $1/x$ otherwise,  where $x = |\{v \in V \ | \ID(v) \in [i,j] \ \}|$. Given a set of messages $\mathcal{M}$, the summation of their $u$-weight is the \emph{expected} number of messages that $u$ receive if all messages $\mathcal{M}$ are sent to a uniformly random destination within their  destination range.

It is straightforward to reduce a   routing instance on $G$, where each vertex $v$ is a source and a destination of at most $O(L) \cdot \deg(v)$ messages, to a  routing instance on $\Gsp$ where the destination of a vertex is a range of numbers in such a way that each vertex $u$ in $\Gsp$ is a source of at most $O(L)$ messages, and the summation of $u$-weight over all messages is also $O(L)$. All we need to do is to reset the destination of each message from $\ID(u)$ to the range $[(\ID(u),1), (\ID(u), n)]$.

A feature of this reduction is that the range of destination of any two messages are either identical or disjoint. Having this property makes things simpler but we note that our deterministic algorithm does not depend this property and is able to work with overlapping destination ranges.

\paragraph{Notation.}
%In this section we present our deterministic routing algorithm.
Let $G=(V, E)$ be the current graph under consideration, and let $\MM$ be a set of messages, where each $m \in \MM$ consists of the following.
\begin{itemize}
    \item An $O(\log N)$-bit message.
    \item A destination range $[L_m, U_m]$.
\end{itemize}

%denote $\routing(G, \MM)$ as corresponding routing problem. We write $T^r(n, N, \Delta, L_S, L_T, \phi)$ as the round complexity for solving $\routing(G, \MM)$. 
Each message $m \in \MM$ is initially located at some vertex $v \in V$. The goal of the routing is to re-distribute the messages in such a way that each $m \in \MM$ is sent to a vertex $u$ with $\ID(u) \in [L_m, U_m]$.
We describe some parameters relevant to us in the routing task.
\begin{description}
    \item[Basic parameters:] As before, $n$, $\Delta$, and $D$ are the number of vertices, maximum degree, and the Steiner tree diameter of the current  graph $G=(V,E)$ under consideration.
    Since the diameter of a graph $G$ with $\Phiv(G) \geq \psi$ is always $\poly
    \left(\Delta\psi^{-1}, \log n\right)$,  we can get rid of the parameter $D$ by always replacing it with $\poly
    \left(\Delta\psi^{-1}, \log n\right)$.
    
    \item[Graph sparsity:] $\Delta > \Phiv(G) \geq  \psi$ is the sparsity lower bound of the current  graph. 
    \item[Range of identifiers:]  The range of the unique identifiers is $\{1, 2, \ldots, N\}$. We assume that the  length of IDs $O(\log N)$ fits into one message, and we assume $n \leq N$.
    
    \item[Maximum load at a source:]  Each vertex $v \in V$ is a source of at most $\Ls$ messages $m \in \MM$ initially. 
    
     \item[Maximum expected load at a destination:] Consider the distribution where each message   $m \in \MM$ is sent to a uniformly random vertex in $\{ v \in V \ | \ \ID(v) \in [L_m, U_m] \}$. Then $\Lt$ is defined as the  maximum expected number of messages that a vertex $v \in V$ receives. Using the terminologies of previous discussion,  the summation of $v$-weight 
 over all messages in $\MM$ is at most $\Lt$, for each $v \in V$. 
 
     \item[Maximum load at a destination:] We write $\Lf$ to be the maximum allowed number of messages sent to a vertex after the routing algorithm is finished. The parameter $\Lf$  depends on the algorithm, and we usually have $\Lf \gg \Lt$.
     %since $\Lt$ is the \emph{expected} one.
\end{description}

As discussed earlier, the  routing problem~\cite{GhaffariKS17,GhaffariL2018} on a  graph $G=(V,E)$ where each vertex $v\in V$ is a source and a destination of at most $O(L) \cdot \deg(v)$ messages  can be reduced to the aforementioned routing task on the expander split graph $\Gsp=(\Vsp, \Esp)$ with parameters $n = 2|E|$, $N = \poly(n)$,  $\Delta = O(1)$, $\Ls = O(L)$,  $\Lt = O(L)$, and $\psi = \Phiv(\Gsp) = \Omega(\Phie(G))$. The goal of this section is to prove the following result.

\begin{theorem}[Deterministic routing on bounded-degree expanders]\label{thm-routing-bounded-deg}
Let $G=(V,E)$ be a bounded-degree graph with $\Phiv(G) = \psi$.
Suppose each vertex $v \in V$ is a source of $\Ls = O(L)$ messages,
and the destination of each message $m$ is specified by a range of identifiers $[L_m, U_m]$ in such a way that the expected number of messages that a vertex receives is $\Lt = O(L)$ if all messages are delivered to a uniformly random vertex whose ID is within the allowed destination range.
Then there is a deterministic algorithm with round complexity
\[O(L) \cdot \poly\left(\psi^{-1}\right) \cdot  2^{O\left( \log^{2/3} n \log^{1/3} \log n \right)}\]
that sends each message $m$ to a vertex $v \in V$ with $\ID(v) \in [L_m, U_m]$, and each vertex receives at most \[\Lf = O(L) \cdot 2^{O\left( \log^{1/3} n \log^{-1/3} \log n \right)}\] messages.
\end{theorem}

By applying the algorithm of \cref{thm-routing-bounded-deg} to the expander split graph $\Gsp$ using the straightforward reduction described earlier, we have the following result for general graphs.

\Rrouting*

%\paragraph{Simplification.} Because we only consider the case where $\Delta$ is at most $O(\log n)$ and the diameter of a graph $G$ with $\Phiv(G) = \psi$ is at most $\poly(\Delta \psi^{-1}, \log n)$, we often eliminate these two parameters $\Delta$  and $D$ in  the calculation in subsequent discussion by hiding them into $\poly(\psi, \log n)$.

\paragraph{Load balancing.}
A crucial ingredient of our routing algorithm is a deterministic \emph{load balancing} algorithm of Ghosh et al.~\cite{GhoshLMMPRRTZ99}.
Consider a graph $G=(V,E)$ with $\Phiv(G) \geq \psi$, and $n$ is the number of vertices, and $\Delta$ is the maximum degree. Suppose each vertex $v \in V$ initially has $t_v$ tokens. Let $L = \sum_{v \in V} t_v / n$ be the average load, and let $M =\max_{v \in V} t_v$ be the maximum load. Ghosh et al.~\cite{GhoshLMMPRRTZ99} showed that there is an $O(\psi^{-1} M)$-round algorithm that redistributes the tokens in such a way that the maximum load at a vertex is at most $L + O(\psi^{-1} \Delta^2  \log n)$.

\begin{lemma}[{Load balancing~\cite[Theorem 3.5]{GhoshLMMPRRTZ99}}]\label{lem-loadbal-main}
There is an $O(\psi^{-1} M)$-round deterministic algorithm that redistribute the tokens in such a way that the maximum load at a vertex is  at most $L + O(\psi^{-1} \Delta^2  \log n)$. During the algorithm, at most one token is sent along each edge in each round.
%There is an $O(\psi^{-1} \Delta^{-1} M^2)$-round deterministic algorithm that redistribute the tokens in such a way that the maximum load at a vertex is  at most $2L + O(\psi^{-1} \Delta^2  \log n)$. During the algorithm, at most one token is sent along each edge in each round.
\end{lemma}

\subsection{Graph Partitioning} \label{sect-rout-partition}
The first step of the routing algorithm is to
partition the vertex set $V$ into $\VV = \{V_1, V_2, \ldots, V_k\}$ of $V$   in such a way that \[\left\lfloor \frac{n}{k} \right\rfloor \leq |V_i| \leq  \left\lceil \frac{n}{k} \right\rceil, \ \ \ \text{for each $1 \leq i \leq k$}\] and \[\max_{v \in V_i} \ID(v) < \min_{v \in V_j} \ID(v) \ \ \ \text{for each $1 \leq i < j \leq k$,}\]
where $k$ is some parameter to be determined. The computation of the partition takes $O(Dk \log N)$ rounds deterministically via $k-1$ binary searches using \cref{lem-det-bsearch}.

We remark that the purpose of having $\max_{v \in V_i} \ID(v) < \min_{v \in V_j} \ID(v)$ for each $1 \leq i < j \leq k$ is to ensure that any vertex $v \in V$ can locally calculate which part $V_j$ an arbitrary vertex $u$ belongs to, given the information $\ID(u)$. This can be done once we let everyone learn $\max_{v \in V_i} \ID(v)$ for each $1 \leq i \leq k-1$. 
This is crucial as we need to be able to route each message $m \in \MM$ to a part $V_i$ where $[L_m, U_m] \cap \{ \ID(v) \ | \ v \in V_i \} \neq \emptyset$.

As in~\cite{GhaffariKS17}, the routing will be done recursively in each part  $\VV = \{V_1, V_2, \ldots, V_k\}$ of $V$, and so we
do a simultaneous expander embedding using $\cutmatch{\phicut}{\phiemb}$ with some parameter $\phiemb$ to be determined and  $\phicut = \psi/2$. Recall from the specification of $\cutmatch{\phicut}{\phiemb}$ in \cref{sect-simult-emb-expander} that it outputs both a simultaneous expander embedding and a cut $C$ with $\Phiv(C) \leq \phicut$. Since $\Phiv(G) = \psi > \psi/2 = \phicut$ by our choice of $\phicut$, this forces $C =\emptyset$.

%More specifically, we apply $\cutmatch{\phicut}{\phiemb}$ with a partition 

Then $\cutmatch{\phicut}{\phiemb}$ returns $U_i \subseteq V_i$ for each $1 \leq i \leq k$ with $|U_i| \geq (2/3) |V_i|$, and a simultaneous embedding of $H_1, H_2, \ldots, H_k$ to $U_1, U_2, \ldots, U_k$ with congestion $c = \poly(\Delta\phicut^{-1}, \log n)$ and dilation $d = \poly(\Delta\phicut^{-1}, \log n)$.  Each graph $H_i$ is guaranteed to have sparsity $\Phiv(H_i) \geq \phiemb$ and maximum degree $\Delta(H_i) = O(\log n)$.

% The cost of finding the partition $O(Dk \log N)$  is negligible comparing to the cost of 
By \cref{thm-solve-recursion}, the round complexity of $\cutmatch{\phicut}{\phiemb}$ is 
\begin{align*}
 &\Temb\left(n,\Delta,D,k,\phicut,\frac{1}{  2^{O\left( \epsilon^{-1} \log \log n \right)}}\right) \\
 & \  = 
 \poly\left(D, k, \Delta \phicut^{-1}, \log n\right)
 +  \left(\Delta  \phicut^{-1} \right)^{O(1)} \cdot 2^{O\left(\epsilon \log n + \epsilon^{-1} \log \log n \right)}  \\
 & \ \ \ \text{for} \ \  \phiemb = \frac{1}{  2^{O\left( \epsilon^{-1} \log \log n \right)}}, \ \ \text{where} \ \ 0<\epsilon<1.
\end{align*}
In our routing algorithm, we fix $\epsilon = \log^{-1/3}n \log^{1/3}\log n$, and so we have
\[\phiemb = \frac{1}{  2^{O\left( \log^{1/3}n \log^{2/3}\log n \right)}},\]
and the round complexity of $\cutmatch{\phicut}{\phiemb}$ becomes
\[\poly\left(k, \Delta \psi^{-1}, \log n\right)
 +  \left(\Delta  \psi^{-1} \right)^{O(1)} \cdot 2^{O\left(\log^{2/3}n \log^{1/3}\log n   \right)}. \]
Note that the parameter $D$ is omitted because it is $\poly\left(\Delta\psi^{-1}, \log n\right)$.

\subsection{Updating Destination Ranges} For each message $m$ whose destination range $[L_m, U_m]$  is completely within the ID range of one part $V_i$, i.e., $[L_m, U_m] \subseteq [\min_{v \in V_i} \ID(v), \max_{v \in V_i} \ID(v)]$, then we can simply send $m$  to any vertex in $U_i$ using \cref{lem-comm-links}.
If $[L_m, U_m]$  overlaps with more than one part, then care needs to be taken when deciding which part the message $m$ is sent to, so that the new parameter $\Lt'$ in recursive calls is within a constant factor of the current $\Lt$. Intuitively, we do not want a part $V_i$ to receive significantly more messages then the expected number of messages that it receives, if the messages are sent to a uniformly random destination in their destination ranges.

\begin{lemma}[Update destination ranges of messages]\label{lem-msg-dest}
There is a deterministic algorithm with round complexity $O(k^3 D \log^3 N)$  that resets the destination range of each message $m \in \MM$ from $[L_m, U_m]$ to $[L_m, U_m] \cap [\min_{v \in V_i} \ID(v), \max_{v \in V_i} \ID(v)]$  for some $1 \leq i \leq k$ in such a way that the new parameter $\Lt'$ is within at most a constant factor of the old $\Lt$.
\end{lemma}
\begin{proof}
If $[L_m, U_m] \subseteq  [\min_{v \in V_i} \ID(v), \max_{v \in V_i} \ID(v)]$  for some $1 \leq i \leq k$ already, then nothing needs to be done for $m$.
Consider the case $[L_m, U_m]$ overlaps the ID range for more than one part of $\VV$. If we consider the easy case where the destination ranges of any two messages are either identical or disjoint, then there can be at most $k-1$ distinct distance ranges that overlaps more than one part of $\VV$. Hence we can afford to  deal with them individually. Specifically, let $[a, b]$ be a distance range that overlaps more than one part of $\VV$, then  we calculate the size $s_i$ of $[a, b] \cap \{ \ID(v) \ | \ v \in V_i\}$ for each $V_i \in \VV$ whose ID range overlaps with  $[a, b]$. Let $\MM'$ be the set of messages with this distance range $[a,b]$. We assign the messages in  $\MM$ to different parts in $\VV$ according to the distribution weighted by $s_1, s_2, \ldots, s_k$ using \cref{lem-partition}. Clearly the new parameter $\Lt'$ is within at most a constant factor of the old $\Lt$, as the error is only caused by rounding fractional values, and this can be implemented in $O(kD)$ rounds.

For the rest of the proof, suppose we are in the more challenging setting where the destination range  $[L_m, U_m]$ can be arbitrary. We define   \[Z_{m,i} = \left\lfloor \log |\{v \in V_i \ | \ \ID(v) \in [L_m, U_m] \}| \right\rfloor.\] It is straightforward to see that we can let each vertex $v \in V$ calculates $Z_{m,i}$ for each message $m$ at $v$ and for each $1 \leq i \leq k$ in $O(k D \log^2 n)$ rounds by learning  these $O(k \log n)$ numbers: the  $2^j$th smallest ID and the $2^j$th largest ID of vertices in $V_i$, for each $1 \leq j \leq \lfloor \log |V_i| \rfloor$, for each $1 \leq i \leq k$. 
These numbers can be calculated  using  binary search of \cref{lem-det-bsearch} in  $O(k D \log^2 N)$ rounds. 

Now we can classify the messages $m \in \MM$ based on their vectors $(Z_{m,1}, Z_{m,2}, \ldots, Z_{m,k})$.
Observe that there are at most $O(k^2 \log^2 n)$ possible vectors, and so we can afford to deal with each of them individually using the approach for the easy case described earlier. That is, we use \cref{lem-partition} to distribute the messages of each class to different parts in $\VV$ according to the weighted distribution corresponding to the $Z$-vector. The round complexity  of this step is \[O(k^2 \log^2 n) \cdot O(kD) = O(k^3 D \log^3 n) \leq O(k^3 D \log^3 N),\]
which is the dominating term in the overall round complexity. It is clear that  the new parameter $\Lt'$ is within at most a constant factor of the old $\Lt$, as each value $Z_{m,i}$ is a 2-approximation of the actual size of 
$[L_m, U_m] \cap    \{ \ID(v) \ | \ v \in V_i\}$.
\end{proof}

\subsection{Establishing Communication Links} \label{sect-rout-comm} %Intuitively, we will solve $\routing(\Gsp, \MM)$ recursively by applying  $\routing$ on $H_1, H_2, \ldots, H_k$. 
We write $u \rightsquigarrow v$ to denote the task of routing one message from $u$ to $v$. For two vertex subsets $U \subseteq V$ and $W \subseteq V$, we write $U \rightsquigarrow W$ to denote the task of routing one message from each $u \in U$ to vertices in $W$, and it does not matter which vertices in $W$ are the destinations. The two main quality measures of a routing algorithm are the round complexity and the maximum number of messages that a vertex receives.
We show how to solve the routing task $U_i \rightsquigarrow U_j$ efficiently.
%and $V \setminus (U_1 \cup U_2 \cup \cdots \cup U_k)  \rightsquigarrow U_1 \cup U_2 \cup \cdots \cup U_k$ efficiently.
%Consider the following routing problem.  Each vertex in $U_i$ holds at most one message that has to be delivered to $U_j$, and any vertex in $U_j$ can be the destination.
%We can establish  communication links between $U_i$ and $U_j$  that enables us to solve the above problem efficiently.

\begin{lemma}[Communication links between parts]\label{lem-comm-links}
For any $i, j \in [k]$, the routing task $U_i \rightsquigarrow U_j$ can be solved in 
$\poly\left(k, \Delta\psi^{-1} \right) \cdot 2^{O\left(\sqrt{\log n}\right)}$
rounds deterministically in such a way that each vertex is a destination of at most $2^{O\left(\sqrt{\log n}\right)}$ messages.
\end{lemma}
\begin{proof}
%We only focus on the case where the sender $u$ belongs to $U_i$, and each sender only holds a message needed to delivered to $U_j$. The routing task stated in this lemma can be solved by going over all $i, j \in [k]$.
%\paragraph{Matching between $U_i$ and $U_j$.}
The first part of the algorithm is to run the algorithm of \cref{lem-cut-match-det} with $T = U_j$, and $S$ being an arbitrary subset of $U_i$ with size $\min\{|U_i|, |U_j|\}$, and so $|S| \leq |T|$.
    Note that $|S| \geq (2/3)\lfloor n/k \rfloor \geq (2/3)|U_i| - 1 > (7/12) |U_i|$. We set $\beta' = (1/24)|U_i| / |V| = \Omega(k^{-1})$ and $\psi' = \psi/2$ for  the algorithm of \cref{lem-cut-match-det}. The choice of  $\psi'$ ensures   that the output cut $C$ of \cref{lem-cut-match-det} must be empty. The choice of $\beta'$ ensures that at most $(1/24)|U_i|$ vertices in $S$ are not matched. Therefore, in the end, at least $13/24$ fraction of the vertices in $U_i$ are matched to a vertex in $U_j$ in the output matching $M$ that can be embedded with congestion $O(\Delta^2 \psi^{-2} \log^4 n)$ and dilation   $O(\Delta \psi^{-1} \log n)$.  The round complexity of this part is $\poly\left(\Delta\psi^{-1}, k, \log n\right)$.

%\paragraph{From unmatched vertices to matched vertices.}
The second part is to handle the remaining $11/24$ fraction of the unmatched vertices in $U_i$ by applying \cref{lem-cut-match-det-noleftover} within the graph $H_i$, with $S$ being the set of unmatched vertices in $U_i$, and $T = U_i \setminus S$. As $\Phiv(H_i) \geq \phiemb$, the algorithm of \cref{lem-cut-match-det-noleftover} costs 
$\poly(\phiemb^{-1}) \cdot 2^{O\left(\sqrt{\log n}\right)}$
rounds in $H_i$, and it solves the routing task $S \rightsquigarrow T$  with congestion $\poly(\phiemb^{-1}) \cdot 2^{O(\sqrt{\log n})}$ and dilation $\poly(\phiemb^{-1}, \log n)$. Recall that the maximum degree of $H_i$ is $O(\log n)$, so we can eliminate this parameter in the above complexities. Furthermore, each $v \in T$ is a destination of at most   $\poly(\phiemb^{-1}) \cdot 2^{O(\sqrt{\log n})}$ messages from $S$.

Recall that $H_i$ is embedded into the underlying graph $G$ with congestion $c = \poly(\Delta\psi^{-1}, \log n)$ and dilation $d = \poly(\Delta\psi^{-1}, \log n)$, and so the actual round complexity, congestion, and dilation of the second part have to be multiplied by $\poly(\Delta\psi^{-1}, \log n)$.

Combining the communication links of the first part and the second part solves the required routing problem. We first use the communication links of the second part to route all messages in $U_i$ to the subset of $U_i$ that is matched by the matching $M$ of the first part, and then we use $M$ to deliver all of them to $U_j$. 
 The overall  round complexity can be upper bounded by
$\poly\left(k, \Delta\psi^{-1}, \phiemb^{-1} \right) \cdot 2^{O\left(\sqrt{\log n}\right)} =  \poly\left(k, \psi^{-1}  \right) \cdot 2^{O\left(\sqrt{\log n}\right)}$, and each vertex is a destination of at most  $\poly(\phiemb^{-1}) \cdot 2^{O\left(\sqrt{\log n}\right)} = 2^{O\left(\sqrt{\log n}\right)}$ messages, as $\phiemb = 1/  2^{O\left( \log^{1/3}n \log^{2/3}\log n \right)}$.
%As each $v \in U_j$ can be a destination of at most $\poly(\phiemb^{-1}) \cdot 2^{O(\sqrt{\log n})}$  messages from one part $U_i$, overall each $v \in U_j$ receives at most $\poly(k, \phiemb^{-1}) \cdot 2^{O(\sqrt{\log n})}$  messages since there are $k$ parts.
\end{proof}

%We will do \cref{lem-comm-links} sequentially for each $1 \leq i \leq k$ and $1 \leq j \leq k$. This adds a factor of $k^2$ in the round complexity.

%\paragraph{Step 3: communication links for leftover vertices.}
We can only do recursive calls on expanders, and the nature of our approach is that for each part $V_i$, we can only embed an expander on $U_i \subseteq V_i$, and there are always some \emph{leftover} vertices $V_i \setminus U_i$. By increasing the round complexity, it is possible to reduce the size of $V_i \setminus U_i$, but we cannot afford to make it an empty set using the techniques in this paper.
To deal with these leftover vertices, for each $v \in V_i \setminus U_i$, we will find another vertex $v^\star \in U_i$ that serves as the \emph{representative} of $v$ in all subsequent recursive calls. For each leftover vertex $v$ and its representative $v^\star$, we will establish a communication link between them.

\begin{lemma}[Communication links for leftover vertices]\label{lem-comm-links-leftover}
There is a deterministic algorithm with round complexity 
$\poly\left(k, \Delta\psi^{-1} \right) \cdot 2^{O\left(\sqrt{\log n}\right)}$
that finds a representative $v^\star \in U_i$ for each $v \in V_i \setminus U_i$,  for each $1 \leq i \leq k$. Moreover, each vertex serves as the representative of at most $\poly\left(\phiemb^{-1}, \log n\right) = 2^{O\left( \log^{1/3}n \log^{2/3}\log n \right)}$ vertices.
The algorithm also establishes communication links between them that allows us to solve the routing tasks $\{ v \rightsquigarrow v^\star \ | \ v \in V \setminus (U_1 \cup U_2 \cup \cdots \cup U_k)\}$
and $\{ v^\star \rightsquigarrow v  \ | \ v \in V \setminus (U_1 \cup U_2 \cup \cdots \cup U_k)\}$ deterministically  with round complexity 
$\poly\left(k, \Delta\psi^{-1} \right) \cdot 2^{O\left(\sqrt{\log n}\right)}$.
\end{lemma}
\begin{proof}
We only focus on the case of transmitting messages from each leftover vertices $v \in V_i \setminus U_i$ to a vertex $v^\star \in U_i$. 
The reverse direction can be done by re-using the communication paths used in the forward direction.
The routing algorithm  has three parts. 

The first part is the routing  \[V \setminus (U_1 \cup U_2 \cup \cdots  \cup U_k) \rightsquigarrow U_1 \cup U_2 \cup \cdots \cup U_k.\]
 %route the messages from all vertices in $V \setminus (U_1 \cup U_2 \cup \cdots U_k)$ to $U_1 \cup U_2 \cup \cdots U_k$ arbitrarily. 
 This is done using \cref{lem-cut-match-det-noleftover}, which costs 
$T_1 = \poly\left(\Delta\psi^{-1}\right)\cdot 2^{O\left(\sqrt{\log n}\right)}$ rounds. After the routing, the number of messages at each vertex in $U_1 \cup U_2 \cup \cdots \cup U_k$ is at most $M_1 = \poly\left(\Delta\psi^{-1}\right)\cdot 2^{O\left(\sqrt{\log n}\right)}$.

The second part is the routing  \[U_i \rightsquigarrow U_j, \ \ \ \text{for each} \ \ i, j \in [k],\]
as we would like to route the message originally from each $v \in V_i \setminus U_i$ to some vertex in $U_i$.
This is done by applying the algorithm of \cref{lem-comm-links} sequentially for all $O(k^2)$ pairs $i, j \in [k]$. 
%to ensure that each message arrives at the correct part. 
The round complexity is $T_2 = O(k^2) \cdot \poly\left(k, \Delta\psi^{-1} \right) \cdot 2^{O\left(\sqrt{\log n}\right)} = \poly\left(k, \Delta\psi^{-1} \right) \cdot 2^{O\left(\sqrt{\log n}\right)}$. After the routing, the number of messages at each vertex is at most $M_2 = M_1 \cdot O(k^2) \cdot 2^{O\left(\sqrt{\log n}\right)}
$.

The third part is to do a load balancing to reduce the numbers of messages per vertex from $M_2$ to $M_3 = \poly\left(\phiemb^{-1}, \log n\right)$, so that 
each vertex in $U_i$ serves as a representative for at most $\poly\left(\phiemb^{-1}, \log n\right) =  2^{O\left( \log^{1/3}n \log^{2/3}\log n \right)}$ leftover vertices, for each $1 \leq i \leq k$.
We apply the algorithm of \cref{lem-loadbal-main} on the virtual graph $H_i$ in parallel for all $1 \leq i \leq k$ with parameters $M = M_2$ and $L \leq 1/2$, as $|U_i| \geq 2 |V_i \setminus U_i|$. 
The round complexity of this algorithm is $T_3 = \poly\left(\Delta \psi^{-1}, \log n\right) \cdot O\left( \phiemb^{-1}  M \right) = \poly\left(k, \Delta\psi^{-1} \right) \cdot 2^{O\left(\sqrt{\log n}\right)}$, where  $\poly\left(\Delta \psi^{-1}, \log n\right)$ is the cost of simulating one round of $H_i$ on $G$, and recall that $\Delta(H_i) = O(\log n)$. 
%$= \poly\left( k, \phicut^{-1}, \phiemb^{-1} \right) \cdot 2^{O\left(\sqrt{\log n}\right)}$ in $H_i$, and it can be simulated in $G$ with an overhead of a factor of $\poly\left(\phicut^{-1}, \log n\right)$.
After the routing, the number of messages in each vertex is at most $M_3 = L + O\left(\phiemb^{-1} \left( \Delta(H_i)\right)^2 \log n\right) = \poly\left(\phiemb^{-1}, \log n\right)$ by \cref{lem-loadbal-main}.
\end{proof}

\subsection{Routing the Messages}\label{sect-routing-outline}
We are in a position to describe the entire routing algorithm.

\paragraph{Preprocessing step.} Run the graph partitioning algorithm as described in \cref{sect-rout-partition} to obtain a partition $\VV$, and then apply $\cutmatch{\phicut}{\phiemb}$ to embed an expander $H_i$ to a subset $U_i \subseteq V_i$ of each part $V_i \in \VV$. Apply the algorithm of \cref{lem-msg-dest} to reset the destination range of each message $m \in \MM$. The overall round complexity of the preprocessing step is \[\poly\left(k, \Delta \psi^{-1}, \log N\right)
 +  \left(\Delta  \psi^{-1} \right)^{O(1)} \cdot 2^{O\left(\log^{2/3}n \log^{1/3}\log n   \right)}, \]
as we note that $D = \poly\left(\Delta\psi^{-1}, \log n\right)$ and $\log n \leq \log N$. 

\paragraph{Sending the messages between parts.} 
After the above preprocessing step, we assign a representative $v^\star$ for each leftover vertex $v$ and route all the messages at $v$ to $v^\star$ using \cref{lem-comm-links-leftover}. Then we route the messages between $U_1, U_2, \ldots, U_k$ so that each message $m \in \MM$ goes to the part $U_i$ with $[L_m, U_m] \in \{ \ID(v) \ | \ v \in V_i \}$. This is done using the routing $U_i \rightsquigarrow U_j$, for each $i, j \in [k]$ by applying the algorithm of \cref{lem-comm-links} sequentially for all $O(k^2)$ pairs $i, j \in [k]$. The overall round complexity of this step is
\[O(\Ls) \cdot \poly\left(k, \Delta\psi^{-1} \right) \cdot 2^{O\left(\sqrt{\log n}\right)},\]
and each vertex $v \in U_1 \cup U_2 \cup \cdots U_k$ is the destination of at most 
\[M' = O(\Ls) \cdot \poly\left(k  \right) \cdot 2^{O\left(\sqrt{\log n}\right)}\]
messages.

\paragraph{Preparation for recursive calls.} 
%\thatchaphol{This paragraph and the previous one are quite repetitive to \Cref{lem-comm-links-leftover}. It seems the content in \Cref{lem-comm-links-leftover} should be separated into two parts(?)}
Several things need to be done before starting the recursive calls. Specifically, we need to do the following two things.
\begin{enumerate}
    \item For each $1 \leq i \leq k$, each $u \in U_i$ needs to locally simulate each leftover vertex $v$ such that $u = v^\star$ is the representative of $v$. 
    \item Balance the number of messages at each vertex.
\end{enumerate}

To handle the leftover vertices, we  modify the virtual subgraph $H_i$ as follows. For each $u \in U_i$ served as representative for $s = \poly\left(\phiemb^{-1}, \log n\right)$ leftover vertices, replace $u$ by an arbitrary $(s+1)$-vertex bounded-degree graph. These $s$ new virtual vertices have the same IDs as the other $s$ leftover vertices whose representative is $u$. Denote the resulting graph by $H_i^\star$. Now the set of IDs in $H_i^\star$ is identical to the set of IDs in $V_i$. This modification  worsens the sparsity of $H_i$ by a factor of at most $\poly\left(\phiemb^{-1}, \log n\right)$. This is not very bad, as we still have $\Phiv(H_i^\star) \geq  \Phiv(H_i) / \poly\left(\phiemb^{-1}, \log n\right)  =  \phiemb / \poly\left(\phiemb^{-1}, \log n\right) = 1 / 2^{O\left(\log^{2/3}n \log^{1/3}\log n   \right)}$. Note that this modification does not change the congestion and dilation of the simultaneous embedding.

Next, we apply the load balancing algorithm of \cref{lem-loadbal-main} on each graph $H_i^\ast$ to balance the load at each vertex. This allows us to reduce the maximum number of messages per vertex from the above $M'$ to 
\[M'' = \Lt + \poly\left(\phiemb^{-1}, \log n\right) = \Lt + 2^{O\left(\log^{2/3}n \log^{1/3}\log n   \right)},\]
and the round complexity is
\[\poly\left(\Delta \psi^{-1}, \log n\right) \cdot O\left( \phiemb^{-1},  {M'} \right) = \poly\left(\Ls, k, \Delta\psi^{-1} \right) \cdot 2^{O\left(\sqrt{\log n}\right)},\]
where $\poly\left(\Delta \psi^{-1}, \log n\right)$ is the overhead of simulating $H_1^\ast, H_2^\ast, \ldots, H_{k}^\ast$ in $G$ in parallel.

%After applying the algorithm of \cref{lem-msg-dest} to update the destination ranges, we transmit all messages to the parts $U_1, U_2, \ldots, U_k$ where their destination ranges belong to using \cref{lem-comm-links}, and then we use \cref{lem-loadbal-main} to do a load balancing so that each vertex holds at most $O(\Lt) + \poly\left(\phiemb^{-1}, \log n\right)$ messages. This will be the new $\Ls$-value for recursive calls.

%Recall that each $U_i \subseteq V_i$ is embedded an expander $H_i$ with congestion and dilation $\poly\left(\phicut^{-1}, \log n\right)$. 

\paragraph{Recursive calls.}
Now we can route the messages in $H_i^\ast$ to their destinations by a recursive call on $H_i^\ast$, in parallel for each $1 \leq i \leq k$.
We have the following parameters for this recursive call: 
\begin{align*}
  n' &= \lceil n/k \rceil, \\
  N' &= N,\\
  \Delta' &= O(\log n),\\
  \Ls' &= M'' = \Lt + \poly\left(\phiemb^{-1}, \log n\right) = \Lt + 2^{O\left(\log^{2/3}n \log^{1/3}\log n   \right)},\\
  \Lt' &= O(\Lt),\\
  \psi' &= 1 / \poly\left(\phiemb^{-1}, \log n\right) = 1 / 2^{O\left(\log^{2/3}n \log^{1/3}\log n   \right)},\\
  \Lf' &= \Lf,
\end{align*}
%$n' \leq \lceil n/k \rceil$, $N' = N$, $\Delta' = O(\log n)$, $\Ls' = M'' = O(\Lt) + \poly\left(\phiemb^{-1}, \log n\right) = O(\Lt) + 2^{O\left(\log^{2/3}n \log^{1/3}\log n   \right)}$, $\Lt' = O(\Lt)$, and $\psi' = 1 / \poly\left(\phiemb^{-1}, \log n\right) = 1 / 2^{O\left(\log^{2/3}n \log^{1/3}\log n   \right)}$, and
and each round in the parallel recursive calls can be simulated in $ \poly\left(\phicut^{-1}, \log n\right)$ rounds in the current graph $G$.

\paragraph{Postprocessing.}
After the recursive call, we just need to route the messages whose destination is a leftover vertex $v \in V_i \setminus U_i$ from its representative $v^\ast$ to $v$ using \cref{lem-comm-links-leftover}, and then we are done. This routing costs
\[O(\Lf) \cdot \poly\left(k, \Delta\psi^{-1} \right) \cdot 2^{O\left(\sqrt{\log n}\right)}\]
rounds.

\subsection{Round Complexity Analysis}

We are now in a position to analyze the round complexity of the routing algorithm and determine what value of $\Lf$ we can use.

We denote $\Troute(n, \psi, \Lt, \Lf)$ to denote the round complexity for solving the routing problem with these parameters with $\Delta = O(\log n)$  and $\Ls = O(\Lt) + \poly\left(\psi^{-1}, \log n\right)$.
The reason that we can have $\Ls = O(\Lt) + \poly\left(\psi^{-1}, \log n\right)$
is that we always do a load-balancing before starting applying the routing algorithm, see the part of preparation for recursive calls in \cref{sect-routing-outline}.
The parameter $N = \poly(n)$ is omitted as it is the same for all recursive calls. Remember that at the top level of recursion, we have $\Lt = O(L)$.

%The round complexities of different parts of the routing algorithms are summarized as follows Note that there are two parameters $1 \leq k \leq n$ and $0 < \epsilon < 1$ that we can select.
%First of all, the cost of the simultaneous expander embedding is
%\begin{align*}
% &\Temb\left(n,\Delta,D,k,\psi,\phiemb\right)  =  
% \poly\left(D, k, \Delta \psi^{-1}, \log n\right)
% +  \left(\Delta  \psi^{-1} \right)^{O(1)} \cdot 2^{O\left(\epsilon \log n + \epsilon^{-1} \log \log n \right)}  \\
% & \ \ \ \ =   \poly\left(k, \psi^{-1}, \log n\right)
% +  \left(\psi^{-1} \right)^{O(1)} \cdot 2^{O\left(\epsilon \log n + \epsilon^{-1} \log \log n \right)}  \\
% & \ \ \ \text{for} \ \  \phiemb = \frac{1}{ \left(\Delta  \psi^{-1} \right)^{O(1)} 2^{O\left( \epsilon^{-1} \log \log n \right)}}, \ \ \text{where} \ \ 0<\epsilon<1.
%\end{align*}
%After having the expander embeddings, the round complexity of all the tasks needed to be done before starting the recursive calls is

%\begin{align*}
% &O(\Ls) \cdot \poly\left(k, \psi^{-1}, \phiemb^{-1} \right) \cdot 2^{O\left(\sqrt{\log n}\right)}  = \poly(k) \cdot 2^{O\left(\epsilon \log n + \epsilon^{-1} \log \log n \right)}.
%\end{align*}
In view of the calculation in \cref{sect-routing-outline}, we have
\begin{align*}
&\Troute(n, \psi, \Lt, \Lf)\\
 & \leq \poly\left(\Lt, \Lf, \psi^{-1}, k, \log N \right) \cdot 2^{O\left(\log^{2/3}n \log^{1/3}\log n   \right)}\\
& \  \  +
 \poly\left(\psi^{-1}, \log n \right) \cdot
 \Troute\left( O\left(\frac{n}{k}\right), \frac{1}{  2^{O\left( \log^{1/3} n \log^{2/3} \log n \right)}}, O(\Lt), \Lf \right).
\end{align*}

We fix \[k = 2^{O\left( \log^{2/3} n \log^{1/3} \log n \right)}\] in all recursive calls, where $n$ is the number of vertices at the top level of recursion. Then it is clear that the depth of recursion is \[d = \log_k n = O\left(\frac{\log^{1/3} n }{\log^{1/3} \log n}\right).\]
At the bottom level of recursion with $n = O(1)$, we can solve the routing task by brute force and have $\Lf = O(\Lt)$. As $\Lt$ increases by at most a constant factor in each level of recursion, we always have $\Lt = O(L) \cdot 2^{O(d)} = O(L) \cdot  2^{O\left( \log^{1/3} n \log^{-1/3} \log n \right)}$, as we have $\Lt = O(L)$ at the top level of recursion. Thus, we can set $\Lf = O(L) \cdot  2^{O\left( \log^{1/3} n \log^{-1/3} \log n \right)}$ in all recursive calls.

Note that we have $\psi = 1 /  2^{O\left( \log^{1/3} \log^{2/3} \log n \right)}$ for all recursive calls except the top level. In these recursive calls, all of $\Lt, \Lf, \psi^{-1}, k, \log N$ are upper bounded by $O(L) \cdot  \poly\left(\psi^{-1}\right) \cdot 2^{O\left( \log^{2/3} \log^{1/3} \log n \right)}$. Therefore, we have
\begin{align*}
&\Troute\left(n, \psi, O(L), O(L) \cdot  2^{O\left( \log^{1/3} n \log^{-1/3} \log n \right)}\right)\\
&=    O(L) \cdot  \poly\left(\psi^{-1}\right) \cdot 2^{O\left( \log^{2/3} \log^{1/3} \log n \right)}  \cdot 2^{O(d) \cdot O\left( \log^{1/3} n \log^{2/3} \log n \right)}\\
& = O(L) \cdot  \poly\left(\psi^{-1}\right) \cdot  2^{O\left( \log^{2/3} n \log^{1/3} \log n \right)}.
\end{align*}
This completes the proof of \cref{thm-routing-bounded-deg}.
\section{Derandomization}\label{sect-derand}

In this section, we present two simple applications of our results in derandomizing distributed graph algorithms.

\subsection{Triangle Finding}
In this paper, we consider the following variants of the distributed triangle finding problems. We say that a vertex $v$ found a triangle $\{x,y,z\}$ if $v$ learned the three edges $\{x,y\}$, $\{y,z\}$, and $\{x,z\}$.
\begin{description}
\item[Triangle detection:] If the graph contains at least one triangle, then at least one vertex finds a triangle.
\item[Triangle counting:] Each vertex $v$ outputs a number $t_v$ such that $\sum_{v\in V}t_v$ equals the number of triangles in the graph.
\item[Triangle enumeration:] Each triangle in the graph is found by at least one vertex.
\end{description}
We prove the following theorem.

\Rtriangle*

Our triangle finding algorithm follows the high-level idea of~\cite{ChangPZ19}. We start with an expander decomposition, and consider a threshold $d$. For vertices with degree at most $d$, we can apply the trivial $O(d)$-round triangle listing algorithm, For the remaining high-degree vertices, we simulate a known $\CLIQUE$ algorithm for triangle finding on each high-conductance component of the expander decomposition. 

An   $(\epsilon,\phi)$-expander decomposition of a graph $G=(V,E)$
is a partition of its edges $E=E_{1}\cup E_2 \cup \cdots\cup E_{k} \cup \Er$.  We write $G_i = (V_i, E_i) = G[E_i]$ to denote the subgraph induced by $E_i$. Consider the following sets.
\begin{align*}
 W_i &= \{ \,  v \in V_i \ | \  \deg_{E_i}(v) \geq \deg_{E \setminus E_i}(v) \, \}\\
 E_i^+ &= \{ \,  e=\{u,v\} \in E \ | \ (u \in W_i) \vee (v \in W_i) \vee (e \in E_i) \, \}\\
 E_i^- &= \{ \,  e=\{u,v\}  \in E_i \ | \ \{u,v\} \subseteq W_i \, \}
\end{align*}

In other words, $W_i$ is a subset of $V_i$ containing vertices whose majority of incident edges are in $E_i$,  $E_i^+$ is the set of all edges that are incident to a vertex in $W_i$ or belong to $E_i$, $E_i^-$ is the set of edges in $E_i$ whose both endpoints are contained in $W_i$.

The following lemma allows us to focus triangles with at least one edge in $E_1^-,  E_2^-, \ldots, E_k^-$, and then recurse on the remaining edges, and the depth of the recursion is at most $O(\log n)$.

\begin{lemma}[Number of remaining edges is small]
If $\epsilon \leq  1/6$, then $|E_1^- \cup E_2^- \cup \cdots \cup E_k^-| \geq |E|/2$.
\end{lemma}
\begin{proof}
Observe that $|E_i \setminus E_i^-|$ must be at most the number of edges in $\Er$ incident to $V_i$. Since each edge $e \in \Er$ can be incident to at most two parts $V_i, V_j$, we have $|E| - |\Er| - |E_1^- \cup E_2^- \cup \cdots \cup E_k^-| = \sum_{1 \leq i \leq k} |E_i \setminus E_i^-| \leq 2|\Er|$, and so $|E_1^- \cup E_2^- \cup \cdots \cup E_k^-| \geq |E| - 3|\Er| \geq (1 - 3 \epsilon) |E| \geq |E|/2$.
\end{proof}

Note that for any triangle with at least one edge in $E_i^-$, all its three edges must be completely within $E_i^+$. The following lemma shows that $\Phie(G[E_i^+])$ is high.

\begin{lemma}[{$G[E_i^+]$ has high conductance}]
$\Phie(G[E_i^+]) \geq \phi/4$.
\end{lemma}
\begin{proof}
Recall that $G[E_i] = (V_i, E_i)$ has $\Phie(G[E_i]) \geq \phi$.
We write $G[E_i^+] = (V_i^+, E_i^+)$. To prove the lemma, it suffices to show that for any cut $C \subseteq V_i^+$ of $G[E_i^+]$ with $\Phie_{G[E_i]}(C) \leq 1/2$, we always have $C \cap V_i \neq \emptyset$ and
\[
\frac{|\partial_{G[E_i^+]}(C)|}{\vol_{G[E_i^+]}(C)}
\geq \frac{1}{4} \cdot  \frac{|\partial_{G[E_i]}(C \cap V_i)|}{\vol_{G[E_i]}(C \cap V_i)}.
\]
Since $E_i \subseteq E_i^+$, we  have $|\partial_{G[E_i^+]}(C)| \geq |\partial_{G[E_i]}(C)| \geq |\partial_{G[E_i]}(C \cap V_i)|$.
Therefore, we only need to show that 
\[
\vol_{G[E_i^+]}(C) \leq 4 \cdot  \vol_{G[E_i]}(C \cap V_i).
\]
We define the following four numbers.
\begin{align*}
    a &= \left|\{ \, e =\{u,v\} \in E_i \ | \  u \in C, v \notin C \, \}\right| \\
    b &= |\{ \, e =\{u,v\} \in E_i \ | \ \{u,v\} \subseteq C \, \}| \\
    c &= |\{ \, e =\{u,v\} \in E_i^+ \setminus E_i \ | \ u \in C, v \notin C \, \}|  \\
    d &= |\{ \, e =\{u,v\} \in E_i^+ \setminus E_i \ | \ \{u,v\} \subseteq C \, \}|  
\end{align*}
It is clear that $a+2b = \vol_{G[E_i]}(C \cap V_i)$ and $a+2b+c+2d = \vol_{G[E_i^+]}(C)$.

Since  $\Phie_{G[E_i]}(C) \leq 1/2$, we have 
$2a+2c = 2\cdot |\partial_{G[E_i^+]}(C)| \leq \vol_{G[E_i^+]}(C) = a+2b+c+2d$.
Observe that all edges in $\{ \, e =\{u,v\} \in E_i^+ \setminus E_i \ | \ \{u,v\} \subseteq C \, \}$ are outside of $E_i$ and are incident to $C \cap W_i \subseteq C \cap V_i$. In view of the definition of $W_i$,   the size $d$ of this set $\{ \, e =\{u,v\} \in E_i^+ \setminus E_i \ | \ \{u,v\} \subseteq C \, \}$ is at most $a+b$, since $a+b$ is the total number of edges in $E_i$ incident to $C \cap V_i$. 
Using these two inequalities $2a+2c \leq a+2b+c+2d$ and $d \leq a+b$, we have
\begin{align*}
    c &\leq -a +2b+2d  & 2a+2c \leq a+2b+c+2d\\
    c+2d &\leq -a +2b+4d  \\
    c+2d &\leq 3a +6b  & d \leq a+b\\
    a+2b+c+2d &\leq 4a+8b
\end{align*}
Therefore, $\vol_{G[E_i^+]}(C) \leq 4 \cdot \vol_{G[E_i]}(C \cap V_i)$, as required.
\end{proof}

In view of the above lemma, we are able to apply the routing algorithm of \cref{thm-routing-general} on $G[E_i^+] = (V_i^+, E_i^+)$. Note that each edge $e \in E$ belongs to at most two sets $E_i^+$ and $E_j^+$, and so congestion of processing $G[E_i^+]$ for $1 \leq i \leq k$ in parallel is not an issue. We deal with low-degree vertices and high-degree vertices separately.

\paragraph{Low-degree vertices.} Consider the $S_i = \{ v \in V_i \ | \  \deg_{G[E_i^+]}(v) \leq d \}$, where $d$ is a threshold to be determined. It is clear that by having all vertices in $v \in S_i$ sending  $\{ \ID(u) \ | \ \{u,v\} \in E_i^+ \}$ to all its neighbors, we are able to list all triangles in $G[E_i^+]$ involving at least one edge incident to $S_i$, in the sense that any such triangle is found by some vertex in $G[E_i^+]$. This algorithm takes $O(d)$ rounds deterministically.
Furthermore, we can additionally require that each triangle is listed by exactly one vertex by introducing a tie breaking mechanism. For a triangle $\{x,y,z\}$ in $G[E_i^+]$  with  $\{x,y,z\} \cap S_i \neq \emptyset$, we let the vertex in $\{x,y,z\} \cap S_i$ that has the largest ID to list the triangle $\{x,y,z\}$.

\paragraph{High-degree vertices.}
Now, we let $G_i^\ast=(V_i^\ast, E_i^\ast)$ to denote the subgraph of $G[E_i^+]$ induced by $V_i^\ast = V_i^+ \setminus S_i$. All triangles in $G[E_i^+]$ that are not listed yet belong to this graph $G_i^\ast$.  We will handle these triangles by simulating a known $\CLIQUE$ algorithm. Specifically, triangle enumeration can be solved in $O(n^{1/3} / \log n)$ rounds~\cite{DolevLP12} deterministically in  $\CLIQUE$, and triangle detection and counting can be solved in $n^{1 - 2\omega^{-1} + o(1)} < O(n^{0.158})$ rounds~\cite{Censor2016} deterministically in  $\CLIQUE$, where $\omega < 0.2373$ is the exponent for the complexity of matrix multiplication.

Let $n' = |V_i^\ast| \leq n$. To simulate  $\CLIQUE$ algorithms on $G_i^\ast$, we first need to re-assign the IDs of vertices in $V_i^\ast$ to $\{1, 2, \ldots, n'\}$.  ID re-assignment can be done straightforwardly in $O(D) = O(\phi^{-2} \log n)$ rounds deterministically, see~\cite[Lemma 4.1]{ChangPZ19}.
%\thatchaphol{Refer a lemma.} 

To simulate one round of  $\CLIQUE$ on $G_i^\ast$, 
we need to be able to let each vertex in $v \in V_i^\ast$ to route a separate $O(\log n')$-bit message to all other vertices in $V_i^\ast$. Since each vertex   $v \in V_i^\ast$ has degree  $\deg_{G_i^+}(v) > d$, one invocation of the routing algorithm of  \cref{thm-routing-general} on $G[E_i^+]$ with $L = O(n/d)$ allows each vertex 
  $v \in V_i^\ast$ to send and receive $O(n)$ messages, which is enough to  simulate one round of $\CLIQUE$. Therefore, the overhead of simulation   is $O(n/d)\cdot \poly\left(\phi^{-1}\right) \cdot n^{o(1)}$.

\paragraph{Round complexity analysis.}
We set $\epsilon = 1/6$ and $\phi =  1 / 2^{O(\sqrt{\log n \log \log n})}$. An $(\epsilon,\phi)$-expander decomposition of a graph $G=(V,E)$ can be found in $2^{O(\sqrt{\log n \log \log n})}$ rounds deterministically using \cref{thm-det-expander-decomp}.
Note that the round complexity for finding an $(\epsilon, \phi)$-expander decomposition is negligible comparing with other costs.

For triangle enumeration, the overall round complexity is $O(d) +  O(n^{1/3} /\log n) \cdot O(n/d) \cdot n^{o(1)}$ by simulating the  $O(n^{1/3} / \log n)$-round $\CLIQUE$ algorithm of~\cite{DolevLP12}.
Setting $d = n^{2/3}$, we obtain the overall round complexity  $n^{(2/3) +o(1)}$. 

For triangle counting and detection,  the overall round complexity is $O(d) +  n^{1 - 2\omega^{-1} + o(1)} \cdot O(n/d) \cdot n^{o(1)}$ by simulating the  $n^{1 - 2\omega^{-1} + o(1)}$-round $\CLIQUE$  algorithm of~\cite{Censor2016}.
Setting $d = n^{1 - \omega^{-1}}$, we obtain the overall round complexity 
$n^{1 - \omega^{-1} + o(1)} < O(n^{0.158})$. Hence we conclude the proof of \cref{thm-triangle}.

\paragraph{Isolating a triangle.} We note that the triangle detection algorithm of~\cite{Censor2016} only lets us know whether a triangle exists, but it does not find one explicitly when there is at least one triangle. To be able to find a triangle explicitly, we can apply the following strategy.
Suppose that the algorithm of~\cite{Censor2016} tells us that there is at least one triangle in   $G_i^\ast=(V_i^\ast, E_i^\ast)$. Then we partition the edge set  $E_i^\ast$ into four parts $E^1, E^2, E^3, E^4$ of equal size. We apply the same triangle detection algorithm for the subgraph induced by $E_i^\ast \setminus E^j$, for each $1 \leq j \leq 4$. Note that one of these four edge sets $E_i^\ast \setminus E^j$ must contain a triangle, and the triangle detection algorithm is able to tell us which of them has at least one triangle, and then we recurse on that edge set. It is clear that after $O(\log n)$ iterations, we are able to isolate exactly one triangle of $G_i^\ast$, and this only adds an $O(\log n)$ factor in the round complexity.

\paragraph{Avoiding repeated counting.} To solve the triangle counting problem, we need to avoid counting a triangle more than once. Therefore, when counting triangles in $G_i^\ast=(V_i^\ast, E_i^\ast)$, we need to make sure that we only count the triangles with at least one edge in $E_i^-$. This can be done by first calculating the number $n_1$ of triangles in $G_i^\ast$, and then calculating the number $n_2$ of triangles in the subgraph of $G_i^\ast$ induced by the edges $E_i^\ast \setminus E_i^-$. Then $n_1 - n_2$ equals the number of  triangles in $G_i^\ast$ with at least one edge in $E_i^-$.

\subsection{Minimum Spanning Trees}

We prove the following theorem.

\Rmst*

Our MST algorithm follows the approach of~\cite{GhaffariKS17}, which implements \emph{Boruvka's greedy algorithm}. During the algorithm, we maintain a forest $F$. Initially $F$ is the trivial forest consisting of 1-vertex trees for each vertex $v \in V$. When $F$ contains more than one tree, pick any tree $T$ in the forest, and let $e$ be any smallest weight edge connecting $T$ and $G \setminus T$, and then add this edge $e$ to the forest $F$. Note that in each step, two trees in the forest are merged into one. It is well-known that at the end of this process, we obtain a minimum spanning tree~\cite{GhaffariKS17}.

\paragraph{Distributed implementation.}
The distributed implementation in~\cite{GhaffariKS17} is randomized. Here we consider the following distributed version of Boruvka's greedy algorithm that can be implemented deterministically.
\begin{enumerate}
    \item Suppose currently the forest $F$ consists of the trees $T_1, T_2, \ldots, T_k$. Let $e_i$ be any minimum weight edge connecting $T_i$ and $G \setminus T_i$. 
    \item Let $G^\ast=(V^\ast, E^\ast)$ be a graph defined by $V^\ast = \{ T_1, T_2, \ldots, T_k\}$ and $\{T_i, T_j\} \in E^\ast$ if either $e_i$ or $e_j$ connects $T_i$ and $T_j$. It is straightforward to see that $G^\ast$ is a forest, and $|E^\ast| \geq k/2$.
    \item For each edge  $\{T_i, T_j\} \in E^\ast$, if it is added to $E^\ast$ because of $e_i$, then we orient this edge as $T_i \rightarrow T_j$; if it is  added to $E^\ast$ because of $e_j$, then we orient this edge as $T_j \rightarrow T_i$. If an edge $\{T_i, T_j\}$ can be oriented in both directions, then it is oriented arbitrarily. Define $\indeg(T_i)$ to be the number of edges in $E^\ast$ oriented towards $T_i$.
    \item Find an independent set $I^\ast$ of $G^\ast$ with $\sum_{T_i \in I^\ast} \indeg(T_i) \geq |E^\ast|/3$.
    \item For each $T_i \in I^\ast$ and for each $T_j$ with $T_j \rightarrow T_i$, add $e_j$ to the current forest.
\end{enumerate}

Since $\sum_{T_i \in I^\ast} \indeg(T_i) \geq |E^\ast|/3 \geq k/6$, the number of trees in the forest $F$ is reduced by a factor of $5/6$ in each iteration of the above algorithm. Therefore, $O(\log n)$ iterations suffice to obtain an MST. For the rest of the proof, we focus on the implementation detail of this algorithm, and we will show that $\poly(\log n)$ invocations of the routing algorithm of  \cref{thm-routing-general} with $L=O(1)$ are enough. In subsequent discussion, we write $\tau_0 = \poly\left(\phi^{-1}\right) \cdot n^{o(1)}$ to denote the round complexity of the routing algorithm of  \cref{thm-routing-general} with $L=O(1)$.

\paragraph{Maintenance of low-diameter Steiner trees.} During the algorithm, for each tree $T_i$ in the current forest $F$, we maintain a low-diameter Steiner tree $\Ts_i$ such that the leaf vertices of  $\Ts_i$ are exactly the set of vertices $V_i$ of $T_i$. Note that the diameter of $T_i$ can be very large, and so we need a low-diameter Steiner tree $\Ts_i$ to enable efficient communication between the vertices in $T_i$. 
\begin{itemize}
    \item The Steiner tree $\Ts_i$  is rooted, and denote $r_i$ the root of  $\Ts_i$.
    \item Each vertex $v$ in $\Ts_i$ has at most two children, and $v$ knows the IDs of its children. Each vertex $v \neq r_i$ in $\Ts_i$ also knows the ID of its parent. 
    \item We say that a vertex $v$ in $\Ts_i$ is in layer $j$ if $\dist_{\Ts_i}(v, r_i) = j$. Define $\ell_i = \max_{v \in \Ts_i} \dist(v, r_i)$ as highest layer number. We assume that each vertex in  $\Ts_i$ knows its layer number  and $\ell_i$.
    \item Each edge in   $\Ts_i$ does not need to be an edge in the underlying graph $G$. We use the routing algorithm of \cref{thm-routing-general} for vertices in $\Ts_i$ to communicate.
    \item We require that  $\Ts_i$  contains only the vertices in $V_i$. We allow each vertex $v \in V_i$ to correspond to multiple vertices in $\Ts_i$, but they have to belong to different layers.
\end{itemize}

Given this implementation, each tree $T_i$ can finds its minimum weight  outgoing edge $e_i$ in $O(\ell_i \tau_0)$ rounds by a bottom-up information gathering. In the end, all vertices in $V_i$ knows $e_i$.

Intuitively, if the algorithm under consideration is sufficiently simple, then one round of $G^\ast$ can be simulated by $O(\ell^\ast \tau_0)$ rounds in $G$, where $\ell^\ast = \max_{1 \leq i \leq k} \ell_i$, and $k$ is the number of trees $T_1, T_2, \ldots, T_k$ in the current forest $F$.

\begin{lemma}[Finding an independent set]\label{lem-find-indset}
An independent set $I^\ast$ of $G^\ast$ with $\sum_{T_i \in I^\ast} \indeg(T_i) \geq |E^\ast|/3$ can be found in $O(D + \ell^\ast \tau_0 \log n)$ rounds deterministically.
\end{lemma}
\begin{proof}
A proper 3-vertex coloring of a tree can be found in $O(\log n)$ rounds deterministically~\cite{BarenboimE10}. We simulate this algorithm on $G^\ast$, and the simulation takes $O(\ell^\ast \tau_0 \log n)$ rounds on $G$. After that, we go over each color class $I_1$, $I_2$, and $I_3$ to calculate $\sum_{T_i \in I_j} \indeg(T_i)$ for each $1 \leq j \leq 3$ in $O(D)$ rounds. One $I_j$ of these three sets must have $\sum_{T_i \in I_j} \indeg(T_i) \geq |E^\ast|/3$. We set $I^\ast = I_j$.
\end{proof}

Given  $I^\ast$, we know the set of edges that we need to add to $F$, but we still need to update the Steiner trees. We focus on one tree $T_i \in I^\ast$, and denote $\TT_i$ as the set of trees $T_j$ such that $T_j \rightarrow T_i$. We need to merge the Steiner tree $\Ts_i$ of $T_i$ and the Steiner tree $\Ts_j$ of all $T_j \in \TT$ together into a new Steiner tree.

\begin{lemma}[Merging the Steiner trees]\label{lem-merge-trees}
For each $T_i \in I^\ast$ in parallel, we can merge the Steiner tree $\Ts_i$ of $T_i$ and the Steiner trees $\Ts_j$ of all $T_j \in \TT_i$ into a new Steiner tree in $O(\ell^\ast \tau_0 \log n)$ rounds. Moreover, the height of the new Steiner tree is at most $\ell^\ast + O(\log n)$.
\end{lemma}
\begin{proof}
Recall that $V_i$ denotes the vertex set of $T_i$, and $V_i$ is identical to the leaf vertices of $\Ts_i$. We assume that at the beginning of the algorithm, each leaf $v$ of $\Ts_i$ holds a list $L_v$ indicating the set of all $\ID(r_j)$, for each edge $e_j$ added to the forest that are incident to $v$. Note that for each $T_j \in \TT_i$, $r_j$ belongs to the list $L_v$ for the unique vertex $v \in V_i$ incident to the edge $e_j$. If $v  \in \Ts_i$ it not a leaf, then we assume $L(v) = \emptyset$ initially.

Consider the following algorithm based on a bottom-up traversal of $T_i$. 
From $j = \ell_i$ to $j = 0$, all vertices $v$ at layer $j$ of $\Ts_i$ do the following. It organizes the elements in its set $L_v$ into pairs. If $|L_v|$ is odd, then there will be one leftover element, then $v$ sends it to its parent $u$ to have it added to the list $L_u$, unless $v$ itself is the root $r_i$ of $\Ts_i$. For each pair $(\ID(r_x), \ID(r_y))$,  $v$ inform the two vertices $r_x$ and $r_y$ to ask them to merge $\Ts_x$ and $\Ts_y$ into a new tree $T'$ by  selecting any one of $r' \in \{r_x, r_y\}$ to be the new root, and  adding the two new edges $\{r, r_x\}$ and  $\{r, r_y\}$. Note that we allow a vertex to appear multiple times in a Steiner tree, so long as all of its appearances are in different layers. After processing all pairs, the vertex $v$ resets its set $L_v$ as the set of IDs of the roots of the merged trees. 

It is clear if $|\TT_i|$ is an even number, then all trees are merged, and if  $|\TT_i|$ is an odd number, then all trees except one of them are merged. Therefore, if we continue this process with the new $L$-sets, then we are done merging all trees in $\TT_i$ in $O(\log n)$ iterations. In the end, we merge $T_i$ with the tree resulting from combining all trees in $\TT_i$. Overall, the algorithm costs  $O(\ell_i \log n) = O(\ell^\ast \log n)$ rounds, and the final tree resulting from merging has height at most $\ell^\ast + O(\log n)$, since the number of iterations is $O(\log n)$.
\end{proof}

Since the overall algorithm has $O(\log n)$ iterations, \cref{lem-merge-trees} implies that the height of all the Steiner trees in the algorithm is at most $O(\log^2 n)$. Therefore, \cref{lem-find-indset} costs $O(D +   \tau_0 \log^3 n)$ rounds and \cref{lem-merge-trees} costs $O(\tau_0 \log^3 n)$ rounds in each iteration. Recall that the diameter of a graph with conductance $\phi$ is at most $D = O\left(\phi^{-2} \log n \right)$.
Hence the overall round complexity is $O(D \log n +   \tau_0 \log^4 n) = \poly\left( \phi^{-1} \right) \cdot n^{o(1)}$. We conclude the proof of \cref{thm-mst}. 
\section{Conclusions and Open Questions}\label{sect:conclusion}
In this paper, we give the first subpolynomial-round deterministic distributed algorithms for expander decomposition and routing, and we also give the first polylogarithmic-round randomized distributed algorithms for expander decomposition.

The main obstacle that we overcome is the lack of efficient distributed algorithms for expander trimming and expander pruning.
For the randomized setting, we develop a new technique of extracting an expander from a near-expander resulting from the cut-matching game, without using expander trimming. To use this technique, it is crucial that the embedding of the matchings has small dilation. For the deterministic setting, we carry out the  simultaneous execution of KKOV cut-matching games using the ``non-stop'' style of~\cite{RST14}. This allows us to avoid adding fake edges to the graph as in~\cite{chuzhoy2019deterministic}.

We believe that our end-results and the techniques therein are of interest beyond the $\CONGEST$ model of distributed computing. Since we do not abuse the unlimited local computation power in the definition of  $\CONGEST$, our expander decomposition algorithms can be adapted to PRAM and the \emph{massively parallel computation} (MPC) model~\cite{KSV10} by a straightforward simulation. Note that PRAM algorithms can be simulated in the MPC model efficiently with the same round complexity~\cite{GSZ11,KSV10}.

Our deterministic distributed expander decomposition algorithm implies that an   $(\epsilon,\phi)$-expander decomposition of an $m$-edge $n$-vertex graph %$G=(V,E)$ 
with  $\phi =  \poly(\epsilon) 2^{-O(\sqrt{\log n \log \log n})}$ can be computed in $O(m) \cdot \poly(\epsilon^{-1}) 2^{O(\sqrt{\log n \log \log n})}$ work and   $\poly(\epsilon^{-1}) 2^{O(\sqrt{\log n \log \log n})}$ depth deterministically in PRAM. This improves the tradeoff between conductance and  time complexity of the  deterministic \emph{sequential} algorithm of~\cite{chuzhoy2019deterministic}, although the expander decomposition of~\cite{chuzhoy2019deterministic} is stronger in that each expander in the decomposition is a \emph{vertex-induced} subgraph.  %This opens up the possibility of  \emph{parallelizing} existing algorithms that are based on expander decompositions.

\paragraph{Open questions.}
The ultimate goal of this research is to give deterministic polylogarithmic-round distributed algorithms for both expander decomposition and routing.

Currently we only have \emph{randomized} polylogarithmic-round algorithm for expander decomposition, and it is still open whether expander routing can be solved in $\poly(\phi^{-1}, \log n)$ rounds even randomness is allowed. 

A shortcoming of our expander decomposition algorithms  is that each expander in the decomposition is not a vertex-induced subgraph. 
To transform any expander decomposition to an expander decomposition where each expander is a vertex-induced subgraph, it suffices to use  \emph{expander trimming}~\cite{SaranurakW19}. It is an open question whether expander trimming can be solved efficiently in $\CONGEST$.

A key ingredient in our deterministic algorithms is the nearly maximal flow algorithm of~\cite{GPV93}, which results in a small number of leftover vertices. It is an open question whether we can do it without leftover vertices. Specifically, given any two subsets $S \subseteq V$ and $T \subseteq V$, the goal is to find a maximal set of vertex-disjoint paths of length at most $d$. Can we solve this problem deterministically in $\poly(d, \log n)$ rounds?
An affirmative answer to this question will simplify our deterministic expander decomposition and routing quite a bit.
 
A common version of the low-diameter decomposition is a partition %of the vertex set 
$V = V_1 \cup V_2 \cup \cdots \cup V_x$ such that the number of inter-cluster edges is at most $\beta|E|$, and the (strong) diameter of $G[V_i]$ is $\poly(\beta^{-1}, \log n)$ for each $1 \leq i \leq x$. Such a decomposition can be computed in $\poly(\beta^{-1}, \log n)$ rounds deterministically in the $\LOCAL$ model, and it is open whether it can also be computed in $\poly(\beta^{-1}, \log n)$ rounds in $\CONGEST$ deterministically~\cite{RozhonG20}. An affirmative answer to this question will also simplify the proofs in this paper as this allows us to get rid of Steiner trees.

Our expander routing algorithm uses the  load balancing algorithm of Ghosh~et~al.~\cite{GhoshLMMPRRTZ99}. To the best of our knowledge, this is the first time this technique is applied to the $\CONGEST$ model of distributed computing.
It will be interesting to see more applications of this technique. In particular, does it give us some kind of distributed local graph clustering like PageRank~\cite{AndersenCL08,spielman2004nearly}?

%\begin{itemize}
%    \item Deterministic maximal flow (would simplify deterministic expander decomposition and routing quite a bit)
%    \item Deterministic triangle detection/listing in $O(n^{1/3})$ time. The main reason is that we do not know a deterministic algorithm for the following task:
%    Given a graph $G=(V,E)$ such that vertices have average and min degree around $n^{1/3}$, find a partition of $V$ into $V_1,\dots,V_{n^{1/3}}$ such that for any $i,j,k$ the number of edges in $G[V_i \cup V_j \cup V_k]$ is at most $\tilde O(n^{2/3})$. For randomized algorithms, this task can be solved by simply assigning each vertex to a random part in the partition.
%    \item Given that we have polylog randomized expander decomposition now, if one can solve expander routing in polylog rounds, then we can simulate several PRAM algorithm in CONGEST with polylog overhead. This would give a near optimal short cycle decomposition as well.   
%    \item polylog deterministic expander decomposition and routing are the ultimate goal. 
%    \item More application of deterministic load balancing. Does it gives some kind of distributed local clustering like PageRank? Is it better?
%\end{itemize}

%\section*{Acknowledgment}
%We thank Seth Pettie for very useful discussion.
%
\newpage
\appendix
\section*{\LARGE Appendix}
\vspace{0.5cm}
\section{Communication Primitives of Steiner Trees}\label{sect-steiner}

In this section, we provide some basic communication primitives of Steiner trees. We consider the following setup. Let $G=(V, E)$ be any $n$-vertex graph, and $T$ is a Steiner tree of diameter $D$ whose leaf vertices are $V$. Each $v \in V$ holds a number $t_v$ that can be represented in $O(\log n)$ bits. Note that all lemmas below also work if $T$ is a spanning tree of $V$, as we can pretend that $T$ is a Steiner tree by having each vertex locally simulate a leaf corresponding to itself. 

The algorithm of \cref{lem-basic} is a simple bottom-up information gathering along the Steiner tree. The algorithm of \cref{lem-det-bsearch} is a distributed implementation of the binary search, where each iteration is implemented in $O(D)$ rounds using the Steiner tree.
The algorithm of \cref{lem-rand-bsearch} implements a version of the binary search where the middle element is selected uniform at random, allowing us to handle \emph{arbitrary} numbers.

\begin{lemma}[Information gathering]\label{lem-basic}
Suppose each vertex $v$ holds a number $t_v$.
There is a deterministic algorithm that computes $\sum_{v \in V} t_v$, $\max_{v \in V} t_v$, and $\min_{v \in V} t_v$.
The algorithm terminates in $O(D)$ rounds.
By pipelining, we can solve $k$ instances of this problem in $O(D + k)$ rounds.
\end{lemma}

\begin{lemma}[Deterministic binary search]\label{lem-det-bsearch}
Suppose each vertex $v$ holds an integer $t_v \in [N]$.
Let $1 \leq s \leq n$ be any integer.
There is a deterministic algorithm that finds the subset $U \subseteq V$ consisting of any $s$ vertices in $V$ with the highest $t$-values, breaking tie arbitrarily. The algorithm terminates in $O(D \log N)$ rounds.
By pipelining, we can solve $k$ instances of this problem in $O((D + k) \log N)$ rounds.
\end{lemma}

%For the case that each vertex holds an arbitrary  number, there is an efficient randomized algorithm for this task, see~\cite{ChangPZ19} for the proof of \cref{lem-rand-bsearch}.

\begin{lemma}[Randomized binary search~\cite{ChangPZ19}]\label{lem-rand-bsearch}
Suppose each vertex $v$ holds a number $t_v$.
Let $1 \leq k \leq n$ be any integer.
There is a deterministic algorithm that finds the subset $U \subseteq V$ consisting of any $k$ vertices in $V$ with the highest $t$-values, breaking tie arbitrarily. The algorithm terminates in $O(D \log n)$ rounds with high probability.
\end{lemma}

\begin{lemma}[Balanced partition]\label{lem-bal-partition}
Suppose each vertex $v$ holds a number $t_v$ such that $\max_{v \in V} t_v \leq (1/2) \sum_{v \in V} t_v$. There is a deterministic algorithm that computes a subset $U \subseteq V$ such that $(1/3) \sum_{v \in V} t_v \leq \sum_{v \in U} t_v \leq (1/2) \sum_{v \in V} t_v$.
The algorithm terminates in $O(D)$ rounds.
\end{lemma}
\begin{proof}
Let $u = \arg \max_{v \in V}t_v$. If $t_u > (1/3) \sum_{v \in V} t_v$, then we can simply return $U = \{u\}$.
In what follows, we assume $\max_{v \in V} t_v \leq (1/3) \sum_{v \in V} t_v$.

Let $T$ be the underlying Steiner tree. Make $T$ a rooted tree. For each $v$ in $T$, let $T_v$ be the subtree rooted at $v$, and let $n_v$ to be the summation of $t_u$ over all vertices $u$ in $T_v$. Denote $M = \sum_{v \in V} t_v$. By a simple bottom-up information gathering, in $O(D)$ rounds, we can let each vertex $v$ to calculate $n_v$, and make $M$ global knowledge.

There exists a vertex $v^\ast$ in $T$ such that (1) $n_{v^\ast} \geq M/3$, and (2) $n_u < M/3$ for each child $u$ of $v^\ast$. Let $u_1, u_2, \ldots, u_k$ be the children of $v^\ast$.
We select $r$ such that $M/3 \leq \sum_{1 \leq l \leq r} n_{u_l} < 2M/3$.  If $\sum_{1 \leq l \leq r} n_{u_l} \leq M/2$, then we can select $U$ to be the set of leaf vertices in $T_{u_1} \cup T_{u_2} \cup \cdots \cup T_{u_r}$. Otherwise, we have $M/2 < \sum_{1 \leq l \leq r} n_{u_l} \leq M/3$, and so we can select $U$ to be the set of leaf vertices not in $T_{u_1} \cup T_{u_2} \cup \cdots \cup T_{u_r}$.
\end{proof}

\begin{lemma}
[Partition]\label{lem-partition}
Suppose each vertex $v$ holds an integer $t_v \geq 0$ indicating the number of tokens it has. 
Given a vector $(p_1, p_2, \ldots, p_k)$  of non-negative numbers with $\sum_{1 \leq i \leq k}p_i = 1$,
there is a deterministic algorithm that partition the tokens into $k$ subsets $S_1, S_2, \ldots, S_k$ in such a way that \[\left(p_i \cdot \sum_{v \in V} t_v\right) - 1 < |S_i| < \left(p_i \cdot \sum_{v \in V} t_v\right) + 1,\]  and each vertex knows the number of its tokens in each subsets.
The algorithm terminates in $O(kD)$ rounds.
\end{lemma}
\begin{proof}
It suffices to consider the following simpler task.
Each vertex $v$ holds an integer $t_v \geq 0$.
We are given an integer $0 \leq K \leq \sum_{v \in V} t_v$, and the goal is to let each vertex computes an integer $s_v \geq 0$ in such a way that $s_v \leq t_v$ and $\sum_{v \in V} s_v = K$. If we can solve this problem in $O(D)$ rounds, then we can solve the problem given in the lemma in  $O(kD)$ rounds.

Similar to the proof of \cref{lem-bal-partition},
let $T$ be the underlying Steiner tree, and  make $T$ a rooted tree. For each $v$ in $T$, let $T_v$ be the subtree rooted at $v$, and let $n_v$ to be the summation of $t_u$ over all vertices $u$ in $T_v$.  By a simple bottom-up information gathering, in $O(D)$ rounds, we can let each vertex $v$  calculate $n_v$, and make $M = \sum_{v \in V} t_v$ global knowledge. Pick $v$ to be any vertex such that $n_v \geq K$ and $n_u < K$ for all children $u$ of $v$. 
Let $u_1, u_2, \ldots, u_x$ be its children. Pick $i$ to be the smallest index such that $\sum_{1 \leq j \leq i}n_{u_j} \geq K$.
All vertices $w \in V$ in the subtrees rooted at $u_1, u_2, \ldots, u_{i-1}$ set $s_w = t_w$. Then we recurse on the vertex $u_i$ by replacing $K$ by $K - \sum_{1 \leq j \leq i-1}n_{u_j}$. For the base case $v \in V$ is a leaf, it sets $s_v = K$.
\end{proof}

\begin{lemma}[Uniform sampling~\cite{ChangPZ19}] \label{lem-uniform-sample}
There are $s$ balls and $n$ bins, where each bin is associated with a vertex $v \in V$. Each ball is independently placed into a bin uniformly at random. This process can be simulated in $O(D)$ rounds, and in the end each vertex knows the number of balls in its bin. 
By pipelining, we can solve $k$ instances of this problem in $O(D + k)$ rounds.
\end{lemma}
\begin{proof}
Root the Steiner tree $T$ arbitrarily, and let each vertex $v$ in the Steiner tree calculates the number $n_v$ of vertices in $V$ within the subtree $T_v$ rooted at $v$. This can be done by a bottom-up traversal in $O(D)$ rounds. Next, we do a top-down traversal to distribute the balls to the vertices as follows. Initially, there are $s$ balls in the root. Whenever a intermediate vertex $v$ receives a ball, it sends the ball to its child $u$ with probability $n_u / n_v$. Note that $v$ only needs to inform each of its children the number of balls it receives, and so this process can be done in $O(D)$ rounds.
\end{proof}
\section{Low-diameter Decomposition}\label{sect-low-diam-decomp}

A common version of the \emph{low-diameter decomposition} of $G=(V,E)$ is a partition of the vertex set $V$ into clusters of small diameter such that the number of inter-cluster edges is small. It is well-known~\cite{LinialS93,miller2013parallel} that for any given parameter $0 < \beta < 1$, there is an $O(\beta^{-1} \log n)$-round randomized algorithm that computes a low diameter decomposition such that the diameter of each cluster is $O(\beta^{-1} \log n)$, and the \emph{expected} number of inter-cluster edges is $\beta|E|$.

\begin{lemma}[Randomized low-diameter decomposition~\cite{LinialS93,miller2013parallel}]
\label{thm-low-diam-decomp-rand}
Given a parameter $0 < \beta < 1$, there is a randomized algorithm that decomposes the vertex set $V$ into clusters $V = V_1 \cup V_2 \cup \cdots \cup V_x$ in $O(\beta^{-1}  \log n)$ rounds meeting the following conditions.
\begin{itemize}
   \item The expected number of inter-cluster edges is at most $\beta|E|$.
    \item The diameter of the subgraph $G[V_i]$ induced by each cluster $V_i$ is at most $O(\beta^{-1} \log n)$.
\end{itemize}
Moreover, we can make the guarantee on the number of inter-cluster edges to hold with high probability by increasing the round complexity to $O(D + \beta^{-1}  \log n)$.
\end{lemma}

Rozho\v{n} and Ghaffari~\cite{RozhonG20} recently obtained a $\poly\left(\beta^{-1}, \log n\right)$-round \emph{deterministic} algorithm that   achieves a similar result. The clusters returned by the Rozhon-Ghaffari algorithm do not necessarily have small diameters, but they are associated with low-diameter \emph{Steiner trees} that can be simultaneously embedded into the underlying graph with congestion $O(\log n)$.

\begin{lemma}[Deterministic low-diameter decomposition~\cite{RozhonG20}]
\label{thm-low-diam-decomp}
Given a parameter $0 < \beta < 1$, there is a deterministic algorithm that decomposes the vertex set $V$ into clusters $V = V_1 \cup V_2 \cup \cdots \cup V_x$ in $O(\beta^{-2} \log^6 n)$ rounds meeting the following conditions.
\begin{itemize}
   \item The number of inter-cluster edges is at most $\beta|E|$.
    \item Each cluster $V_i$ is associated with a Steiner tree $T_i$ such that the leaf vertices of $T_i$ is $V_i$. The diameter of $T_i$ is   $O(\beta^{-1} \log^3 n)$. Each edge $e \in E$ belongs to at most $O(\log n)$  Steiner trees.
\end{itemize}
\end{lemma}

\begin{proof}
\cref{thm-low-diam-decomp} is a result of a minor modification of the original algorithm of 
Rozho\v{n} and Ghaffari~\cite{RozhonG20}.
We present a brief description of the algorithm of~\cite{RozhonG20}, with the small modification that we need. 
Suppose each vertex is initially equipped with a distinct $\ID$ of $b = O(\log n)$ bits.
At the beginning, each vertex $v$ hosts a cluster $C = \{v\}$ whose identifier $\ID(C)$ is initialized to $\ID(v)$. 
This trivial clustering already satisfies the diameter requirement, but it does not meet the requirement on the number of inter-cluster edges.

The algorithm works in $b$ phases. In each phase, the clustering will be updated, and at most $\beta / b$ fraction of edges will be removed from the graph. As there are $b$ phases, the total number of removed edges is at most $\beta|E|$. The induction hypothesis specifies that at the end of the $i$th phase, for any two clusters $C_1$ and $C_2$ such that  $\ID(C_1)$ and $\ID(C_2)$ have different $i$-bit suffix, there is no edge connecting $C_1$ and $C_2$.

The goal of the $i$th phase is to achieve the following. 
For any \emph{fixed} $(i-1)$-bit suffix $Y$, separate all clusters whose ID is of the form 
$(\cdots 0 Y)$ (called \emph{blue} clusters) 
from those whose ID is of the form 
$(\cdots 1 Y)$ (called \emph{red} clusters).
The $i$th phase of the algorithm consists 
of $k = (b/\beta) \cdot O(\log n)$ iterations 
in which 
blue clusters may acquire new members and red clusters
lose members.

In each such iteration, each vertex in a red cluster that is adjacent to one or more blue clusters requests to join 
any one of the blue clusters. 
Now consider a blue cluster $C$. 
Define $E_1$ to be the set of edges inside $C$
and let $E_2$ be the set of new edges that would 
be added to $C$ if \emph{all} join requests were accepted. There are two cases:
\begin{enumerate}
    \item if $|E_1|=0$ or $|E_2|/|E_1| > \beta / b$, then all of $C$'s join requests are accepted;
    \item otherwise, all edges in $E_2$ are removed from the graph.
\end{enumerate}
After $k = \log_{1+\beta/b} m = (b/\beta) \cdot O(\log n)$ iterations, all red clusters are separated from all blue clusters, since it is impossible for a blue cluster to enter Case~1 for all $k$ iterations.

Whenever a blue cluster $C$ enters Case~1, 
each new member $v$ of $C$ is attached to $C$'s Steiner tree by including an edge $\{v,u\}\in E_2$
joining it to an existing member $u$ of $C$.
 Therefore, the diameter of the final Steiner tree will be $O(k b) =  O\left(\beta^{-1} \log^3 n \right)$, and each edge belongs to at most $b = O(\log n)$ Steiner trees. 
 
 For the round complexity, there are $b = O(\log n)$ phases, and each phase consists of $k = (b/\beta) \cdot O(\log n) = O(\beta^{-1} \log^2 n )$ iterations.
 The round complexity for each iteration is linear in the diameter of the Steiner trees $O(k b) =  O\left(\beta^{-1} \log^3 n \right)$.
 Therefore, the total round complexity of the entire algorithm is $O\left( \beta^{-2} \log^{6} n \right)$.
 \end{proof}

 By spending an additional $O(D)$ rounds, we can improve the bound on the number of inter-cluster edges.

 \begin{lemma}[Modified deterministic low-diameter decomposition]
\label{thm-low-diam-decomp-mod}
Given a parameter $0 < \beta < 1$, there is a deterministic algorithm that decomposes the vertex set $V$ into clusters $V = V_1 \cup V_2 \cup \cdots \cup V_x$ in $O(D + \beta^{-2} \log^6 n)$ rounds meeting the following conditions.
\begin{itemize}
   \item The number of inter-cluster edges is at most $\beta \vol(V \setminus V_{i^\ast})$, where $V_{i^\ast}$ is a cluster with the highest volume.
    \item Each cluster $V_i$ is associated with a Steiner tree $T_i$ such that the leaf vertices of $T_i$ is $V_i$. The diameter of $T_i$ is   $O(\beta^{-1} \log^3 n)$. Each edge $e \in E$ belongs to at most $O(\log n)$  Steiner trees.
\end{itemize}
\end{lemma}
\begin{proof}
Use \cref{thm-low-diam-decomp} to find a decomposition $V = V_1 \cup V_2 \cup \cdots \cup V_x$. Each part $V_i$ locally computes $\vol(V_i)$ in parallel, and it costs $O(\beta^{-1} \log^4 n)$ rounds, as the diameter of  $T_i$ is   $O(\beta^{-1} \log^3 n)$, and each edge $e \in E$ belongs to at most $O(\log n)$ Steiner trees. Then we use \cref{lem-basic} to calculate the index $i^\ast$ such that $V_{i^\ast}$ is a cluster with the highest volume in $O(D)$ rounds.
If $\vol(V_{i^\ast}) < \vol(V)/2$, then the number of inter-cluster edges is already at most \[\beta|E| = \beta \vol(V)/2 \leq \beta(\vol(V) - \vol(V_{i^\ast})) = \beta \vol(V \setminus V_{i^\ast}).\]
For the rest of the proof, we assume $\vol(V_{i^\ast}) \geq \vol(V)/2$.
For each $j$, we define the set \[S_j = \{ u \in V \ | \ \dist(u, V_{i^\ast}) \leq j\}.\] We claim that one of the following conditions is met.
\begin{enumerate}
    \item $S_j = V$ for some  $0 \leq j \leq O(\beta^{-1} \log n)$.
    \item $|E(S_j, V \setminus S_j)| \leq (\beta/2) \vol(V \setminus S_j)$ for some $0 \leq j \leq O(\beta^{-1} \log n)$.
\end{enumerate}
The reason is as follows. If we have  $|E(S_j, V \setminus S_j)| > (\beta/2) \vol(V \setminus S_j)$, then we must have $\vol(V \setminus S_{j+1}) \leq (1 - \beta/2) \vol(V \setminus S_{j})$. Therefore, if the second condition is not met, then the first condition must be met.

If the first condition is met, then the underlying graph $G$ itself has diameter $O(\beta^{-1} \log^3 n) + O(\beta^{-1} \log n)$, and so we can simply return the trivial decomposition where everyone belongs to the same cluster.

Now suppose the second condition is met. Let $0 \leq j \leq O(\beta^{-1} \log n)$ be an index such that $|E(S_j, V \setminus S_j)| \leq (\beta/2) \vol(V \setminus S_j)$. 
Calculation of the index $j$ costs $O(D + \beta^{-1} \log n)$ rounds using \cref{lem-basic}.
Run the algorithm of \cref{thm-low-diam-decomp} on $G[V \setminus S_j]$ to obtain a decomposition $V \setminus S_j = V_1' \cup V_2' \cup \cdots \cup V_x'$, where the number of inter-cluster edges is at most $(\beta/2) \vol(V \setminus S_j)$. It is clear that the decomposition $V = S_j  \cup V_1' \cup V_2' \cup \cdots \cup V_x'$ satisfies all the requirements, as the number of inter-cluster edges is at most $\beta \vol(V \setminus S_j)$.
\end{proof}

\section{Conductance and Sparsity}\label{sect-conductance}

In this section, we provide tools for analyzing the conductance and sparsity of graphs and cuts. %that are used in the paper. 
Many of the tools in this section are  from~\cite{chuzhoy2019deterministic}, with small modifications in some cases.

\subsection{Inner and Outer Sparsity}\label{sect-in-out-conductance}

%We review a tool  from~\cite{chuzhoy2019deterministic} that allow us to focus on bounded-degree graphs instead of general graphs.

 Note that the bound $\lambda < 1/2$ in the following lemma is arbitrary, and $1/2$ can be replaced by any constant in the range $(0, 1)$.

\begin{lemma}[Inner and outer sparsity]\label{lem-conductance-v}
Consider a graph $G=(V,E)$  and a partition $\VV = \{V_1, V_2, \ldots, V_k\}$ of $V$. 
Define $\tilde{\Delta} = \max_{1 \leq i \leq k} \max_{v \in V_i} |N(v) \setminus V_i|$.
The following holds.
\begin{itemize}
    \item $\Phiv(G) = \Omega(\tilde{\Delta}^{-1} \cdot \vPhiIn{G}{\VV} \cdot \vPhiOut{G}{\VV})$.
    \item Let $C$ be any cut of $G$ with $0 < |C| \leq |V|/2$ and $\Phiv(C) \leq \psi = \lambda  \vPhiIn{G}{\VV}$, where $0 < \lambda < 1/2$.
    Define the cut $C'$ as the union of all $V_i \in \VV$ such that $|C \cap V_i| \geq (1/2) |V_i|$. Then
\begin{itemize}
    \item $\Phiv(C') = O(\tilde{\Delta} \psi / \vPhiIn{G}{\VV}) = O( \lambda\tilde{\Delta}|C|)$ and
    \item $||C'| - |C|| \leq \lambda |C|$.
\end{itemize}
\end{itemize}
\end{lemma}
\begin{proof}
Consider any $C \subseteq V$ in $G$ with $0 < |C| \leq |V|/2$ and $\Phiv(C) \leq \psi = \lambda  \vPhiIn{G}{\VV}$, where $0 < \lambda < 1/2$.
Define $C'$ as the result of applying the following operations to $C$. From $i = 1$ to $i = k$, do the following. If $|C \cap V_i| \geq (1/2) |V_i|$, then update $C \leftarrow C \cup V_i$; otherwise update $C \leftarrow C \setminus V_i$. The resulting cut $C'$ is identical to the one in the lemma statement.
Consider the $i$th iteration.
\begin{itemize}
    \item The number of edges removed from $\partial(C)$ in this iteration is at least $x = |E(C \cap V_i, V_i \setminus C)|$.
    \item Suppose $|C \cap V_i| \geq (1/2) |V_i|$. Then $V_i \setminus C$ is added to $C$ in this iteration. This results in at most $|V_i \setminus C| \leq   x / \vPhiIn{G}{\VV}$ new vertices added to $C$, and at most
    $\tilde{\Delta} \cdot |V_i \setminus C| \leq \tilde{\Delta} \cdot  x / \vPhiIn{G}{\VV}$ new edges added to $\partial(C)$.
    \item Suppose $|C \cap V_i| \leq (1/2) |V_i|$. Then $V_i \cap C$ is removed from $C$ in this iteration. This results in at most $|V_i \cap C| \leq  x / \vPhiIn{G}{\VV}$ vertices removed from $C$, and at most   
    $\tilde{\Delta} \cdot |V_i \cap C| \leq \tilde{\Delta} \cdot  x / \vPhiIn{G}{\VV}$ new edges added to $\partial(C)$.   
\end{itemize}
Therefore, we have the following two bounds.
\begin{itemize}
    \item $|\partial(C')| \leq |\partial(C)| \cdot   \tilde{\Delta} / \vPhiIn{G}{\VV} \leq \tilde{\Delta} \psi |C|/ \vPhiIn{G}{\VV}  =  \lambda \tilde{\Delta} |C|$.
    \item $||C'| - |C|| \leq |\partial(C)| / \vPhiIn{G}{\VV} \leq \psi |C| / \vPhiIn{G}{\VV} =  \lambda |C|$.
\end{itemize}
The sparsity $\Phiv(C')$  of $C'$ can be upper bounded as follows. We use the assumption that $\lambda < 1/2$.
\begin{itemize}
    \item For the case $|C'| \leq |V|/2$, we have \[\Phiv(C') = \frac{|\partial(C')|}{|C'|} \leq \frac{\tilde{\Delta} \psi |C|/ \vPhiIn{G}{\VV}}{(1-\lambda)|C|} = O(\tilde{\Delta} \psi / \vPhiIn{G}{\VV}).\]
    \item For the case $|C'| > |V|/2$, we have %$(1+\lambda)|C| > |V|/2$, and so $|C| > |V|/(2(1+\lambda)) > |V|/3$. This implies 
    $|V \setminus C'| \geq |V \setminus C| - \lambda|C| > |V|/2 - |C|/2 \geq |V|/2 - |V|/4 = |V|/4$. We can now upper bound the sparsity $\Phiv(C')$ as \[\Phiv(C') = \frac{|\partial(C')|}{|V \setminus C'|} \leq \frac{\tilde{\Delta} \psi |C|/ \vPhiIn{G}{\VV}}{|V|/4} \leq 
    \frac{\tilde{\Delta} \psi |V|/ (2\vPhiIn{G}{\VV})}{|V|/4}
    = O(\tilde{\Delta} \psi / \vPhiIn{G}{\VV}).\]
\end{itemize}
To show that $\Phiv(G) = \Omega(\tilde{\Delta}^{-1} \cdot \vPhiIn{G}{\VV} \cdot \vPhiOut{G}{\VV})$, we select $C$ to be a sparsest cut of $G$, i.e., $\Phiv(G) = \Phiv(C)$. If $\Phiv(C) > \psi = \lambda  \vPhiIn{G}{\VV}$ for $\lambda = 1/3$, then we are done already. Otherwise we apply the above analysis, which shows that 
$\Phiv(C') = O(\tilde{\Delta}\Phiv(C) / \vPhiIn{G}{\VV})$. Since $\Phiv(G) = \Phiv(C)$ and $\Phiv(C') \geq \vPhiOut{G}{\VV}$, we infer that $\Phiv(G) = \Omega(\tilde{\Delta}^{-1} \cdot \vPhiIn{G}{\VV} \cdot \vPhiOut{G}{\VV})$. 
\end{proof}

Roughly speaking, \cref{lem-conductance-v} shows that any sufficiently sparse cut $C$ of $G$ can be transformed into another cut $C'$  
that respects $\VV$ by losing an $O(1/ \PhiIn{G}{\VV})$ factor in sparsity (assuming $\tilde{\Delta} = O(1)$) %and $\epsilon^{-1} = O(1)$) 
and an $1 + o(1)$ factor in balance (assuming $\lambda = o(1)$).

\subsection{Expander Split}\label{sect-expander-split}

   In the construction of $\Gsp$,  each edge $e=\{u,v\} \in E$ is associated with a vertex $u' \in X_u \subseteq \Vsp$ and a vertex  $v' \in X_v \subseteq \Vsp$. Each vertex in $\Vsp$ is associated with exactly one edge in $E$. From now on, we write $\xi(e) = \{u', v'\}$ for each edge $e=\{u,v\} \in E$. Note that $\xi(e)$ is also an edge of $\Gsp$.

We list some crucial properties of the expander split graph $\Gsp$ of $G$ as a lemma.

\begin{lemma}[Properties of $\Gsp$]\label{lem-expander-split-property}
For any graph $G=(V,E)$ an its expander split graph $\Gsp = (\Vsp, \Esp)$, the following holds.
\begin{enumerate}
    \item $\Gsp$ has constant maximum degree.
    \item $\vol(V) = 2|E| = |\Vsp|$.
    \item For any cut $C$ in $G$, its corresponding cut
$C' = \bigcup_{v \in C} X_v$  in $\Gsp$ respects the partition  $\VV=\{ X_v \ | \ v \in V \}$, and it satisfies $\vol(C) = |C'|$,
$|\partial(C)| = |\partial(C')|$, and so $\Phie(C) = \Phiv(C')$. 
    \item $\Phie(G) = \vPhiOut{\Gsp}{\VV} \geq \Phiv(\Gsp)$.
    \item $\Phie(G) = \Theta(\Phiv(\Gsp))$.
\end{enumerate}
\end{lemma}
\begin{proof}
All statements are straightforward, except that to prove the last statement $\Phie(G) = \vPhiOut{\Gsp}{\VV} = \Theta(\Phiv(\Gsp))$, we need to apply \cref{lem-conductance-v} to  $\Gsp$ and  $\VV$. \cref{lem-conductance-v} guarantees that $\Phiv(\Gsp) = \Omega(\Delta^{-1} \cdot \vPhiIn{\Gsp}{\VV} \cdot \vPhiOut{\Gsp}{\VV}) = \Omega(\vPhiOut{\Gsp}{\VV})$, implying   $\vPhiOut{\Gsp}{\VV} = \Theta(\Phiv(\Gsp))$.
We use the fact that $\vPhiIn{\Gsp}{\VV} = \Omega(1)$ and the maximum degree $\Delta$ of $\Gsp$ is a constant.
\end{proof}

The following lemma extends the above lemma to subgraphs.

\begin{lemma}[Subgraphs of $G$ and $\Gsp$] \label{lem-conductance-aux-1}
Let $G=(V,E)$ be any graph, and let $\Gsp = (\Vsp, \Esp)$ be its expander split graph.
Let $G^\ast=(V^\ast, E^\ast)$ be any subgraph of $G$.
Define $W^\ast \subseteq \Vsp$ as the union of all vertices in $\xi(e)$ over all $e \in E^\ast$.
Let $C$ be any cut of $G^\ast$, and
define its corresponding cut $C'$ in $\Gsp[W^\ast]$ by 
$C' = \bigcup_{v \in C} X_v \cap W^\ast$. 
The following holds.
\begin{enumerate}
    \item $\vol_{G^\ast}(V^\ast) = 2|E^\ast| = |W^\ast|$.
    \item $\vol_{G^\ast}(C) = |C'|$.
    \item $|\partial_{G^\ast}(C)| = |\partial_{\Gsp[W^\ast]}(C')|$.
    \item $\Phie(G^\ast) = \vPhiOut{\Gsp[W^\ast]}{\VV^\ast} \geq \Phiv(\Gsp[W^\ast])$, where $\VV^\ast$ is $\VV=\{ X_v \ | \ v \in V \}$ restricted to $\Gsp[W^\ast]$.
\end{enumerate}
\end{lemma}
\begin{proof}
It is straightforward to see $\vol_{G^\ast}(V^\ast) = 2|E^\ast| = |W^\ast|$ by the definition of $W^\ast$. To see that $\vol_{G^\ast}(C) = |C'|$, we observe that $|C'| = \sum_{v \in C} |X_v \cap W^\ast| =  \sum_{v \in C} \deg_{G^\ast}(v) = \vol_{G^\ast}(C)$.

The statement $|\partial_{G^\ast}(C)| = |\partial_{\Gsp[W^\ast]}(C')|$ follows from the fact that $\xi$ gives  a bijective mapping between $\partial_{G^\ast}(C)$ and $\partial_{\Gsp[W^\ast]}(C')$.

%Consider any edge $e \in \partial(C')$.  Observe that  $e=\{u', v'\}$ for some $u' \in X_u$ and $v' \in X_v$ with $u \neq v$. 
%Thus, we have $|\partial(C')| = |\partial(C)|$, since $\xi$ gives  a bijective mapping between $\partial(C)$ and $\partial(C')$. 
Combining $\vol_{G^\ast}(C) = |C'|$ with $|\partial_{G^\ast}(C)| = |\partial_{\Gsp[W^\ast]}(C')|$, we have   $\Phie_{G^\ast}(C) = \Phiv_{\Gsp[W^\ast]}(C')$.
For any cut $C$ in $G^\ast$, its corresponding cut $C'$  in $\Gsp[W^\ast]$ respects $\VV^\ast$. Conversely, any cut  $C'$ in $\Gsp[W^\ast]$ respecting $\VV^\ast$ corresponds to a cut $C$ in $G^\ast$. Therefore, $\Phie(G^\ast) = \vPhiOut{\Gsp[W^\ast]}{\VV^\ast} \geq \Phiv(\Gsp[W^\ast])$.
\end{proof}

The following two lemmas show that sparse cuts and well-connected subgraphs in expander split $\Gsp$ can be transformed into sparse cuts and well-connected subgraphs in $G$.

%\cref{lem-conductance} implies the following. We use the fact that on bounded-degree graphs $\Phiv$ and $\Phie$ are within a constant factor.

\begin{lemma}[Sparse cuts in expander splits]\label{lem-expander-split-cut}
Let $G=(V,E)$ be any graph, and let $\Gsp = (\Vsp, \Esp)$ be its expander split graph simulated in the communication network $G$.
Given any  cut $C'$ in $\Gsp$ with $|C'| = \beta |\Vsp|$, where $0 < \beta \leq 1/2$, and
$\Phiv(C') = \psi$, in $O(D)$ rounds we can obtain a cut $C$ of $G$ with $\vol(C) = \Omega(\beta) \cdot \vol(V)$, $\vol(C) \leq \vol(V)/2$, and $\Phie(C) = O(\psi)$.
\end{lemma}
\begin{proof}
For the case $\psi > \epsilon$ for some small constant $\epsilon > 0$, we can simply pick $C$ as any balanced bipartition of $G$. Specifically, we can apply \cref{lem-bal-partition} with $t_v = \deg(v)$. This costs $O(D)$ rounds.

To prove this lemma, to suffices to find a cut $C''$  respecting   $\VV=\{ X_v \ | \ v \in V \}$ in $\Gsp$ such that 
\begin{itemize}
    \item $\Phiv(C'') = O(\psi)$,
    \item $|C''| = \Omega(\beta) \cdot |\Vsp|$, and
    \item $|C''| \leq |\Vsp|/2$.
\end{itemize}

 If such a cut $C''$ is given, then we can pick $C$ to be a cut in $G$ that corresponds to $C''$. Specifically, we  pick $C = \{ v \in V \ | \ X_v \subseteq C''\}$.  In view of \cref{lem-expander-split-property}, we have $\vol(C) / \vol(V) = |C''| / |\Vsp|$ and $\Phie(C) = \Phiv(C'')$. Hence $C$ satisfies all the requirements. %If the volume of $C$ exceeds $\vol(V)/2$, then we can pick $V \setminus C$ to satisfy the requirements in the lemma.

From now on, we assume that $\Phiv(C') = \psi \leq \epsilon$ for some small constant $\epsilon$, and we focus on finding the cut $C''$ described above. Applying \cref{lem-conductance-v} to the expander split graph $\Gsp$ with the partition $\VV=\{ X_v \ | \ v \in V \}$. We obtain a cut $C''$ of $\Gsp$ respecting $\VV$ such that $\Phiv(C'') = O(\psi)$ and $|C''| = \Omega(\beta) \cdot |\Vsp|$. Here we use the fact that $\epsilon$ is a sufficiently small constant so that we can have
$\Phiv(C') = \psi = \lambda  \vPhiIn{G}{\VV}$, for some small $0 < \lambda < 0.1$ that allows us to argue \[\Phiv(C'') = O(\tilde{\Delta}\psi / \vPhiIn{\Gsp}{\VV}) = O(\psi),\] as both $\tilde{\Delta}$ and $\vPhiIn{\Gsp}{\VV}$ are constants for $\Gsp$. Also, we have $||C'| - |C''|| \leq \lambda |C'| \leq 0.1 |C'|$, which implies \[|C''| \geq 0.9|C'| = \Omega(\beta) \cdot |\Vsp|.\]

%and also remember that the maximum degree of $\Gsp$ is a constant.

If it happens that $|C''| \leq |\Vsp|/2$, then we are done. Otherwise, we can use this cut $\Vsp \setminus C''$ to satisfy the requirement $|\Vsp \setminus C''| \leq |\Vsp|/2$. Note that we still have \[\Phiv(\Vsp \setminus C'') = \Phiv(C'') = O(\psi)\] and \[|\Vsp \setminus C''| \geq  |\Vsp \setminus C'| - \lambda |C'| 
\geq 0.5|\Vsp| - 0.1|\Vsp| = 0.4|\Vsp| =  \Omega(\beta) \cdot |\Vsp|.\]
The step of calculating the size of a cut costs $O(D)$ rounds by \cref{lem-basic}. The other parts can be done in zero rounds.
\end{proof}

\begin{lemma}[Well-connected subgraphs in expander splits]\label{lem-expander-split-high-conductance}
Let $G=(V,E)$ be any graph, and let $\Gsp = (\Vsp, \Esp)$ be its expander split graph simulated in the communication network $G$.
Given any  $W \subseteq \Vsp$ of $\Gsp$ with $|W| = \beta |\Vsp|$, where $0 < \beta \leq 1$, and
$\Phiv(\Gsp[W]) = \psi$, in zero rounds we can obtain a subset $E^\ast$ of $G$ with  $|E^\ast| \geq \beta |E|$
 and $\Phie(G[E^\ast]) = \Omega(\psi)$.
\end{lemma}
\begin{proof}
The subset $E^\ast$ is selected by the set of all edges in $E$ associated with at least one vertex in $W$.  Specifically, $e  \in E^\ast$ if its corresponding edge $\xi(e)$ is incident to a vertex in $W$. This immediately implies that \[|E^\ast| \geq \frac{|W|}{2} \geq \frac{\beta |\Vsp|}{2} = \beta |E|,\] 
since $|\Vsp| = 2|E|$ by \cref{lem-expander-split-property}.

To see that $\Phie(G[E^\ast]) = \Omega(\psi)$, we consider the subset $W^\ast$ with $W \subseteq W^\ast \subseteq \Vsp$ defined by including all vertices in $\Vsp$ associated with the edges in $E^\ast$. 
Specifically, $v \in W^\ast$ if $v$ is incident to an edge $\xi(e)$ for some $e \in E^\ast$. Note that  $|W^\ast| = 2 |E^\ast|$.

We claim that $\Phiv(\Gsp[W^\ast]) = \Omega(\psi)$. To see this, observe that $\Gsp[W^\ast]$ can be constructed from $\Gsp[W]$ by applying the following  steps: 
\begin{enumerate}
    \item Let $U \subseteq W$ be a defined as follows. For each $u' \in W$,  let  $e \in E$ be the unique edge  associated with $u'$, i.e., $u'$ is incident to $\xi(e)$, and write  $\xi(e) = \{u', v'\}$. Then $u'$ is added to $U$ if $v' \notin W$. 
    \item For each vertex $u' \in U$ in the graph  $\Gsp[W]$, we append a leaf $v'$ to $u'$ by adding the edge  $\xi(e) = \{u', v'\}$, where $v'$ and $\xi(e)$ are the ones defined above.
    \item After the above step, the vertex set of current graph is identical to $W^\ast$, and the current graph is a subgraph of $\Gsp[W^\ast]$. We add extra edges to make it isomorphic to  $\Gsp[W^\ast]$.
\end{enumerate}
%In view of \cref{lem-expansion}, 
As $\Gsp$ has constant maximum degree, the step of appending a leaf vertex to each vertex in $U$ clearly affects the sparsity by at most a constant factor.
% (for the case $\Delta = O(1)$). 
The step of adding extra edges cannot decrease the sparsity. Therefore, indeed $\Phiv(\Gsp[W^\ast]) = \Omega(\psi)$.
Finally,  by \cref{lem-conductance-aux-1}, we have \[\Phie(G[E^\ast]) \geq \Phiv(\Gsp[W^\ast])  = \Omega(\psi),\] as required.
\end{proof}

\subsection{Graph Operations}\label{sect-conductance-op}

The following lemma allows us to bound the sparsity of the graph $G^\ast$ resulting from contracting some vertices of $G$. Note that we do not keep self-loops and multi-edges during contraction, and recall that  the notation $E(V_i,V_j)$ denotes the set of edges $\{ e=\{u,v\} \in E \ | \ u \in V_i, v \in V_j \}$.

\begin{lemma}[Contraction]\label{lem-contraction}
Consider a graph $G=(V,E)$   and a partition $\VV = \{V_1, V_2, \ldots, V_x\}$ of $V$. Define $G^\ast$ as the result of contracting  $V_i$ into a vertex $v_i$, for each $1 \leq i \leq k$. We have $\Phiv(G^\ast) \geq c^{-1} \Phiv(G)$, where $c = \max_{1 \leq i < j \leq k} |E(V_i,V_j)|$. 
\end{lemma}
\begin{proof}
Let $C$ be any cut of $G^\ast=(V^\ast, E^\ast)$. Let $C' = \bigcup_{v_i \in C} V_i$ be its corresponding cut in $G$. To prove the lemma, it suffices to show that $|\partial(C)| / |C| \geq c^{-1} |\partial(C')| / |C'|$, If this is true for any  $C$ be any cut of $G^\ast$, by selecting $C$ as a sparsest cut, we have 
\begin{align*}
     \Phiv(G^\ast)  = \Phiv(C) &= \max\left\{  
    \frac{|\partial(C)|}{|C|}, \frac{|\partial(C)|}{|V^\ast \setminus C|}
    \right\}\\
    &\geq \max\left\{  
    \frac{c^{-1} |\partial(C')|}{|C'|}, \frac{c^{-1}|\partial(C')|}{|V \setminus C'|}
    \right\}\\
    &= c^{-1} \Phiv(C')\\
    &\geq  c^{-1} \Phiv(G).
\end{align*}
Now we prove that $|\partial(C)| / |C| \geq c^{-1} |\partial(C')| / |C'|$.
It is straightforward to see that $|C'| \geq |C|$, so we just need to show that $|\partial(C)|  \geq c^{-1} |\partial(C')|$. Indeed, $|\partial(C')| = \sum_{\{v_i, v_j\} \in \partial(C)} |E(V_i,V_j)| \leq c \cdot |\partial(C)|$ by definition of $c$.
\end{proof}

%The next lemma consider the opposite direction of the previous lemma. In \cref{lem-expansion}, the graph $G$ can be viewed as the result of expanding each $v_i$ in $G^\ast$ to a graph $G[V_i]$. 

%\begin{lemma}[Expansion]\label{lem-expansion}
%Consider a graph $G=(V,E)$   and a partition $\VV = \{V_1, V_2, \ldots, V_x\}$ of $V$. Define $G^\ast$ as the result of contracting  $V_i$ into a vertex $v_i$, for each $1 \leq i \leq x$. We have $\Phiv(G) = \Omega( d^{-1} \cdot \Delta^{-1} \cdot \vPhiIn{G}{\VV} \cdot  \Phiv(G^\ast))$, where $d = \max_{1 \leq i  \leq x} |V_i|$. 
%\end{lemma}
%\begin{proof}
%Apply \cref{lem-conductance-v} on $G$ and $\VV$.
%It is straightforward to see that $\Phiv(G^\ast) \leq d \cdot \PhiOut{G}{\VV}$. Therefore,  $\Phiv(G) = \Omega(\Delta^{-1} \cdot \vPhiIn{G}{\VV} \cdot \vPhiOut{G}{\VV}) =  \Omega( d^{-1} \cdot \Delta^{-1} \cdot \vPhiIn{G}{\VV} \cdot  \Phiv(G^\ast))$.
%\end{proof}

The next lemma considers the operation of subdividing edges into paths of length at most $d$. 

%Note that \cref{lem-expansion} already gives a bound of  $\Phiv(G^\ast) = \Omega( \Delta^{-1} d^{-2} \Phiv(G))$. The following lemma gives a tighter bound.

\begin{lemma}[Subdivision]\label{lem-subdivide}
Consider a graph $G=(V,E)$. Define $G^\ast$ as the result of replacing some of its edges  by  paths of length at most $d$. We have  $\Phiv(G^\ast) = \Omega( \Delta^{-1} d^{-1} \Phiv(G))$.
\end{lemma}
\begin{proof}
Let $C$ be a sparsest cut of  $G^\ast=(V^\ast, E^\ast)$, i.e., $\Phiv(C) = \Phiv(G^\ast)$. We assume  $\Phiv(C) \leq 1/(2d)$, since otherwise we are done already. For each edge $e \in E$, denote $P_e$ as the path of length at most $d$ in $G^\ast$ corresponding to $e$.
Let $C'$ be the cut resulting from applying the following procedure to $C$. 
\begin{enumerate}
    \item Initially $\tilde{C} = C$.
    \item For each edge $e \in E$, do the following.  
If the two endpoints  of $P_e$ belong to the same side of the cut $(\tilde{C}, V^\ast \setminus \tilde{C})$, then we move the entire path $P_e$ to that side.  If the two endpoints of $P_e$ belong to different sides of the cut $(\tilde{C}, V^\ast \setminus \tilde{C})$, then we move the intermediate vertices in $P_e$ appropriately so that $P_e$ crosses the cut $(\tilde{C}, V^\ast \setminus \tilde{C})$ exactly once. 
\item $C' = \tilde{C}$ is the final result.
\end{enumerate}
Clearly we have $|\partial(C')| \leq |\partial(C)|$, and the number of vertices moved from one side of the cut to the other side during the procedure is at most $(d-1)|\partial(C)|$. Since $\Phiv(C) \leq 1/(2d)$,  we have $|C'| \geq  |C| - (d-1)|\partial(C)|
\geq |C|(1 -(d-1)\Phiv(C))
> |C| /2$. Therefore, \[\frac{|\partial(C)|}{|C|} > \frac{1}{2} \cdot \frac{|\partial(C')|}{|C'|}.\]
The cut $C'$ has the property that for each $e \in E$, the path $P_e$ crosses the cut at most once.  Therefore, if we take $C'' = C' \cap V$ as the cut in $G$ corresponding to $C'$, then we have $|\partial(C'')| = |\partial(C')|$ and $|C'| \leq d \sum_{v \in C''} \deg(v) \leq d \Delta |C''|$. Therefore, $|\partial(C')|/|C'| \geq \Delta^{-1} d^{-1}  |\partial(C'')|/|C''|$, and so
\[\frac{|\partial(C)|}{|C|} > \frac{1}{2} \cdot \frac{|\partial(C')|}{|C'|}
\geq  \frac{1}{2 \Delta d} \cdot \frac{|\partial(C'')|}{|C''|}.\]
Similarly, applying the same analysis to the other side of the cut, we  have
\[\frac{|\partial(C)|}{|V^\ast \setminus C|} >  \frac{1}{2} \cdot\frac{|\partial(C')|}{|V^\ast \setminus C'|} \geq \frac{1}{2 \Delta d} \cdot  \frac{|\partial(C'')|}{|V \setminus C''|},\]
and so
\begin{align*}
\Phiv(G^\ast) = \Phiv(C) 
&= \max \left\{ \frac{|\partial(C)|}{|C|}, \frac{|\partial(C)|}{| V^\ast \setminus C|}\right\}\\
%&>
%\frac{1}{4} \cdot \left(  \frac{|\partial(C')|}{|C'|} +  \frac{|\partial(C')|}{| V^\ast \setminus C'|}\right)\\
&>
\frac{1}{2 \Delta d} \cdot \max  \left\{  \frac{|\partial(C'')|}{|C''|},  \frac{|\partial(C'')|}{| V \setminus C''|}\right\}\\
&> \frac{1}{2 \Delta d} \cdot  \Phiv(G),
\end{align*}
as required.
\end{proof}

\subsection{Well-connected Subgraphs from Expander Embeddings}\label{sect-expander-emb}

In this section, we show how to extract a well-connected subgraph from a given expander embedding with small congestion and dilation.

\begin{lemma}[Expander embeddings]\label{lem-expander-emb}
Let $G=(V,E)$ be a graph. Suppose a graph $H$ with $\Phiv(H) \geq \psi$ and maximum degree $\Delta_H$ can be embedded into $U \subseteq V$ with congestion $c$ and dilation $d$. Let $W  \subseteq V$ be the set of all vertices involved in the embedding. Then we have $$\Phiv(G[W]) = \Omega\left(c^{-1} d^{-1} \Delta_H^{-1}  \psi\right).$$
\end{lemma}
\begin{proof}
Starting from $G_0 = H$, we construct $G[W]$ as follows.
\begin{enumerate}
    \item For each edge $e$ in $G_0 = H$, suppose the path that embeds $e$ in the embedding has length $x \leq d$, then replace $e$  by a path of length $x$.  Denote the resulting graph by $G_1$.
    \item Some sets of vertices in $G_1$ correspond to the same vertex in $G$ in the embedding of $H$ to $U$. We contract these sets of vertices. Denote the resulting graph by $G_2$. Note that each edge in $G_2$ is the result of merging at most $c$ edges of $G_1$.
    \item Now the set of vertices of $G_2$ is identical to $W$, and $G_2$ is a subgraph of $G[W]$. We add edges to $G_2$ to make it isomorphic to $G[W]$. Denote the resulting graph by $G_3 = G[W]$.
\end{enumerate}
By \cref{lem-subdivide}, $\Phiv(G_1) = \Omega(\Delta_H^{-1} d^{-1} \Phiv(G_0)) = \Omega(\Delta_H^{-1} d^{-1} \psi)$. By \cref{lem-contraction}, $\Phiv(G_2) \geq c^{-1} \Phiv(G_1) = \Omega(\Delta_H^{-1} c^{-1} d^{-1} \psi)$. Since adding edges does not decrease sparsity, we have $\Phiv(G[W])= \Phiv(G_3) \geq \Phiv(G_2) = \Omega(\Delta_H^{-1} c^{-1} d^{-1} \psi)$, as required.
\end{proof}

The next lemma considers the case we have a simultaneous embedding of multiple expanders.

\begin{lemma}[Simultaneous expander embeddings]\label{lem-expander-emb-multi}
Let $G=(V,E)$ be a graph. 
 Let $U_1, U_2, \ldots, U_k$ be disjoint subsets of $V$.
 Suppose we can embed  $H_1, H_2, \ldots, H_k$ simultaneously to  $U_1, U_2, \ldots, U_k$ with congestion $c$ and dilation $d$. Consider the following parameters.
 \begin{itemize}
     \item $\hat{\Delta}$ is an upper bound of the maximum degree of $H_1, H_2, \ldots, H_k$ and $G$.
     \item $\phi_i$ is a lower bound of $\min_{1 \leq i \leq k}\Phiv(H_i)$.
     \item $\phi_o$ is an lower bound of $\vPhiOut{G[U]}{\UU}$, where $U = U_1 \cup U_2 \cup \cdots \cup U_k$ and $\UU = \{U_1, U_2, \ldots, U_k\}$.
 \end{itemize}
 Let $W  \subseteq V$ be the set of all vertices involved in the embedding. Then we have 
 $$\Phiv(G[W]) = \Omega\left(\hat{\Delta}^{-2} c^{-1} d^{-1} \phi_i \phi_o\right).$$
\end{lemma}
\begin{proof}
The proof is similar to that of \cref{lem-expander-emb}. Consider the  graph $G[U]$,  but each $G[U_i]$ is replaced by $H_i$. We denote this graph as $G_0$. We construct $G[W]$ from $G_0$ as follows.
\begin{enumerate}
    \item For each $1 \leq i \leq k$, for each edge $e$ in $H_i \subseteq G_0$, if in the embedding $e$ is a path of length $x \leq d$, then we replace $e$ by a path of length $x$. Denote the resulting graph by $G_1$.
    \item Some sets of vertices in $G_1$ correspond to the same vertex in $G$ in the embedding of $H_1, H_2, \ldots, H_k$. We contract these sets of vertices. Denote the resulting graph by $G_2$. Note that each edge in $G_2$ is the result of merging at most $c$ edges of $G_2$.
    \item Now the set of vertices of $G_2$ is identical to $W$, and $G_2$ is a subgraph of $G[W]$. We add edges to $G_2$ to make it isomorphic to $G[W]$. Denote the resulting graph by $G_3 = G[W]$.
\end{enumerate}
Applying
\cref{lem-conductance-v} to $G_0$ and $\UU = \{U_1, U_2, \ldots, U_k\}$, we have $\Phiv(G_0) = \Omega( \tilde{\Delta}^{-1} \cdot \vPhiIn{G_0}{\UU} \cdot \vPhiOut{G_0}{\UU})$, where  $\tilde{\Delta} \leq \hat{\Delta}$, $\vPhiIn{G_0}{\UU} \geq \phi_i$, and $\vPhiOut{G_0}{\UU} \geq \phi_o$. Therefore, $\Phiv(G_0) = \Omega(\hat{\Delta}^{-1} \phi_i \phi_o)$.
By \cref{lem-subdivide}, $\Phiv(G_1) = \Omega(\hat{\Delta}^{-1} d^{-1} \Phiv(G_0)) = \Omega(\hat{\Delta}^{-2} d^{-1} \phi_i \phi_o)$. By \cref{lem-contraction}, $\Phiv(G_2) \geq c^{-1} \Phiv(G_1) = \Omega(\hat{\Delta}^{-2} c^{-1} d^{-1} \phi_i \phi_o)$. Since adding edges does not decrease $\Phiv$, we have $\Phiv(G[W])= \Phiv(G_3) \geq \Phiv(G_2) = \Omega(\hat{\Delta}^{-2} c^{-1} d^{-1} \phi_i \phi_o)$, as required.
\end{proof}
\section{Maximal Flow}\label{sect-flow}

In this section, we provide distributed algorithms for variants of maximal flow problems and show how  sparse cuts can be obtained   if we cannot find a desired  solution for the given flow problems.  
\subsection{Sparse Cuts from Well-separated Sets}
%\subsection{Cuts from Flows}
\label{sect-cut-from-flow}

We review a technique from~\cite{chuzhoy2019deterministic} that enables us to obtain sparse cuts from certain maximal flows.
Consider the following motivating example. We have a set of source vertices $S$ and a set of sink vertices $T$ in a graph $G = (V,E)$. Suppose we invoke an algorithm that returns a maximal set of edge-disjoint $S$-$T$ paths $\PP=\{P_1, P_2, \ldots, P_x\}$ subject to the constraint that the length of each path is at most $d$.  Let $E^\ast$ be the set of edges involved in theses paths. Now it is clear that for each $u \in S$ and for each $v \in T$, we must have $\dist_{E \setminus E^\ast}(u, v) > d$, since otherwise $\PP$ is not maximal.
If $|\PP|$ is very small, then we expect that there is a sparse cut in the graph, and such a cut can be found using the following lemma.

\begin{lemma}[Cuts from flows~\cite{chuzhoy2019deterministic}]\label{lem-cutfromflow-basic}
Consider a graph $G = (V,E)$ with maximum degree $\Delta$.
Suppose we are given $S \subseteq V$, and $T \subseteq V$ satisfying $\dist(u, v) > d$ for each $u \in S$ and $v \in T$. %Define $S' = \{v \in S \ | \ \deg_{E \setminus E^\ast}(v) > 0\}$ and $T' = \{v \in T \ | \ \deg_{E \setminus E^\ast}(v) > 0\}$. 
Then there is a cut $C$ with $|C| \leq |V|/2$ and $\Phiv(C) = O(\Delta d^{-1} \log |V|)$, and this cut separates $S$ and $T$. Such a cut $C$ can be found in $O(D+d)$ rounds.
\end{lemma}
\begin{proof}
For each $i$, we define the sets $S_i = \{ u \in V \ | \ \dist(u,S) \leq i\}$  and $T_i = \{ u \in V \ | \ \dist(u,T) \leq i\}$.
Note that $S = S_0 \subseteq S_1 \subseteq \cdots \subseteq S_{d/2}$ and 
$T= T_0 \subseteq T_1 \subseteq \cdots \subseteq T_{d/2}$.
Since $\dist(S,T) > d$, we have $S_{d/2} \cap T_{d/2} = \emptyset$, and so either $|S_{d/2}| \leq |V|/2$ or $|T_{d/2}| \leq |V|/2$, or both.
Without loss of generality, assume $|S_{d/2}| \leq |V|/2$. 
 %Define $X_i = \vol_{E \setminus E^\ast}(S_i)$. 
 Pick $i^\ast = \arg \min_{0 \leq i < d/2} |S_{i+1}| / |S_i|$.
 We write $|S_{i^\ast+1}| / |S_{i^\ast}| = 1+ \epsilon$. Then we must have $|S| (1+\epsilon)^{d/2} \leq |V|/ 2$, and so $\epsilon = O(d^{-1} \log (|V| / (2 |S'|)) = O(d^{-1} \log |V|)$.
 We pick $C = S_i$. It is clear that $C$ separates $S$ and $T$. We bound the sparsity $\Phiv(C)$ as  $\Phiv(C) = |\partial(C)|/|C| = |E(S_i, S_{i+1} \setminus S_i)| / |S_i| \leq \epsilon \Delta = O(\Delta d^{-1} \log |V|)$.
 To calculate $C$ in the distributed setting, it suffices that each vertex $v \in V$ learns its membership in $S_0, S_1,  \ldots, S_{d/2}, T_0, T_1, \ldots, T_{d/2}$ and learns the size of these sets. This can be done in $O(D + d)$ rounds using \cref{lem-basic}.
\end{proof}

\cref{lem-cutfromflow-basic} can be extended to handle the case where there are multiple source and sink pairs $(S_1, T_1), (S_2, T_2), \ldots, (S_k, T_k)$.   
\cref{lem-cutfromflow-multiple} is proved by going through
$(S_1, T_1), (S_2, T_2), \ldots, (S_k, T_k)$ one-by-one with the same argument in the proof of \cref{lem-cutfromflow-basic}.

\begin{lemma}[Cuts from multi-commodity flows~\cite{chuzhoy2019deterministic}]\label{lem-cutfromflow-multiple}
Consider a graph $G = (V,E)$ with maximum degree $\Delta$.
Let $(S_1, T_1), (S_2, T_2), \ldots, (S_k, T_k)$ be pairs of vertex subsets such that $S_1$, $T_1$, $S_2$, $T_2$, $\ldots$, $S_k$, $T_k$ are disjoint, and $\dist(S_i, T_i) > d$ for each  $1 \leq i \leq k$.
There is an $O(k(D+d))$-round algorithm that finds cut $C$ with $|C| \leq |V|/2$
and $\Phiv(C) = O(\Delta d^{-1} \log |V|)$, and it satisfies either one of the following.
\begin{itemize}
    \item $|C \cap (S_i \cup T_i)| \geq \min \{|S_i|, |T_i|\}$, for each $1 \leq i \leq k$.
    \item $|C| \geq |V|/3$.
\end{itemize}
\end{lemma}

%\paragraph{Algorithm.} 
For the rest of \cref{sect-cut-from-flow}, we prove \cref{lem-cutfromflow-multiple}.  The algorithm for \cref{lem-cutfromflow-multiple} is as follows. Initially $C_0 = \emptyset$. For $i = 1, 2, \ldots, k$, do the following.
\begin{enumerate}
    \item Consider $\tilde{G}= (\tilde{V}, \tilde{E}) = G[\tilde{V}]$ induced by $\tilde{V} = V \setminus C_{i-1}$. 
    
    \item Apply \cref{lem-cutfromflow-basic} on $\tilde{G}$ with $\tilde{S} = S_k \cap \tilde{V}$, $\tilde{T} = T_k \cap \tilde{V}$.

\item Let $\tilde{C}_i$ be the resulting cut. Set $C_i = C_{i-1} \cup \tilde{C}_i$.
\end{enumerate}
The final cut $C$ is chosen as follows.
\begin{enumerate}
    \item If $|C_k| < |V|/3$, then set $C = C_k$.
    \item Otherwise, let $i^\ast$ be the smallest index $i$ such that $|C_i| \geq |V|/3$. If $|C_{i^\ast}| \leq |V|/2$, then set $C = C_{i^\ast}$; otherwise set $C = V \setminus C_{i^\ast}$.
\end{enumerate}

%\paragraph{Analysis.} 

For the round complexity of the algorithm, the  part of calculating $C_1, C_2, \ldots, C_k$ costs $O(k(D+d))$ rounds by \cref{lem-cutfromflow-basic}, and the part of calculating  $C$ can  be done in $O(k + D)$ rounds by using \cref{lem-basic} to calculate the size of these sets $C_1, C_2, \ldots, C_k$.

For the correctness of the algorithm, we break the analysis into several lemmas. This analysis is very similar to the one in \cite{chuzhoy2019deterministic}.
It is clear from the description of the algorithm that $|C| \leq |V|/2$. \cref{lem-flow-aux-1,lem-flow-aux-2} show that  either $|C| \geq |V|/3$  or $|C \cap (S_i \cup T_i)| \geq \min \{|S_i|, |T_i|\}$ for each $1 \leq i \leq k$. \cref{lem-flow-aux-3} shows that $\Phiv(C) = O(\Delta d^{-1} \log |V|)$.

\begin{lemma}\label{lem-flow-aux-1}
If $|C| < |V|/3$, then $C = C_k$.
\end{lemma}
\begin{proof}
For each iteration $i$, we have $|C_{i}| - |C_{i-1}|  = |\tilde{C}| \leq  |\tilde{V}|/2 = |V \setminus C_{i}|/2$ by \cref{lem-cutfromflow-basic}. In particular, if $|C_{i^\ast - 1}| < |V|/3$, then  $|C_{i^\ast}| < 2|V|/3$, and so $|V \setminus C_{i^\ast}| \geq |V|/3$. Therefore, if the algorithm  selects the final cut as  $C = C_{i^\ast}$ or  $C = V \setminus C_{i^\ast}$, then we have $|C| \geq |V|/3$. Therefore, if  $|C| < |V|/3$, then the final cut must be  selected as $C = C_k$.
\end{proof}

\begin{lemma}\label{lem-flow-aux-2}
If  $C = C_k$, then  $|C \cap (S_i \cup T_i)| \geq \min \{|S_i|, |T_i|\}$, for each $1 \leq i \leq k$.
\end{lemma}
\begin{proof}
Consider iteration $i$. Let $\tilde{z} = \min\{\tilde{S}, \tilde{T}\}$. It is clear that we already have $|C_{i-1} \cap (S_i \cup T_i)| \geq \min \{|S_i|, |T_i|\} - z$. Therefore, to prove the lemma, it suffices to show that $|\tilde{C}_i \cap (\tilde{S} \cup \tilde{T})| \geq \tilde{z}$. This is true because $\tilde{C}_i$ separates $\tilde{S}$  and $\tilde{T}$ by \cref{lem-cutfromflow-basic}.
\end{proof}

\begin{lemma}\label{lem-flow-aux-3}
$\Phiv(C) = O(\Delta d^{-1} \log |V|)$.
\end{lemma}
\begin{proof}
First of all, we show that for each $1 \leq i \leq k$, $|\partial(C_i)| / |C_i| \leq  O(\Delta d^{-1} \log |V|)$. Using \cref{lem-cutfromflow-basic}, we  bound $|\partial(C_i)|$ as follows.
\begin{align*}
|\partial(C_i)| 
&\leq \sum_{1 \leq j \leq i} |\partial_{G[V \setminus C_{i-1}]}(\tilde{C}_i)|\\
&\leq \sum_{1 \leq j \leq i} 
O(\Delta d^{-1} \log |V|)  \cdot |\tilde{C}_i|\\
&
= O(\Delta d^{-1} \log |V|)  \cdot |C_i|.
\end{align*}
If the final cut $C$ equals $C_i$ for some $1 \leq i \leq k$ and $|C| \leq |V|/2$, then we already have $\Phiv(C) = |\partial(C_i)| / |C_i| = O(\Delta d^{-1} \log |V|)$.
The only remaining case is  $C = V \setminus C_{i^\ast}$. By \cref{lem-flow-aux-1}, we know that in this case $|C| = |V \setminus C_{i^\ast}| \geq |V|/3 \geq |C_{i^\ast}|/3$, and so $\Phiv(C) = |\partial(C_{i^\ast})| / |V \setminus C_{i^\ast}| \leq 3|\partial(C_{i^\ast})| / |C_{i^\ast}| = O(\Delta d^{-1} \log |V|)$.
\end{proof}

\subsection{Randomized Maximal Flow}\label{sect-rand-flow}

 Given the two sets $S$ and $T$, we show how to find an  maximal set of vertex-disjoint $S$-$T$ paths  subject to the length constraint $d$ in $\poly(d, \log n)$ rounds with high probability. The algorithm is the augmenting path finding algorithm of Lotker, Patt-Shamir, and Pettie~\cite{LotkerPP08}. For the sake of completeness, we provide a complete proof here.  

\begin{lemma}[Randomized maximal flow~\cite{LotkerPP08}]\label{lem-maximal-flow-rand}
Let $G=(V,E)$ be a graph of maximum degree $\Delta$.
Given two vertex subset $S$ and $T$, there is an algorithm that finds a maximal set of vertex-disjoint $S$-$T$ paths of length at most $d$ in $O((d^2 \log n)(d \log \Delta + \log n))$ rounds with high probability.
\end{lemma}
\begin{proof}
The algorithm is based on the framework of \emph{blocking flow}. Specifically, the algorithm has $d$ stages. At the beginning of the $i$th stage, the current graph satisfies that $\dist(S,T) \geq i$, and the goal of this stage is to find a maximal set of vertex-disjoint $S$-$T$ paths of length exactly $i$. We then update the current graph by removing all vertices involved in these paths. It is clear that after the $d$th iteration, the union of all paths found is a maximal set of vertex-disjoint $S$-$T$ paths of length at most $d$.

In what follows, we focus on the $d'$th stage, and we will show that the task of the $d'$th stage can be solved in $O((d' \log n)(d' \log \Delta + \log n))$ rounds with high probability, and so the entire algorithm takes $O((d^2 \log n)(d \log \Delta + \log n))$ rounds.

Consider the graph $H=(V',E')$ defined as follows. The vertex set $V'$ is the set of all length-$d'$ $S$-$T$ paths.
Two length-$d'$ $S$-$T$ paths $P_1$ and $P_2$  are adjacent in $E'$ if they share some vertex $v \in V$.
Now the task of finding a maximal set of vertex-disjoint $S$-$T$ paths of length exactly $d'$ in $G$ is reduced to finding a maximal independent set (MIS) of $H$. It is infeasible to carry out a direct simulation of $H$ on $G$, but the following sampling task can be carried out efficiently.

\medskip

\noindent{\sc Sample($p$):} 
 Let $0< p < 1$ be any parameter. Let $U$ be the result of sampling each $P \in V'$ with probability $p$. Return the independent set $I$ containing all $P \in U$ such that none of the neighbors of $P$ is in $U$. 

\medskip
The task {\sc Sample($p$)} can be solved in $O(d')$ rounds, as follows. For $0 \leq i \leq d'$, define $L_i = \{v \in V \ | \ \dist(v, S) = i\}$. Note that $L_0 = S$. Since the current graph does not have any $S$-$T$ path of length less than $d'$, we have $T \cap (L_0 \cup L_1 \cup \cdots \cup L_{d'-1}) = \emptyset$. The high level idea is to (1) count the number of paths in a ``forward'' manner from level $0$ to $d'$, and (2) to realize the path sampling using the number we counted in a ``backward'' manner from level $d'$ back to $0$.

Each vertex $v \in L_0 \cup L_1 \cup \cdots \cup L_{d'}$ calculates a number $n_v$ as follows.
For each $v \in L_0$, define $n_v = 1$.
For each $v \in L_i$ with $1 \leq i \leq d'$, define $n_v = \sum_{u \in L_{i-1} \cap N(v)} n_u$. It is clear that $n_v$ equals the number of $S$-$v$ shortest paths.

We do the sampling as follows by having each edge $e \in E$ (resp., each vertex $v \in V$) compute a number $s_e$ (resp., $s_v$) indicating the number of sampled paths that pass $e$ (resp., $v$). From these numbers, it is straightforward to recover the independent set $I$ required in the task  {\sc Sample($p$)} in $O(d')$ rounds. %\thatchaphol{Could you explain a bit? Is it true that we just look for the path where all vertices $v$ in the path has $s_v = 1$?}
Specifically, if $s_v = 1$, then $v$ knows that it belongs to a path $P \in I$. Given that $v \in P \in I$, a neighbor $u \in N(v)$ belongs to the same path $P$ if $s_{\{u,v\}} = 1$, and so $v$ also knows its neighboring vertices in $P$.

These numbers $\{s_e\}$ and $\{s_v\}$ are calculated as follows.
For $k = d', d'-1, \ldots, 1, 0$, do the following.
\begin{enumerate}
    \item If $k = d'$, then each $v \in L_{d'}$ samples $s_v \leftarrow \Bernoulli(n_v, p)$. Otherwise, each $v \in L_k$ sets $s_v \leftarrow \sum_{u \in N(v) \cap L_{k+1}} s_{\{u,v\}}$.
    \item If $k > 0$, then each $v \in L_{k}$ sets the numbers $\{s_e \ | \ e=\{u,v\}, u \in N(v) \cap L_{k-1} \}$ by simulating the following procedure. Write $N(v) \cap L_{k-1} = \{u_1, u_2, \ldots, u_x\}$. We prepare $s_v$ balls and $x$ bins. We throw each ball into a bin randomly in such a way that the probability that a ball lands in the $i$th bin is $n_{u_i}/n_v$. After this process, calculate the number of balls in each bin, and set $s_{\{u_i, v\}}$ to be the number of balls in the $i$th bin, for each $1 \leq i \leq x$.
\end{enumerate}

Now given that {\sc Sample($p$)}  can be solved efficiently,  an MIS of $H$ can be computed as follows. 

\medskip
\noindent{\sc MIS Computation: }
Let $\Delta' \leq n\cdot  (\Delta+1)^{d'-1}$ be the maximum degree of $H$. For $k = \Delta'$, $\Delta'/2$, $\Delta'/4$, $\ldots$, $2$, $1$, repeat the following procedure for $O(\log n)$ iterations. Compute an independent set $I \leftarrow \textsc{Sample}(p)$ with $p = 1/(2k)$, and then remove all vertices that are in $I$ or adjacent to $I$ from $H$. The final output is the union of all independent sets found during the procedure.

\medskip

As {\sc Sample($p$)} can be solved in $O(d')$ rounds and for each $k$, we repeat for $O(\log n)$ iterations,
the round complexity of the MIS algorithm is $O(d' \log n \log \Delta') = O((d' \log n)(d' \log \Delta + \log n))$.
We now analyze this MIS algorithm. To show the correctness of the algorithm, it suffices to show that at the end $H$ becomes empty. We prove this by showing that the following invariant holds. 

\medskip
\noindent{\sc Invariant($k$): }
At the beginning of an iteration with parameter $k$, the current remaining graph has maximum degree at most $k$.
\medskip

This invariant trivially holds for the first iteration $k = \Delta'$. For the inductive step, now suppose that we are at the beginning of an iteration with parameter $k$, and the current graph already has maximum degree at most $k$, and we need to show that at the end of this iteration, the maximum degree is reduced to at most $k/2$. 

Recall that an iteration consists of $O(\log n)$ stages of finding an independent set $I \leftarrow \textsc{Sample}(p)$ with $p = 1/(2k)$ and removing all vertices that are in $I$ or adjacent to $I$ from $H$.
Consider any vertex $v$ in $H$ with degree at least $k$ at the beginning of one stage of computing  $I \leftarrow \textsc{Sample}(p)$.

Consider the process of sampling each vertex with probability $p = 1/(2k)$.
Consider the following two events:
\begin{itemize}
    \item  Define $\mathcal{E}_1$ as the event that exactly one vertex $u$ in $N^+(v)$ is sampled. We have $\Prob[\mathcal{E}_1] = (\deg(v)+1) p (1-p)^{\deg(v)}$.
    Recall that $k/2 \leq \deg(v) \leq k$. The local minimum of this function $\Prob[\mathcal{E}_1]$ on the domain $k/2 \leq \deg(v) \leq k$ has local minima at the two ends. For the case $\deg(v) = k$, we have $\Prob[\mathcal{E}_1] \geq \frac{k+1}{2k} (1 - 1/(2k))^{k} > 1/2$. For the case $\deg(v) = k/2$, we have $\Prob[\mathcal{E}_1] \geq \frac{(k/2) +1}{2k} (1 - 1/(2k))^{k/2} = \Omega(1)$.
    \item Now condition on the event  $\mathcal{E}_1$. Define $\mathcal{E}_2$ as the event that the unique vertex $u$ in $N^+(v)$ that is sampled satisfies that all vertices in $N(u)$ are not sampled. Note that if  $\mathcal{E}_2$ occurs, then we must have $u \in I$, and so $v$ will be removed from $H$. Let $s = |N(u) \setminus N^+(v)| \leq k-1$. Hence $\Prob[\mathcal{E}_2]  = (1-p)^{s} \geq 1-sp > 1/2$.
\end{itemize}
Thus, with a constant positive probability, $v$ is removed after this stage. Hence we conclude that if $v$ has $\deg(v) \geq k/2$ at the beginning, then after $C \log n$ stages, with probability $1 - n^{-\Omega(1)}$,  either $\deg(v) < k/2$ or $v$ is removed.
\end{proof}

The purpose of \cref{lem-maximal-flow-rand} is to prove the following lemma. 
%To facilitate efficient distributed implementation, we  use the technique of obtaining sparse cuts from maximal flows in~\cite{chuzhoy2019deterministic}, combined with an efficient distributed maximal flow algorithm of Lotker, Patt-Shamir, and Pettie~\cite{LotkerPP08}.

\begin{lemma}[Randomized cut or match] \label{lem-matching-player-rand} Consider a bounded-degree graph $G=(V,E)$ an a parameter $0 < \psi < 1$.
Given a set of source vertices $S$ and a set of sink vertices $T$ with $|S| \leq |T|$, there is an algorithm that finds a cut $C$ and a set of $S$-$T$ paths $\PP$ embedding a matching $M$ between $S$ and $T$ satisfying the following requirements in $O(\psi^{-2} \log^4 n (D + \psi^{-3} \log^4 n) )$ rounds with high probability.
\begin{description}
    \item[Match:] The embedding $\PP$ has congestion  $c = O(\psi^{-2} \log^4 n)$ and dilation   $d = O(\psi^{-1} \log n)$.
    \item[Cut:] Let $S' \subseteq S$ and $T' \subseteq T$ be the subsets that are not matched by $M$. If $S' \neq \emptyset$, then $C$ satisfies $S' \subseteq C$,  $T' \subseteq V \setminus C$, and $\Phiv(C) \leq \psi$; otherwise $C = \emptyset$.
\end{description}
\end{lemma}
\begin{proof}
Initially $S_1 = S$, $T_1 = T$,   $M = \emptyset$, and $\PP = \emptyset$.
For $i = 1, 2, \ldots, O(d^{2} \log^2 n)$, do the following. Apply the algorithm of \cref{lem-maximal-flow-rand} to find a maximal set of vertex-disjoint paths $\PP_i = \{P_1=(s_1, \ldots, t_1), P_2 = (s_2, \ldots, t_2), \ldots, P_x=(s_x, \ldots, t_x)\}$ between $S_{i}$ and $T_{i}$ subject to the constraint that the length of the paths is at most $d$, where $P_j=(s_j, \ldots, t_j)$ is interpreted as an embedding of $\{s_j, t_j\}$ with dilation $d$. We emphasize that we invoke \cref{lem-maximal-flow-rand} to find each maximal set $\PP_i$ of paths on the same graph $G$, and only the sets $S_i$ and $T_i$ are changed for each invocation.
There are two cases.

\paragraph{Case 1.} The first case is when $|\PP_i| \geq |S_i| / ( d^2 \log n )$. In this case, add $\PP_i$ to the current embedding $\PP$, and add $M_i = \{s_1, t_1\}, \{s_2, t_2\}, \ldots, \{s_x, t_x\}$ to the current matching $M$. Set $S_{i+1}$ and $T_{i+1}$ by removing the vertices in $S_{i}$ and $T_{i}$ that are matched in $M_i$. Proceed to the next iteration $i+1$. 

\paragraph{Case 2.}  The other case is when $|\PP_i| < |S_i| /( d^2 \log n )$. In this case, apply the algorithm of \cref{lem-cutfromflow-basic} on $G' = G[V \setminus W]$  and $(S_i, T_i)$, where $W$ is the set of all vertices used in $\PP_i$. In the application of \cref{lem-cutfromflow-basic}, we still assume that the underlying vertex set is $V$ by making $W$ an independent set.
The algorithm returns a cut $C'$ with $\Phiv_{G'}(C') = O(d^{-1} \log n)$, and it separates $S_i$ and $T_i$.
We pick $C$ to be one of $C'$ and $V \setminus C'$ in such a way that
$S_i \subseteq C$, and $T_i \subseteq V \setminus C$. We bound  $\Phiv_{G}(C)$ as follows.
\begin{align*}
  \Phiv_{G}(C) &= |\partial_{G}(C)| / \min\{|C|, |V \setminus C|\} \\
   &= |\partial_{G}(C)| / |C'| \\
  &\leq (|\partial_{G'}(C)| + \Delta |W|) / |C'|\\
  &\leq \Phiv_{G'}(C') + \Delta d |\PP_i|/|S_i|\\
  &= O(d^{-1} \log n).
\end{align*}
In the calculation $\Delta = O(1)$ is the maximum degree of $G$. We also use the following two facts $|C'| \geq \min\{S_i, T_i\} \geq |S_i|$ and $|\PP_i| < |S_i| /( d^2 \log n )$.

\paragraph{Termination.} The algorithm terminates whenever it enters Case~2 or has $S_i = \emptyset$. 
The number of iterations is at most $O(d^2 \log^2 n)$, since $|S_{i+1}| \leq |S_i|(1 - 1/( d^2 \log n))$ whenever it enters Case~1 in the $i$th iteration. Therefore, the final embedding has congestion $c = O(d^2 \log^2 n)$. It is clear that the embedding of the matching $M$ and the cut $C$ returned by the algorithm satisfy all the requirements by a change of variable $\psi = O(d^{-1} \log n)$.

\paragraph{Round complexity.} There are $c = O(\psi^{-2} \log^4 n)$ iterations. In each iteration, we apply the algorithm of \cref{lem-maximal-flow-rand} to compute $\PP$ and calculate its size using \cref{lem-basic}. The cost of \cref{lem-maximal-flow-rand} is  \[O((d^2 \log n)(d \log \Delta + \log n)) = O(d^3 \log n + d^2 \log^2 n) = O(\psi^{-3} \log^4 n).\]
Note that $\Delta = O(1)$ and $d = O(\psi^{-1} \log n)$. %\thatchaphol{do you mean $O(d^3 \log^2 n)$?} 
The cost of  \cref{lem-basic} is $O(D)$.
Lastly, the cost $O(d + D) = O(\psi^{-1} \log n + D)$ of the algorithm  \cref{lem-cutfromflow-basic} is not a dominating term.
To summarize, the overall round complexity is
$O(\psi^{-2} \log^4 n (D + \psi^{-3} \log^4 n) )$
\end{proof}

\subsection{Deterministic Nearly Maximal Flow}\label{sect-det-flow}

In this section, we provide the analogues of \cref{lem-maximal-flow-rand,lem-matching-player-rand} in the deterministic setting. We are not aware of an efficient deterministic algorithm that solves the maximal vertex-disjoint path problems. However, using a technique of of Goldberg, Plotkin, and Vaidya~\cite{GPV93}, we are able to \emph{nearly} solve the problem.

\begin{lemma}[Deterministic nearly maximal flow]\label{lem-GPV-basic}
Consider a graph $G=(V, E)$ of maximum degree $\Delta$.
Let $S \subseteq V$ and $T \subseteq V$ be two subsets.
There is an $O(d^3 \beta^{-1} \log^2 \Delta \log n)$-round deterministic algorithm that finds a set $\PP$ of $S$-$T$ vertex-disjoint paths of length at most $d$, together with a vertex set $B$ of size at most $\beta |V \setminus T| < \beta|V|$, such that any $S$-$T$ path of length at most $d$ that is vertex-disjoint to all paths in $\PP$ must contain a vertex in $B$.
\end{lemma}
\begin{proof}
The proof of the lemma uses the approach of Goldberg, Plotkin, and Vaidya~\cite{GPV93} in the framework of blocking flow.
Similar to the proof of \cref{lem-maximal-flow-rand}, 
 the algorithm has $d$ stages. At the beginning of the $i$th stage, the current graph satisfies that $\dist(S,T) \geq i$, and the goal of this stage is to find a maximal set of vertex-disjoint $S$-$T$ paths $\PP_i$ of length exactly $i$. We then update the current graph by removing all vertices involved in these paths, together with a set $B_i$ of \emph{leftover} vertices of size $|B_i| \leq (\beta/d) |V|$. We will show that after removing these vertices, we have  $\dist(S,T) \geq i+1$.  It is clear that after the $d$th iteration, the union of all paths found $\PP = \PP_1 \cup \PP_2 \cup \cdots \cup \PP_d$, together with the union of all leftover vertices $B = B_1 \cup B_2 \cup \cdots \cup B_d$ satisfy all the requirements.
 
 \paragraph{Algorithm for one stage.}
 In what follows, we focus on the $d'$th stage. For $0 \leq i \leq d'$, define $L_i = \{v \in V \ | \ \dist(v, S) = i\}$. Note that $L_0 = S$ and $(L_{0} \cup L_1 \cup \cdots \cup  L_{d' - 1}) \cap T = \emptyset$. For each edge $e = \{ u, v\}$ between two adjacent layers $u \in L_i$ and $v \in L_{i+1}$ ($0 \leq i \leq d'-1$), we direct the edge $e$ from $u$ to $v$. Let $G'$ be the directed graph induced by these directed edges. The algorithm for the $d'$th stage only considers $G'$. Indeed, any $S$-$T$ path of length exactly $d'$ in the current graph must be a directed path from $L_0$ to $L_{d'}$ in $G'$.

\paragraph{States of vertices.} 
During the algorithm, there are four possible states of a vertex: \idle, \aactive, \dead, and \success. Initially, all vertices are \idle. The algorithm proceeds in iterations. 
In each iteration, we maintain a set of \emph{active paths} $\Pactive$. At the beginning of the algorithm, $\Pactive = L_0$ in the sense that each $v \in L_0$ is a length-0 path in $\Pactive$.

\paragraph{Algorithm for one iteration.}
The algorithm for one stage proceeds in iterations.
The task of each iteration is to extend the paths in $\Pactive$ along directed edges simultaneously by solving a bipartite maximal matching on the following bipartite graph:  one part $X$ of the bipartite graph is the set of the last vertices in all paths in $\Pactive$,  the other part $Y$ is the set of all \idle\ vertices, and a vertex $u \in X$ and a vertex $v \in Y$ are adjacent if $\{u,v\}$ is an directed edge from $u$ to $v$ in $G'$.

The vertices that are currently in a path in $\Pactive$ is \aactive. During the above path extension procedure, if a path $P = (v_1, v_2, \ldots, v_t)$ is unable to extend because all neighboring \idle\ vertices of $v_t$ are matched to the last vertex of other paths in $\Pactive$, then $v_t$ is removed from the path $P$, and $v_t$ changes is state to \dead, and $v_{t-1}$ becomes the last vertex of $P$.
If a path $P = (v_1 \in L_0 = S, v_2, \ldots, v_{t-1}, v_t \in L_{d'} \subseteq T)$ reaches a sink, then this path is removed from $\Pactive$ and is added to $\PP_i$, and all vertices in $P$ change their state to  \success.

\paragraph{Termination.}
After $x$ iterations, there can be at most $2|V \setminus T|/x$ remaining active paths in $\Pactive$. The reason is that there must be at least $|\Pactive|$ vertices in $V \setminus T$ changing their state in each iteration, and each vertex can change its state at most twice: $\idle \rightarrow \aactive \rightarrow \dead$ or $\idle \rightarrow \aactive \rightarrow \success$. 
Also, note that the number of active paths never increases. So if  $|\Pactive|>2|V \setminus T|/x$  after $x$ iterations, there must have been more than $2|V \setminus T|$ state changes of vertices, which is a contradiction. 
We terminate the algorithm after $2 d^2 /\beta$ iterations. This ensures that there are at most $(\beta / d^2) |V \setminus T|$ paths in $\Pactive$ in the end, and they contain at most $(\beta / d) |V \setminus T|$ vertices. We set $B_i$ to be the set containing these vertices, i.e., $B_i$ is the set of all  \aactive\ vertices.

\paragraph{Correctness.} To show the correctness of the algorithm for the $d'$th stage, we need to show that at the end of the algorithm for the $d'$th stage, any $S$-$T$ path of length exactly $d'$ must contain a vertex in $B_i$ or a vertex in a $S$-$T$ path in the set $\PP_i$ found by the algorithm. If there is an $S$-$T$ path  of length exactly $d'$ not using any vertex in $B_i$ or in a path in $\PP$, then such a path $P$ must be a directed path in $G'$ starting from a \dead\ vertex of $L_0$, ending at an \idle\ vertex of $L_{d'}$, and   $P$ does not contain any \aactive\ or \success\ vertex. The existence of $P$ implies the existence of a directed edge $e = (u,v)$ where $u$ is \dead\ but $v$ is \idle. This is impossible, since by the time $u$ changes its state to \dead, all its out-neighbors are not \idle, and we know that a vertex cannot change its state back to \idle. Therefore, we conclude that such a path $P$ does not exist, and so the algorithm for the $d'$th stage is correct.

\paragraph{Round complexity.} There are $d$ stages. Each stage consists of $2 d^2 /\beta$ iterations. Each iteration consists of a computation of a  maximal matching, which can be found in $O(\log^2 \Delta \log n)$ rounds deterministically using~\cite{Fischer17}. Therefore, the overall round complexity is $O(d^3 \beta^{-1} \log^2 \Delta \log n)$.
\end{proof}

The following lemma is similar to  \cref{lem-matching-player-rand}, but here we only need to find a cut $C$ when the number of unmatched sources is at least $\beta|V \setminus T|$. Intuitively, this means that we are allowed to have a small number of \emph{leftover} vertices.

\begin{lemma}[Deterministic cut or match] \label{lem-cut-match-det} Consider a  graph $G=(V,E)$ with maximum degree $\Delta$ and a parameter $0 < \psi < \Delta$.
Given a set of source vertices $S$ and a set of sink vertices $T$ with $|S| \leq |T|$, there is a deterministic algorithm that finds a cut $C$ and a set of $S$-$T$ paths $\PP$ embedding a matching $M$ between $S$ and $T$ satisfying the following requirements in $O( \Delta^2 \psi^{-2} \log^4 n (D + \Delta^{4} \psi^{-4} {\beta}^{-1}  \log^2 \Delta \log^6 n) )$ rounds deterministically.
\begin{description}
    \item[Match:] The embedding $\PP$ has congestion  $c = O(\Delta^2 \psi^{-2} \log^4 n)$ and dilation   $d = O(\Delta \psi^{-1} \log n)$.
    \item[Cut:] Let $S' \subseteq S$ and $T' \subseteq T$ be the subsets that are not matched by $M$. If $|S'| > \beta |V \setminus T|$, then $C$ satisfies $S' \subseteq C$,  $T' \subseteq V \setminus C$, and $\Phiv(C) \leq \psi$; otherwise $C = \emptyset$.
\end{description}
\end{lemma}
\begin{proof}
The proof is very similar to that of \cref{lem-matching-player-rand}.
Initially $S_1 = S$, $T_1 = T$, $\PP = \emptyset$, and $M = \emptyset$.
For $i = 1, 2, \ldots, O( d^{2} \log^2 n)$, do the following. 
If $|S_i| \leq \beta |V \setminus T|$, terminate the algorithm; otherwise apply the algorithm of \cref{lem-GPV-basic} with $\beta' = \beta / ( d \log n)$ to find a set of vertex-disjoint paths $\PP_i = \{P_1=(s_1, \ldots, t_1), P_2 = (s_2, \ldots, t_2), \ldots, P_x=(s_x, \ldots, t_x)\}$ between $S_{i}$ and $T_{i}$ and a subset $B_i \subseteq V$ satisfying the following conditions.
\begin{itemize}
    \item Each path in $\PP_i$ has length at most $d$.
    \item $|B_i| \leq \beta'|V \setminus T| = \beta |V \setminus T| / ( d \log n) < |S_i| / ( d \log n)$. Note that the last inequality is due to $|S_i| > \beta |V \setminus T|$, since other wise the algorithm has been terminated.
    \item Any $S_i$-$T_i$ path of length at most $d$ must contain a vertex in $B_i$ or a vertex in a path in  $\PP_i$. 
\end{itemize}
Each $P_j=(s_j, \ldots, t_j)$ is interpreted as an embedding of $\{s_j, t_j\}$ with dilation $d$. There are two cases.

\paragraph{Case 1.} The first case is  $|\PP_i| \geq |S_i| / (  d^2 \log n )$. In this case, add $M_i = \{ \{s_1, t_1\}$, $\{s_2, t_2\}$, $\ldots$, $\{s_x, t_x\}\}$ to $M$, and add $\PP_i$ to the current embedding $\PP$. Set $S_{i+1}$ and $T_{i+1}$ by removing the vertices in $S_{i}$ and $T_{i}$ that are matched in $M_i$. Proceed to the next iteration $i+1$. 

\paragraph{Case 2.}  The second case is  $|\PP| < |S_i| /(  d^2 \log n )$. In this case, apply the algorithm of \cref{lem-cutfromflow-basic} on $G' = G[V \setminus W]$  and $(S_i, T_i)$, where $W$ is the set including the following vertices.
\begin{itemize}
    \item The vertices used in paths in $\PP_i$.
    \item The vertices in $B_i$.
\end{itemize}
In the application of \cref{lem-cutfromflow-basic}, We assume that the underlying vertex set of $G'$ is $V$ by treating $W$ as an independent set in $G'$.
The algorithm of \cref{lem-cutfromflow-basic} returns a cut $C'$ with $\Phiv_{G'}(C') = O(\Delta d^{-1} \log n)$, and it separates $S_i$ and $T_i$.
We pick $C$ to be one of $C'$ and $V \setminus C'$ in such a way that
$S_i \subseteq C$, and $T_i \subseteq V \setminus C$. We bound  $\Phiv_{G}(C)$ as follows.
\begin{align*}
  \Phiv_{G}(C) &= |\partial_{G}(C)| / \min\{|C|, |V \setminus C|\} \\
   &= |\partial_{G}(C)| / |C'| \\
  &\leq (|\partial_{G'}(C)| + \Delta |W|) / |C'|\\
  &\leq \Phiv_{G'}(C') + \Delta d |\PP|/|S_i| + \Delta |B_i|/|S_i|\\
  &= O(\Delta d^{-1} \log n).
\end{align*}
In the calculation we use the following  facts.
\begin{itemize}
    \item $|C'| \geq \min\{S_i, T_i\} \geq |S_i|$.
    \item $|\PP_i| < |S_i| /(  d^2 \log n )$.
    \item $|B_i| \leq |S_i| / ( d \log n)$.
\end{itemize}

\paragraph{Termination.} The algorithm terminates whenever it enters Case~2 or has $|S_i| \leq \beta|V \setminus T|$. 
The number of iterations is at most $O( d^2 \log^2 n)$, since $|S_{i+1}| \leq |S_i|(1 - 1/(  d^2 \log n))$ whenever it enters Case~1 in the $i$th iteration. Therefore, the final embedding has congestion $c = O( d^2 \log^2 n)$. It is clear that the embedding of the matching $M$ and the cut $C$ returned by the algorithm satisfy all the requirements by a change of variable $\psi = O(\Delta d^{-1} \log n)$.

\paragraph{Round complexity.} There are $c = O( d^2 \log^2 n) = O(\Delta^2 \psi^{-2} \log^4 n)$ iterations. In each iteration, we apply the algorithm of \cref{lem-GPV-basic} to compute $\PP_i$ and calculate its size using \cref{lem-basic}. The cost of \cref{lem-GPV-basic} is  $O(d^3 {\beta'}^{-1} \log^2 \Delta \log n) 
= O(d^4 {\beta}^{-1}  \log^2 \Delta \log^2 n)
= O(\Delta^{4} \psi^{-4} {\beta}^{-1}  \log^2 \Delta \log^6 n)$. The cost of  \cref{lem-basic} is $O(D)$.
Lastly, the cost $O(d + D) = O(\Delta \psi^{-1} \log n + D)$ of the algorithm  \cref{lem-cutfromflow-basic} is not a dominating term.
To summarize, the overall round complexity is
$O( \Delta^2 \psi^{-2} \log^4 n (D + \Delta^{4} \psi^{-4} {\beta}^{-1}  \log^2 \Delta \log^6 n) )$.
\end{proof}

%\cref{lem-cut-match-det-2} considers an extension of \cref{lem-cut-match-det} to handle the case of $|S| > |T|$.

%\begin{lemma}[Deterministic cut or match, $|S| > |T|$] \label{lem-cut-match-det-2} Consider a  graph $G=(V,E)$ with maximum degree $\Delta$ and a parameter $0 < \psi < \Delta$.
%Given a set of source vertices $S$ and a set of sink vertices $T$ with $\epsilon|S| \leq |T|$, there is a deterministic algorithm that finds a cut $C$ and a set of $S$-$T$ paths $\PP$ embedding a $b$-matching $M$ between $S$ and $T$ where
%\[
%b(v)=
%\begin{cases}
%1, & v \in S\\
%O(\epsilon^{-1}) &  v \in T\\
%\end{cases}
%\]
%satisfying the following requirements in $O( \Delta^2 \psi^{-2} \log^4 n (D + \Delta^{4} \psi^{-4} {\beta}^{-1}  \log^2 \Delta \log^6 n) )$ rounds deterministically.
%\begin{description}
%    \item[Match:] The embedding $\PP$ has %congestion  $c = O(\Delta^2 \psi^{-2} \log^4 n)$ and dilation   $d = O(\Delta \psi^{-1} \log n)$.
%    \item[Cut:] Let $S' \subseteq S$ and $T' \subseteq T$ be the subsets that are not matched by $M$. If $|S'| > \beta |V \setminus T| > \beta|V|$, then $C$ satisfies $S' \subseteq C$,  $T' \subseteq V \setminus C$, and $\Phiv(C) \leq \psi$; otherwise $C = \emptyset$.
%\end{description}
%\end{lemma}

%The following lemma shows that in certain circumstances we can solve the flow problem deterministically without leftover vertices.

\begin{lemma}[Deterministic flow without leftover] \label{lem-cut-match-det-noleftover} Consider a  graph $G=(V,E)$ with $\Phiv(G) \geq \psi$. 
Given a set of vertices $S=\{s_1, s_2, \ldots, s_x\} \subseteq V$   with $|S| < |V|/2$, there is an $O( \Delta^2 \psi^{-2} \log^{9/2} n) \cdot (D + 2^{O(\sqrt{\log n})}\cdot \Delta^{4} \psi^{-4})$-round deterministic algorithm that finds a set of $S$-$T$ paths $\PP=\{P_1, P_2, \ldots, P_x\}$ meeting the following conditions, with $T = V \setminus S$.
\begin{itemize}
    \item For each $1 \leq i \leq x$, $P_i$ is a path of length $O(\Delta \psi^{-1} \log^{3/2} n)$ starting at $s_i \in S$ and ending at a vertex in $T = V \setminus S$.
    \item Each vertex $v \in V$ belongs to at most $O(\Delta^2 \psi^{-2}) \cdot 2^{O(\sqrt{\log n})}$ paths in $\PP$.
    \item The set of paths  $\PP$ is stored implicitly in the following sense. For each vertex $v_j$ in the path $P_i = (s_i= v_1, v_2, \ldots, v_{k-1}, v_k \in T)$, given $\ID(s_i)$, the vertex $v_j$ can locally calculate $\ID(v_{j-1})$ (if $j > 1$) and $\ID(v_{j+1})$ (if $j < k$).
\end{itemize}
\end{lemma}
\begin{proof}
Initially $S_1 = S$. For $i = 1, 2, \ldots, \sqrt{\log n}$, do the following. Apply the algorithm of \cref{lem-cut-match-det} with $\psi' =  0.9 \psi$, $\beta' = 2^{-\sqrt{\log n}}$, $S' = S_i$, and $T' = T_i = V \setminus S_i$. Since $\Phiv(G) \geq \psi > \psi'$, the algorithm of \cref{lem-cut-match-det} is guaranteed to return an embedding $\PP_i$ of a matching $M_i$ between $S_i$ and $T_i$ with congestion $c' = O(\Delta^2 \psi^{-2} \log^4 n)$ and dilation $d' = O(\Delta \psi^{-1} \log n)$, and the set of unmatched vertices $S_i'$ in $S_i$ has size $|S_i'| \leq \beta' |V \setminus T| \leq |S_i| \cdot 2^{-\sqrt{\log n}}$.
If $S_i' \neq \emptyset$, proceed to the next iteration $i+1$ with $S_{i+1} = S_i'$.

\paragraph{Post-processing.}
It is clear that after  $i = \sqrt{\log n}$ iterations, we have $S_{i}' =  S_{i+1} = \emptyset$. That is, every vertex in $S$ is matched during the above algorithm. However, many of these vertices are not matched to vertices in $T$. We do the following post-processing step to fix it.

Initially, let $\PP_1' = \PP_1$  be the matching embedding of the first iteration of the above algorithm. Suppose by inductive hypothesis that $\PP_{i}'$ is a set of vertex-disjoint paths between $S \setminus S_{i+1}$ and $T$ such that each $v \in S \setminus S_{i+1}$ is the starting vertex of exactly one path in $\PP_{i}'$. Then we construct $\PP_{i+1}'$ from  $\PP_{i}'$ as follows. 

For each vertex $v \in S_{i+1} \setminus S_i$, there is a path $P \in \PP_{i+1}$ starting at $v$ and ending at a vertex $u \in T_{i+1} = V \setminus S_{i+1}$. There are  two cases.
\begin{itemize}
    \item If $u \in T$, then we add $P$ to $\PP_{i+1}'$.
    \item If $u \notin T$, then $u \in S \setminus S_{i+1}$, and there is a path $P' \in \PP_{i}'$ starting at $u$ and ending at a vertex $w \in T$. We add the concatenation of $P$ and $P'$ to $\PP_{i+1}'$.
\end{itemize}
We set $\PP = \PP_{\sqrt{\log n}}$. It is clear that each vertex $v \in S$ is the starting vertex of exactly one path in $\PP$, and each path in $\PP$ ends at a vertex in $T$. The congestion $c_i$ and dilation $d_i$ of $\PP_i'$ can be calculated recursively as follows.
\begin{align*}
    c_1 &= c',\\
    c_{i} &= 2 c_{i-1} + c',\\
    d_1 & = d',\\
    d_i & = d_{i-1} + d'.
\end{align*}
Therefore, the paths in $\PP$ have length at most $d_{\sqrt{\log n}} = \sqrt{\log n} \cdot d' = O(\Delta \psi^{-1} \log^{3/2} n)$, and each vertex $v \in V$ belongs to at most $c_{\sqrt{\log n}} = 2^{O(\sqrt{\log n})} \cdot c' = O(\Delta^2 \psi^{-2}) \cdot 2^{O(\sqrt{\log n})}$.

\paragraph{Round complexity.}
The algorithm consists of $\sqrt{\log n}$ iterations of the algorithm of \cref{lem-cut-match-det}, and each of them costs \[O( \Delta^2 \psi^{-2} \log^4 n (D + \Delta^{4} \psi^{-4} {\beta'}^{-1}  \log^2 \Delta \log^6 n) )
= O( \Delta^2 \psi^{-2} \log^4 n) \cdot (D + 2^{O(\sqrt{\log n})}\cdot \Delta^{4} \psi^{-4}  ) 
\] rounds. Therefore, the overall round complexity is
\[O( \Delta^2 \psi^{-2} \log^{9/2} n) \cdot (D + 2^{O(\sqrt{\log n})}\cdot \Delta^{4} \psi^{-4}  ).\]
The post-processing step can be seen as the instruction for routing that allows us to store the set of paths $\PP$ implicitly, and so this step does not incur any overhead in the round complexity. 
\end{proof}

\subsection{Multi-commodity Deterministic Nearly Maximal Flow}\label{sect-det-flow-multi}

%In \cref{sect-det-flow-multi}, 
We extend \cref{lem-cut-match-det} to the case where there are multiple pairs of source vertices $S_i$ and sink vertices $T_i$.

\begin{lemma}[Simultaneous deterministic cut or match] \label{lem-cut-match-det-multi} Consider a  graph $G=(V,E)$ with maximum degree $\Delta$ and a parameter $0 < \psi < \Delta$. We are given the following as input.
\begin{description}
    \item[Sources and sinks:] $S_1, T_1, S_2, T_2, \ldots, S_k, T_k$ are $2k$ disjoint subsets of $V$ such that $|S_i| \leq |T_i|$ for each $1 \leq i \leq k$, where $\min_{1 \leq i \leq k} |S_i| \geq \beta|V|$.
    \item[Cut:] $\Cin \subseteq V$ is a cut with $0 \leq |C| < |V|/3$ and $\Phiv(C) \leq \psi/2$.
\end{description}
Then there is a deterministic algorithm with round complexity
\[
O\left(Dk \Delta^{2} \psi^{-2} \log n \log \beta^{-1} + k \Delta^6 \psi^{-6} \beta^{-1} \log^2\Delta \log^2 n \log \beta^{-1}\right)
\]
that finds a cut $\Cout$ and a set of $S_i$-$T_i$ paths $\PP_i$ embedding a matching $M_i$ between $S_i$ and $T_i$, for each $1 \leq i \leq k$, satisfying the following requirements.
\begin{description}
    \item[Match:] The simultaneous embedding $\PP_1, \PP_2, \ldots, \PP_k$ has congestion  $c = O(\Delta^{2} \psi^{-2} \log n \log \beta^{-1})$ and dilation   $d = O(\Delta \psi^{-1} \log n)$.
    \item[Cut:] For the cut $\Cout$, there are two options. 
    \begin{itemize}
        \item The first option is to have $|V|/3 \leq |\Cout| \leq |V|/2$ and $\Phiv(\Cout) \leq \psi$.
        \item The second option is to have  $0 \leq |\Cout| \leq |V|/2$,  $\Phiv(\Cout)\leq \psi/2$, and $\Cin \subseteq \Cout$. Furthermore, for each $1 \leq i \leq k$, if less than half of the vertices of $S_i$ are matched in $M_i$, then $\Cout$ must contain at least $|S_i|/2$ vertices in $S_i \cup T_i$. 
    \end{itemize}
\end{description}
\end{lemma}
\begin{proof}
The algorithm is an iterated applications of the algorithm of \cref{lem-cut-match-det-multi-oneit} with $\Cin$ and the current remaining pairs $(S_1', T_1')$, $(S_2', T_2')$, $\ldots$, $(S_{k'}', T_{k'}')$ with parameters $\psi$ and $\beta' = \sum_{1 \leq i \leq k'} |S_i'| / |V| \geq (\beta/2) k'$, as  $|S_i'| \geq |S_i|/2 \geq  \beta|V|$ for each remaining pair $(S_i', T_i')$.
Note that it is possible that $k' < k$, and we re-order the pairs to have $S_i' \subseteq S_i$ and $T_i' \subseteq T_i$ for each $1 \leq i \leq k'$. For $1 \leq i \leq k'$, all vertices in $S_i \setminus S_i'$ have been matched to $T_i \setminus T_i'$ in previous iterations.  For $k' < i \leq k$, the pair $(S_i, T_i)$ has been removed because at least half of $S_i$ have been matched to $T_i$ in previous iterations. The congestion $c$ of the embedding is at most the number of iterations, as the embedding returned by the algorithm of \cref{lem-cut-match-det-multi-oneit} has congestion 1.

If the output of the algorithm of \cref{lem-cut-match-det-multi-oneit} is a cut $\Cout$, then we return $\Cout$ and the current matching and its embedding. This cut satisfies the requirements stated in the lemma. By  \cref{lem-cut-match-det-multi-oneit}, there are two possibilities. One possibility is $|V|/3 \leq |\Cout| \leq |V|/2$ and $\Phiv(\Cout) \leq \psi$. The other possibility is $0 \leq |\Cout| \leq |V|/2$,  $\Phiv(\Cout)\leq \psi/2$,  $\Cin \subseteq \Cout$, and also for each $1 \leq i \leq k'$,   $\Cout$ must contain at least $|S_i'| \geq|S_i|/2$ vertices in $S_i' \cup T_i' \subseteq S_i \cup T_i$. In both cases, $\Cout$ satisfies all the requirements. Remember that for each $k' < i \leq k$, at least half of the vertices in $S_i$ are matched already.

If the output of the algorithm of \cref{lem-cut-match-det-multi-oneit} is a set of vertex-disjoint paths $\PP = \PP_1 \cup \PP_2 \cup \cdots \cup \PP_k$, then we update the current matching and embedding accordingly by including $\PP$. The vertices  in $S_1', T_1', S_2', T_2', \ldots, S_k', T_k'$ that have been matched are removed from these sets. After the update, if $|S_i'| < |S_i|/2$ for some $1 \leq i \leq k'$, then we remove the pair $(S_i', T_i')$, as we have already found a large enough matching between them.

By \cref{lem-cut-match-det-multi-oneit}, the total number of vertices in the current remaining $S_1' \cup S_2' \cup \cdots \cup S_k'$ is guaranteed to be reduced by a factor of $1 - \Omega(\Delta^{-2} \psi^{2} \log^{-1} n)$ in each iteration, if the algorithm does not return $\Cout$. Since we cannot have $|S_1' \cup S_2' \cup \cdots \cup S_k'| < (\beta/2)|V|$, the number of iterations can be bounded by $c = O(\log \beta^{-1}) \cdot O(\Delta^{2} \psi^{-2} \log n)$, as required.

For the round complexity, each invocation of the algorithm of \cref{lem-cut-match-det-multi-oneit} costs \[O\left({k'}D + {k'}^2 \Delta^4 \psi^{-4} {\beta'}^{-1} \log^2\Delta \log n\right)
= O\left(kD + k \Delta^4 \psi^{-4} \beta^{-1} \log^2\Delta \log n\right)
\] because $k' \leq k$ and $\beta' = \Omega(\beta k')$. Therefore, the overall round complexity is 
\[
O\left(D k \Delta^{2} \psi^{-2} \log n \log \beta^{-1} + k \Delta^6 \psi^{-6} \beta^{-1} \log^2\Delta \log^2 n \log \beta^{-1}\right).
\qedhere
\]
\end{proof}

\begin{lemma}[Simultaneous deterministic cut or match, one iteration] \label{lem-cut-match-det-multi-oneit} Consider a  graph $G=(V,E)$ with maximum degree $\Delta$ and a parameter $0 < \psi < \Delta$. We are given the following as input.
\begin{description}
    \item[Sources and sinks:] $S_1, T_1, S_2, T_2, \ldots, S_k, T_k$ are $2k$ disjoint subsets of $V$ such that $|S_i| \leq |T_i|$ for each $1 \leq i \leq k$. Define  $\beta = \sum_{1 \leq i \leq k} |S_i| / |V|$.
    \item[Cut:] $\Cin \subseteq V$ is a cut with $0 \leq |C| < |V|/3$ and $\Phiv(C) \leq \psi/2$.
\end{description}
Then there is a deterministic algorithm with round complexity
\[
O\left(kD + k^2 \Delta^4 \psi^{-4} \beta^{-1} \log^2\Delta \log n\right)
\]
that outputs either one of the following.
\begin{description}
    \item[Match:] $\PP= \PP_1 \cup \PP_2 \cup \cdots \cup \PP_k$ is a set of vertex-disjoint paths of length at most $d = O(\Delta \psi^{-1} \log n)$, where the paths in $\PP_i$ are $S_i$-$T_i$ paths. It is required that $|\PP| = \Omega(\Delta^{-2} \psi^{2} \log^{-1} n) \cdot \sum_{1 \leq i \leq k} |S_i|$.
    \item[Cut:] For the cut $\Cout$, there are two options. 
    \begin{itemize}
        \item The first option is to have $|V|/3 \leq |\Cout| \leq |V|/2$ and $\Phiv(\Cout) \leq \psi$.
        \item The second option is to have  $0 \leq |\Cout| \leq |V|/2$,  $\Phiv(\Cout)\leq \psi/2$, and $\Cin \subseteq \Cout$. Furthermore, for each $1 \leq i \leq k$,   $\Cout$ must contain at least $|S_i|$ vertices in $S_i \cup T_i$. 
    \end{itemize}
\end{description}
\end{lemma}
\begin{proof}
For $i = 1, 2, \ldots, k$, do the following.
\begin{enumerate}
    \item  Let $G_i$ be the subgraph $G$ that excludes all vertices involved in $\PP_1, \PP_2, \ldots, \PP_{i-1}$ and their incident edges. 
    \item Apply the algorithm of \cref{lem-GPV-basic} to $S = S_i$ and $T = T_i$ on the graph $G_i$ with parameter $d$ and $\beta' = (1/8) \beta \psi \Delta^{-1} k^{-1}$.
    Let $\PP_i$ and $B_i$ be the output result.
\end{enumerate}

Set $\PP= \PP_1 \cup \PP_2 \cup \cdots \cup \PP_k$. If $|\PP| >
(1/8) \beta \psi \Delta^{-1} (d+1)^{-1} |V|=
\Omega(\Delta^{-2} \psi^{2} \log^{-1} n) \cdot \sum_{1 \leq i \leq k} |S_i|$, then the algorithm returns $\PP$ as the output.

From now on, we assume $|\PP| \leq (1/8) \beta \psi \Delta^{-1} (d+1)^{-1} |V|$.
Let $W$ be the set of all edges involved in $\PP$. Note that $|W| \leq |\PP| (d+1) \leq (1/8) \beta \psi \Delta^{-1} |V|$.
Set $B = W \cup B_1 \cup B_2 \cup \cdots \cup B_k$. By \cref{lem-GPV-basic}, we have $|B_i| \leq \beta' |V| = (1/8) \beta \psi \Delta^{-1} k^{-1} |V|$, for each $1 \leq i \leq k$. Therefore, the number of edges incident to $B$ is at most 
\[\Delta|B| \leq \Delta|W| + \Delta k \beta' |V| \leq (1/8) \beta \psi   |V| + (1/8) \beta \psi   |V|= (1/4) \beta \psi   |V|.\]

Apply \cref{lem-cutfromflow-multiple} to $(S_1, T_1), (S_2, T_2), \ldots, (S_k, T_k)$ and the parameter $d$ on the graph $G[V \setminus (B \cup \Cin)]$, assuming that the underlying vertex set is $V \setminus \Cin$ by treating $B \setminus \Cin$ as an independent set.
Let $C$ be the output.
By selecting $d = O(\Delta \psi^{-1} \log n)$ to be large enough, we can make $C$ to  have \[\Phiv_{G[V \setminus (B \cup \Cin)]}(C) \leq \psi/4 = O(\Delta d^{-1} \log n).\]

We argue that $\Phiv_{G[V \setminus \Cin]}(C) \leq \psi /2$. Each edge $e \in \partial_{G[V \setminus \Cin]}(C)$ is either incident to $B$ or belongs to $\partial_{G[V \setminus (B \cup \Cin)]}(C)$. The number of such edges in the first case is at most $\Delta |B| =  (1/4) \beta \psi   |V| =   (\psi/4) \sum_{1 \leq i \leq k} |S_i|
\leq (\psi/4)  |C|$. The number of such edges in the second case is also at most $\Phiv_{G[V \setminus (B \cup \Cin)]}(C) \cdot |C| \leq (\psi/4) |C|$.
Therefore, indeed  $\Phiv_{G[V \setminus \Cin]}(C) \leq \psi /2$.

Since $|C| \leq |V \setminus \Cin|/2$ (by \cref{lem-cutfromflow-multiple})  and  $\Phiv_{G[V \setminus \Cin]}(C) \leq \psi /2$, we can apply \cref{lem-cut-combine} with $C_1 = \Cin$ and $C_2 = C$  to deduce that setting $\Cout$ to be either $C_1 \cup C_2$ (if $|C_1 \cup C_2| \leq |V|/2$) or $V \setminus C_1 \cup C_2$ (if $|C_1 \cup C_2| > |V|/2$) satisfies all the requirements stated in the lemma.
 %We bound $|\partial(\Cin \cup C)|$ as follows. $$|\partial(\Cin \cup C)| \leq |\partial(\Cin)| + |\partial_{G[V \setminus \Cin]}(C)|
%\leq (\psi/2)|\Cin| + (\psi/2)|C|
%=(\psi/2)|C \cup \Cin|.$$
%If $|\Cin \cup C| \leq |V|/2$, then the algorithm outputs $\Cout = \Cin \cup C$, and the above calculation implies $\Phiv(\Cout) \leq \psi/2$, as required.
%Otherwise, $|\Cin \cup C| > |V|/2$, then the algorithm outputs $\Cout = V \setminus (\Cin \cup C)$. 
%Since $|\Cin| < |V|/3$ (assumption given in the lemma) and $|C| \leq |C \setminus \Cin|$ (\cref{lem-cutfromflow-multiple}), we must have $|C \setminus (C \cup \Cin)| \geq |V|/3$, and so $|C \setminus (C \cup \Cin)| \geq |C \cup \Cin|/2$. 
% Therefore, we have 
% \begin{align*}
%   |\Cout| &\geq |V|/3 \ \ \ \text{and}\\
%   \Phiv(\Cout) &=  \frac{|\partial(\Cin \cup C)|}{|C \setminus (C \cup \Cin)|}
% \geq \frac{2|\partial(\Cin \cup C)}{|C \cup \Cin|}
% \geq 2 \cdot (\psi/2) = \phi,
% \end{align*}
%as required.
 
 For the round complexity, the cost of \cref{lem-cutfromflow-multiple} is $O(k(D+d))$, and the cost of $k$ invocations of \cref{lem-GPV-basic} is $k \cdot O(d^3 {\beta'}^{-1} \log^2 \Delta \log n)$. As $d = O(\Delta \psi^{-1} \log n)$ and $\beta' = \Omega(\beta \psi \Delta^{-1}k^{-1})$, the overall round complexity is 
 \[
 O(kD + k^2 \Delta^4 \psi^{-4} \beta^{-1} \log^2\Delta \log n). \qedhere
 \]
\end{proof}

\section{Tools for Sparse Cut Computation}\label{sect-sparse-cut-tools}

In this section, we provide tools  for sparse cut computation.

\subsection{Diameter Reduction}
\label{sect-diam-reduce}

The goal of this section is to prove \cref{lem-diam-reduction-main}.

%, i.e., the following inequality.
%\begin{align*}
%   &\Tcut\left(n, \Delta, D, \phicut, \phiemb, \betacut, \betaleft \right)\\ 
%    &\leq O\left(D + \phicut^{-2}\Delta^2 \log^6 n\right)
%    + \Tcut\left(n, \Delta, O\left(\psi^{-1}\Delta \log^3 n\right),\phicut/2, \phiemb,  \betacut, \betaleft\right)
%\end{align*}

%\begin{lemma}[Diameter reduction]\label{lem-diam-reduce} We have the following inequalities, where $f'$ is defined by $f_n'(\psi/2) = f_n(\phi)$.
%\begin{align*}
%    \Tdet(n, D, \Delta, \phi, f, \beta) 
%    &\leq O(D + \psi^{-2} \Delta^2 \log^6 n) + 
%    \Tdet(n, O(\psi^{-1}  \Delta\log^3 n), \Delta, \phi, f', \beta),   \\ 
 %   \Trand(n, D, \Delta, \phi, f, \beta) 
  %  &\leq O(D + \psi^{-1} \Delta \log n) + 
   % \Trand(n, O(\psi^{-1} \Delta \log n), \Delta, \phi, f', \beta).
%\end{align*}
%\end{lemma}
%\begin{proof}
Let $G=(V,E)$ be the input graph of maximum degree $\Delta$, with a Steiner tree $T$ of diameter $D$.
The task given to us is to solve $\detbalspcut{\phicut}{\phiemb}{\betacut}{\betaleft}$.  We show how to solve it using one invocation of the $\left(D + \psi^{-2} \Delta^{2} \log^6 n\right)$-round deterministic algorithm $\mathcal{A}_1$ (preprocessing step) of \cref{lem-diam-reduction-aux},  one invocation of an algorithm solving $\detbalspcut{\phicut/2}{\phiemb}{\betacut}{\betaleft}$.  on a subgraph with $D' = O(\psi^{-1}  \Delta \log^3 n)$, and then finally a $O(D)$-round algorithm $\mathcal{A}_2$  (postprocessing step) that produces the final output by combining the solutions in previous steps. 

%The cost of the algorithm of \cref{lem-diam-reduction-aux} is
%$$O\left(D + \psi^{-2} \Delta^{2} \log^6 n\right).$$ 
%Our algorithm for \cref{lem-diam-reduction-main} works as follows.
\paragraph{Preprocessing step.}
To solve $\detbalspcut{\phicut}{\phiemb}{\betacut}{\betaleft}$, we first apply the algorithm $\mathcal{A}_1$ of \cref{lem-diam-reduction-aux} with  $\psi= \phicut$. 
Denote $C_1$ as its output.
There are two possibilities.
If   $C_1$ satisfies $|V|/3 \leq |C_1| \leq |V|/2$ and $\Phiv(C_1) \leq \phicut$, then we are done solving $\detbalspcut{\phicut/2}{\phiemb}{\betacut}{\betaleft}$ already, as $\betacut \leq 1/3$ always.
Otherwise, the output of the algorithm of \cref{lem-diam-reduction-aux} consists of a cut $C_1$ with  $\Phiv(C_1) \leq \phicut/2$ and a subtree $T'$ of $G$ that spans all vertices in $V \setminus C_1$, and the tree has diameter $D' = O(\psi^{-1}  \Delta \log^3 n)$.

\paragraph{Main part.}
Now, we apply an algorithm that solves $\detbalspcut{\phicut/2}{\phiemb}{\betacut}{\betaleft}$ on the subgraph $G[V \setminus C_1]$ using the Steiner tree $T'$ with diameter $D' = O(\psi^{-1}  \Delta \log^3 n)$. 
%Let $n' = |V \setminus C_1| \leq n$.  
%The round complexity of this  algorithm is at most
%\begin{align*}
%  &T\left(n, \Delta, O\left(\psi^{-1}\Delta \log^3 n\right),\phicut/2, \phiemb,  \betacut, \betaleft\right).
%\end{align*}
 Let $W_2 \subseteq V \setminus C_1$ and $C_2 \subseteq V \setminus C_1$ be its output result.

\paragraph{Postprocessing step.} 
Now we describe the algorithm $\mathcal{A}_2$ that produces the final output.
It is straightforward to see that the algorithm can be implemented in $O(D)$ rounds deterministically using the straightforward information gathering algorithm of \cref{lem-basic}.

According to the specification of $\detbalspcut{\phicut/2}{\phiemb}{\betacut}{\betaleft}$, the subgraph $G[W_2]$ has $\Phiv(G[W_2]) \geq \phiemb$,
which meets the requirements for the task $\detbalspcut{\phicut}{\phiemb}{\betacut}{\betaleft}$ on $G$, and the cut $C_2$ satisfies
$0 \leq  |C_2| \leq \frac{|V \setminus C_1|}{2}$ and $\Phiv_{G[V \setminus C_1]}(C_2) \leq \frac{\phicut}{2}$, Moreover, at least one of the following is met.
    \begin{itemize}
        \item $|C_2| \geq \betacut|V \setminus C_1|$ and $W_2 = \emptyset$.
        \item $|(V \setminus C_1) \setminus (C_2 \cup W_2)| \leq \betaleft|V \setminus C_1|$.
    \end{itemize}
%the cut $C_2$ satisfies $$\beta|(V \setminus C_1) \setminus W_2 | \leq |C_2| \leq \frac{|V \setminus C_1|}{2}$$ and $$\Phiv_{G[V \setminus C_1]}(C_2) \leq \frac{\phi}{2}.$$
%If $|C| \geq |V|/3 \geq \beta|V|$ already, then this cut $C$ alone fulfills the task $\balspcut{f}{\beta}$ on $G$. 
%In what follows, assume $|C| < |V|/3$.
Remember that we have $|C_1|  < |V|/3$, since otherwise we are done already, and so we can apply \cref{lem-cut-combine} with the above $C_1$ and $C_2$  to deduce the following. In all cases we find a valid solution $(C^\ast, W^\ast)$ for solving  $\detbalspcut{\phicut}{\phiemb}{\betacut}{\betaleft}$ on $G$.
 
 \paragraph{Case 1.} If $|C_1 \cup C_2| \leq |V|/2$, then $C^\ast = C_1 \cup C_2$ satisfies $\Phiv(C^\ast) \leq \phicut/2 < \phicut$ and $|C^\ast| \leq |V|/2$. 
We show that can output the cut $C^\ast$ with $W^\ast = W_2$.
Specifically, we need to show that either $|V \setminus (C^\ast \cup W^\ast)| \leq \betaleft|V|$ or 
$|C^\ast| \geq \betacut|V|$.

If we already have $|C_2| \geq \betacut|V \setminus C_1|$, then we have $|C^\ast| = |C_1| + |C_2| \geq \betacut|V \setminus C_1| + |C_1| \geq \betacut|V|$, as required.
Otherwise, we must have $|(V \setminus C_1) \setminus (C_2 \cup W_2)|
\leq \betaleft|V \setminus C_1|$, and so 
\begin{align*}
 |V \setminus (C^\ast \cup W^\ast)|
&=
|(V \setminus C_1) \setminus (C_2 \cup W_2)| \\
&\leq \betaleft|V \setminus C_1|\\
&\leq \betaleft|V|.   
\end{align*}

 \paragraph{Case 2.}  If $|C_1 \cup C_2| > |V|/2$, then $C^\ast = V \setminus (C_1 \cup C_2)$ satisfies $\Phiv(C^\ast) \leq \phicut$ and $|V|/3 \leq |C^\ast| \leq |V|/2$.  Therefore, we can output the cut $C^\ast$ with $W^\ast = \emptyset$, as $\betacut|V| \leq |V|/3 \leq |C^\ast|$.
%\end{proof}

\paragraph{Auxiliary lemmas.} \cref{lem-diam-reduction-aux} and \cref{lem-cut-combine} are auxiliary lemmas needed in the above proof.
%of \cref{lem-diam-reduce}.

\begin{lemma}[Sparse cut with small diameter]\label{lem-diam-reduction-aux}
Consider a graph $G=(V,E)$ with maximum degree $\Delta$ associated with a Steiner tree $T$ with diameter $D$.
Given a parameter $\psi$, there is a deterministic  algorithm with round complexity \begin{align*}
     O(D + \psi^{-2} \Delta^{2} \log^6 n) & \ \ \ \text{in the deterministic model} \\
     \text{or} \ \ O(D + \psi^{-1} \Delta \log n) & \ \ \ \text{in the randomized model,}
\end{align*}
 and it finds a cut $C \subseteq V$ with $|C| \leq |V|/2$ such that either one of the following holds.
\begin{itemize}
    \item $|V|/3 \leq |C|$ and $\Phiv(C) \leq \psi$.
    \item $\Phiv(C) \leq \psi/2$, and the algorithm also computes a subtree $T'$ that spans all vertices in $V \setminus C$, and the tree has diameter 
\begin{align*}
     D' = O(\psi^{-1} \Delta \log^3 n) & \ \ \ \text{in the deterministic model} \\
     \text{or} \ \ D' = O(\psi^{-1} \Delta \log n) & \ \ \ \text{in the randomized model.}
\end{align*}    
\end{itemize}
\end{lemma}
\begin{proof}
The algorithm is as follows. Let $K > 0$ be some constant. 
\begin{enumerate}
    \item Apply the low-diameter decomposition algorithm of \cref{thm-low-diam-decomp-rand} or \cref{thm-low-diam-decomp} with $\beta = \psi\Delta^{-1} / 3$.
    \item Each part $V_i$ of the decomposition $\VV = \{V_1, V_2, \ldots, V_x\}$ measures  its size $|V_i|$. 
    \item If $\max_{1\leq i \leq x} |V_i| \leq |V|/2$, then apply \cref{lem-bal-partition} to find a subset $\SSS \subseteq \VV$ such that $|V|/3 \leq \sum_{V_i \in \SSS} |V_i| \leq |V|/2$, and return $C = \bigcup_{V_i \in \SSS} V_i$   as the output.
    \item From now on, assume $\max_{1\leq i \leq x} |V_i| > |V|/2$. Pick any $V_j$ with $|V_j| > |V|/2$. Define $$S = \{ v \in V \ | \ \dist(v, V_j) > K \psi^{-1} \Delta \log n\}.$$ 
     If $S = \emptyset$, then all vertices are within distance $O(\psi^{-1} \Delta\log n)$ to $V_j$, and then the diameter of the graph $G$ is at most $O(\psi^{-1}\Delta \log n)$ plus the diameter $D_j$ of the Steiner tree $T_j$ associated with $V_j$. Return $C = \emptyset$ as the output. The Steiner tree $T'$ of $V \setminus C = V$ is chosen as an arbitrary BFS tree of $G$.
    \item Otherwise, $S \neq \emptyset$.
    Apply the algorithm of \cref{lem-cutfromflow-basic} with this set $S$ and $T = V_j$. As $\dist(S,T) \geq d = K \psi^{-1} \Delta \log n$, this algorithm returns a cut $C \subseteq V$ with $|C| \leq |V|/2$ and $\Phiv(C) = O(\Delta d^{-1} \log n) \leq \psi/2$ by selecting $K$ to be sufficiently large. Return this cut $C$ as the output. The choice of the Steiner tree $T'$ of $V \setminus C$ is deferred to subsequent discussion.
\end{enumerate}

We explain each step in detail. 

\paragraph{Step~1.}
In the deterministic setting, the low-diameter decomposition of \cref{thm-low-diam-decomp} costs $O(\psi^{-2} \Delta^{2} \log^6 n)$ rounds. Each cluster $V_i$ is associated with a Steiner tree $T_i$ of diameter $D_i = O(\psi^{-1} \Delta \log^3 n)$. Each edge $e \in E$ belongs to at most $c = O(\log n)$  Steiner trees. 

In the randomized setting, the 
low diameter decomposition of \cref{thm-low-diam-decomp-rand} costs $O(D + \psi^{-1} \Delta \log n)$ rounds. Each cluster $V_i$ has diameter $O(\psi^{-1}  \Delta \log n)$, and so we can pick $T_i$ to be any BFS tree of $V_i$, and $T_i$ also have diameter $D_i = O(\psi^{-1}  \Delta \log n)$.

\paragraph{Step~2.}
The computation of $|V_i|$ can be achieved using the Steiner tree $T_i$ associated with $V_i$ using \cref{lem-basic} in $O(D_i)$ rounds for each cluster.
In the randomized setting, this can be done in $O(\psi^{-1} \Delta \log n)$ rounds.
In the deterministic setting, the embedding of the Steiner trees has congestion $c = O(\log n)$
by \cref{thm-low-diam-decomp}. Therefore, the  computation of  $|V_i|$ costs $O(\log n) \cdot O(\psi^{-1} \Delta \log^3 n) = O(\psi^{-1} \Delta \log^4 n)$ rounds. We assume that there is a representative vertex $v_i$ of each part $V_i$ that stores $|V_i|$.

\paragraph{Step~3.}
Suppose we are in the case $\max_{1\leq i \leq x} |V_i| \leq |V|/2$.
The application of \cref{lem-bal-partition} takes $O(D)$ rounds. For the correctness of the output $C = \bigcup_{V_i \in \SSS} V_i$, recall that the number of inter-cluster edges in a low-diameter decomposition is at most $\beta|E|$, and so
\begin{align*}
|\partial(C)| &\leq \beta|E|  & \text{All edges in $\partial(C)$ are inter-cluster.} \\
&\leq \beta \Delta |V| \\
&\leq \beta \Delta  \cdot 3|C| & |V|/3 \leq \sum_{V_i \in \SSS} |V_i| = |C|\\
&= \psi|C|. & \beta = \psi\Delta^{-1} / 3\\
\end{align*}
To summarize, the cut $C$ satisfies both  $|V|/3 \leq |C| \leq |V|/2$ and $\Phiv(C) \leq \psi$, as required.

\paragraph{Step~4.}
Consider the case $\max_{1\leq i \leq x} |V_i| > |V|/2$. Let $V_j$ be any part with $|V_j| > |V|/2$. Suppose $S = \{ v \in V \ | \ \dist(v, V_j) > K \psi^{-1} \Delta \log n\} = \emptyset$, then all vertices are within distance $K \psi^{-1} \Delta \log n$ to $V_j$. Then the diameter of the graph $G$ is at most $K \psi^{-1} \Delta \log n$ plus the diameter $D_j$ of the Steiner tree associated with $V_j$.
Therefore, we can simply return $C = \emptyset$ and pick any BFS tree of $G$ as $T'$, as it spans all of $V$, and it has diameter $O(D_j) + O(\psi^{-1} \Delta \log n)$, which is $O(\psi^{-1} \Delta \log n)$ in the randomized setting, or $O(\psi^{-1} \Delta \log^3 n)$ in the deterministic setting, as required.

For the round complexity, constructing $S$ costs $O(\psi^{-1} \Delta \log n)$ rounds, estimating the size of $S$ costs $O(D)$ rounds (\cref{lem-basic}), finding a BFS tree of $G$ costs $O(\psi^{-1} \Delta \log^3 n)$ rounds in the deterministic setting, or $O(\psi^{-1} \Delta \log n)$ rounds in the randomized setting. Overall, the round complexity of this step is $O(D + \psi^{-1} \Delta \log^3 n)$ in the randomized setting, and $O(D + \psi^{-1} \Delta \log n)$ in the randomized setting.

\paragraph{Step~5.} The algorithm  of \cref{lem-cutfromflow-basic}  costs $O(D +d) = O(D + \psi^{-1} \Delta \log n)$ rounds. We describe how to pick the Steiner tree $T'$ that spans $V \setminus C$ with small diameter. 
Observe that $|S| < |V|/2 < |T|$, and \cref{lem-cutfromflow-basic} guarantees that
  $C$  separates $S$ and $T$. Therefore, we must have $S \subseteq C$ and $T \subseteq V \setminus C$. In particular, all vertices in $V \setminus C$ are within distance $K \psi^{-1} \Delta \log n$ to $T = V_j$. Therefore, we can simply extend the Steiner tree $T_j$ from $V_j =T$ to all of $V \setminus C$ by via shortest paths from $(V \setminus C) \setminus V_j$ to $V_j$. The resulting Steiner tree $T'$ spans $V \setminus C$, and it has diameter $D_j + O(\psi^{-1} \Delta \log n)$, which is $O(\psi^{-1} \Delta^{-1} \log n)$ in the randomized setting, or $O(\psi^{-1} \Delta^{-1} \log^3 n)$ in the deterministic setting, as required.

\paragraph{Round complexity.} In the deterministic model, the round complexities of these five steps are: $O(\psi^{-2} \Delta^{2} \log^6 n)$, $O(\psi^{-1} \Delta \log^4 n)$, $O(D)$, $O(D + \psi^{-1} \Delta \log^3 n)$ and $O(D + \psi^{-1} \Delta \log n)$. Hence the total complexity is \[O(D + \psi^{-2} \Delta^{2} \log^6 n).\]
In the randomized model, the round complexities of these five steps are: $O(D + \psi^{-1} \Delta \log n)$, $O(\psi^{-1} \Delta \log n)$, $O(D)$, $O(D + \psi^{-1} \Delta \log n)$ and $O(D + \psi^{-1} \Delta \log n)$. Hence the total complexity is \[O(D + \psi^{-1} \Delta \log n). \qedhere\]
\end{proof}

\begin{lemma}[Combine two sparse cuts]\label{lem-cut-combine}
Consider a graph $G = (V,E)$. Let $C_1 \subseteq V$ and $C_2 \subseteq V \setminus C_1$. The cut $C_1$ satisfies $\Phiv(C_1) \leq \psi/2$ and $|C_1| < |V|/3$. The cut $C_2$ satisfies  $\Phiv_{G[V \setminus C_1]}(C_2) \leq \psi/2$ and $|C_2| \leq |V \setminus C_1|/2$. Then the following holds.
\begin{itemize}
    \item If $|C_1 \cup C_2| \leq |V|/2$, then $C = C_1 \cup C_2$ satisfies $\Phiv(C) \leq \psi/2$ and $|C| \leq |V|/2$. 
    \item If $|C_1 \cup C_2| > |V|/2$, then $C = V \setminus (C_1 \cup C_2)$ satisfies $\Phiv(C) \leq \psi$ and $|V|/3 \leq |C| \leq |V|/2$.
\end{itemize}
\end{lemma}
\begin{proof}
 We bound $|\partial(C_1 \cup C_2)|$ as follows. $$|\partial(C_1 \cup C_2)| \leq |\partial(C_1)| + |\partial_{G[V \setminus C_1]}(C_2)|
\leq (\psi/2)|C_1| + (\psi/2)|C_2|
=(\psi/2)|C_2 \cup C_1|.$$

If $|C_1 \cup C_2| \leq |V|/2$, then we have $C = C_1 \cup C_2$, and the above calculation implies $\Phiv(C) \leq \psi/2$, as required.

Otherwise, $|C_1 \cup C_2| > |V|/2$, then we have $C = V \setminus (C_1 \cup C_2)$. 
Since $|C_1| < |V|/3$  and $|C_2| \leq |V \setminus C_1|/2$, we must have $|V \setminus (C_1 \cup C_2)| \geq |V|/3$, and so $|V \setminus (C_1 \cup C_2)| \geq |C_1 \cup C_2|/2$. 
 Therefore, we have 
 \begin{align*}
   |C| &\geq |V|/3 \ \ \ \text{and}\\
   \Phiv(C) &=  \frac{|\partial(C_1 \cup C_2)|}{|V \setminus (C_1 \cup C_2)|}
 \leq \frac{2|\partial(C_1 \cup C_2)|}{|C_1 \cup C_2|}
 \leq 2 \cdot (\psi/2) = \phi,
 \end{align*}
as required.
\end{proof}

\subsection{Balance Improvement}
\label{sect-Bal-improve}

The goal of this section is to prove \cref{lem-bal-improve-main}, i.e., the following inequality.
\begin{align*}
   &\Tcut\left(n, \Delta, D, \phicut, \phiemb, 1/3, \betaleft \right)\\
    &\leq O\left(D \betacut^{-1} \right) 
    + O\left(\betacut^{-1} \right) \cdot \Tcut\left(n, \Delta, D,\phicut/2, \phiemb,  \betacut, \betaleft\right)
\end{align*}

%\begin{lemma}[Balance improvement]\label{lem-Bal-improve} We have the following inequalities, where $f'$ is defined by $f_n'(\psi/2) = f_n(\phi)$.
%\begin{align*}
%    \Tdet(n, D, \Delta, \phi, f, 1/3) 
%    & \leq O(D \beta^{-1} \log \beta^{-1}) + O(\beta^{-1} \log \beta^{-1}) \cdot
 %   \Tdet(n, D, \Delta, \phi, f', \beta),   \\ 
  %  \Trand(n, D, \Delta, \phi, f, 1/3) 
  %  &\leq O(D \beta^{-1} \log \beta^{-1}) +  O(\beta^{-1} \log \beta^{-1}) \cdot
%    \Trand(n, D, \Delta, \phi, f', \beta).
%\end{align*}
%\end{lemma}

%For the rest of \cref{sect-Bal-improve}, we prove \cref{lem-Bal-improve}.
Let $G=(V,E)$ be the input graph of maximum degree $\Delta$, with a Steiner tree $T$ of diameter $D$.
The task given to us is to solve $\detbalspcut{\phicut}{\phiemb}{1/3}{\betaleft}$. We present an algorithm that solves this task. The algorithm  has  $O(\betacut^{-1})$ iterations, and each iteration costs $O(D)$ plus the round complexity for solving $\detbalspcut{\phicut/2}{\phiemb}{\betacut}{\betaleft}$ on a subgraph, and so we have the above inequality. 

The $O(D)$ part in the round complexity is due to the fact that we need to measure the size of the cuts returned by $\detbalspcut{\phicut/2}{\phiemb}{\betacut}{\betaleft}$ using \cref{lem-basic}.

\paragraph{Algorithm.}
The algorithm proceeds in iterations. 
In iteration $i$, we apply an algorithm for $\detbalspcut{\phicut/2}{\phiemb}{\betacut}{\betaleft}$ on the subgraph $G[V \setminus (C_1 \cup C_2 \cup \cdots \cup C_{i-1})]$. We write $C_i$ and $W_i$ to denote the result of $\detbalspcut{\phicut/2}{\phiemb}{\betacut}{\betaleft}$ in iteration $i$.

\paragraph{Induction hypothesis.}
We have this induction hypothesis at the beginning of iteration $i$.
\begin{itemize}
    \item $|C_1 \cup C_2 \cup \cdots \cup C_{i-1}| < |V|/3$. \item $\Phiv(C_1 \cup C_2 \cup \cdots \cup C_{i-1}) \leq \phicut/2$.
    \item $|C_j| \geq \betacut |V \setminus(C_1 \cup C_2 \cup \cdots \cup C_{j-1})|$ for each $1 \leq j \leq i-1$.
\end{itemize}

This induction hypothesis holds vacuously for the first iteration $i=1$. For $i >1$, the induction hypothesis is justified by the algorithm stated below.

Consider iteration $i$.
We write $C' = C_1 \cup C_2 \cup \cdots \cup C_{i-1}$. As $|C'| < |V|/3$, we can
Apply \cref{lem-cut-combine} with $C'$ and $C_i$  to deduce the following. 

%In all cases we find a valid solution $(C^\ast, W^\ast)$ for solving  $\detbalspcut{\phicut}{\phiemb}{\betacut}{\betaleft}$ on $G$.
 
 \paragraph{Case 1.} If $|V|/3 \leq |C' \cup C_i| \leq |V|/2$, then $C^\ast = C' \cup C_i$ satisfies $\Phiv(C^\ast) \leq \phicut/2 < \phicut$ and $|C^\ast| \leq |V|/3$. 
Therefore, we can output the cut $C^\ast$ with $W^\ast = \emptyset$.

 \paragraph{Case 2.}  If $|C' \cup C_i| > |V|/2$, then $C^\ast = V \setminus (C' \cup C_i)$ satisfies $\Phiv(C^\ast) \leq \phicut$ and $|V|/3 \leq |C^\ast| \leq |V|/2$.  Therefore, we can output the cut $C^\ast$ with $W^\ast = \emptyset$, as $|V|/3 \leq |C^\ast|$.

 \paragraph{Case 3.} If $|C' \cup C_i| < |V|/3$ and $|C_i| < \betacut|V \setminus C'|$, then the same in Case~1, $C^\ast = C' \cup C_i$ satisfies $\Phiv(C^\ast) \leq \phicut/2 < \phicut$ and $|C^\ast| < |V|/3$. We claim that $C^\ast$ with $W^\ast = W_i$ is a valid output.
 Since $|C_i| < \betacut|V \setminus C'|$, we must have $|(V \setminus C') \setminus (C_i \cup W_i)|
\leq \betaleft|V \setminus C'|$, and so 
\begin{align*}
 |V \setminus (C^\ast \cup W^\ast)|
&=
|(V \setminus C') \setminus (C_i \cup W_i)| \\
&\leq \betaleft|V \setminus C'|\\
&\leq \betaleft|V|.   
\end{align*}
 
  \paragraph{Case 4.} If $|C' \cup C_i| < |V|/3$ and $|C_i| \geq \betacut|V \setminus C'|$, then we proceed to the next iteration. Note that we must have $\Phiv(C' \cup C_i) \leq \phicut/2$ in this case.

\paragraph{Number of iterations.}
Since we proceed to the next iteration only when we are in  Case~4, we must have $|C_j| \geq \betacut |V \setminus(C_1 \cup C_2 \cup \cdots \cup C_{j-1})|$ for each $j$ if the algorithm does not terminate in iteration $j$.
Thus, the algorithm must terminate within $\tau = O(\betacut^{-1})$ iterations, since otherwise
$|C_1 \cup C_2 \cup \cdots \cup C_{\tau-1}| \geq |V|/3$ violates an induction hypothesis.

\section{Analysis of Potential Functions}\label{sect-potential}

In this section we analyze the potential functions in the randomized and the deterministic cut-matching games.

\subsection{Randomized Cut-matching Game}\label{sect-rand-potential}
The only ingredient that is missing in \cref{sect-RST} in order to prove \cref{lem-rand-cutmatch-iterations} is the following.

\begin{lemma}[Potential drop~\cite{RST14,SaranurakW19}] \label{lem-potential-drop}
For any constant  $K > 0$, the following is true  for each iteration $i \geq 1$,
\[\Expect[\Pi(i)] \leq  \Pi(i-1) \left(1 - \frac{1}{K \log n}\right) + n^{-\Omega(K)}.\]
\end{lemma}
By selecting $K$ to be sufficiently large, the $n^{-\Omega(K)}$ term becomes negligible, and so indeed the potential decreases by a factor of $1 - \Omega(1/\log n)$ in expectation in each iteration.
%If this is true for all $i \geq 1$,  then we can set $\tau = \Theta(\log^2 n)$ so that the output $H = H_{\tau}$ and $C = C_1 \cup C_2 \cup \cdots \cup C_{\tau}$ due to the second terminating condition satisfy all the requirements in \cref{lem-rand-cut-embed}, as we have discussed earlier.
The proof of \cref{lem-potential-drop} can be found in~\cite{RST14,SaranurakW19}. 
For the sake of completeness, we include  a proof of \cref{lem-potential-drop} that is almost self-contained except the following two lemmas. 

\begin{lemma}[{Gaussian behavior of projection~\cite[Lemma~3.5]{KRV09}}]\label{lem-projection}
Let $v \in \mathbb{R}^n$ be a vector with $\| v \| = \ell$, and let $R \in \mathbb{R}^n$ be a uniformly random unit vector. For $x \leq n/16$, we have
\[
\Expect\left[\langle v, r\rangle^2\right] = \ell^2 / n \ \ \ \text{and} \ \ \
\Prob\left[\langle v, r \rangle^2 \geq x \ell^2 / n \right] \leq \exp(-x/4).
\]
\end{lemma}

\begin{lemma}[{Properties of $A^l$ and $A^r$~\cite[Lemma~3.3]{RST14}}]\label{lem-rst-property}
For each iteration $i$, the sets $A^l$ and $A^r$ computed by the cut player satisfy the following properties for some separation value $\eta$.
\begin{enumerate}
    \item Either $\max_{v_j \in A^l} \leq \eta \leq \min_{v_j \in A^r}$ or $\max_{v_j \in A^r} \leq \eta \leq \min_{v_j \in A^l}$.
    \item For each $v_j \in A^l$, we have $|u[j] - \eta|^2 \geq (1/9)|u[j] - \bar{u}|^2$.
    \item $\sum_{v_j \in A^l} |u[j] - \bar{u}|^2 \geq (1/80) \sum_{v_j \in A_i} |u[j] - \bar{u}|^2$.
\end{enumerate}
\end{lemma}

Recall that $A_{i+1} = A_i \setminus C_i$, and the definition of the potential function is
\[\Pi(i) = \sum_{v_j \in A_{i+1}} \| F_i[j] - \mu_i \|^2, \ \ \text{where} \ \ \mu_i = \frac{1}{|A_{i+1}|} \cdot\sum_{v_j \in A_{i+1}} F_i[j] .\]
We first lower bound $\Pi(i-1) - \Pi(i)$ as follows.
\begin{align*}
    \Pi(i-1) - \Pi(i) 
    &= \sum_{v_j \in A_{i}} \| F_{i-1}[j] - \mu_{i-1} \|^2 -
    \sum_{v_j \in A_{i+1}} \| F_{i}[j] - \mu_i \|^2\\
    &\geq \sum_{v_j \in A_{i}} \| F_{i-1}[j] - \mu_{i-1} \|^2 -
    \sum_{v_j \in A_{i+1}} \| F_{i}[j] - \mu_{i-1} \|^2\\
    &= \sum_{v_j \in A_{i}} \left( \| F_{i-1}[j] - \mu_{i-1} \|^2 -
    \| F_{i}[j] - \mu_{i-1} \|^2 \right) +
    \sum_{v_j \in C_{i}} \| F_{i-1}[j] - \mu_{i-1} \|^2  \\
    &= \frac{1}{2} \sum_{ \{v_j, v_l\} \in M_{i}} \| F_{i-1}[j] - F_{i-1}[l] \|^2  +
    \sum_{v_j \in C_{i}} \| F_{i-1}[j] - \mu_{i-1} \|^2.  
\end{align*}
The  inequality is due to the fact that for any set of length-$n$ vectors $\{x_1, x_2, \ldots, x_k\}$, this term $\sum_{1 \leq j \leq k} \| x_j - z \|^2$ is minimized when $z = (1/k) \sum_{1 \leq j \leq k} x_j$ is chosen as the average vector.
The last equality is due to the fact that $F_{i-1}[j] \neq F_{i}[j]$ only when $v_j$ is matched in $M_i$ and the following lemma.

\begin{lemma}[{\cite[Lemma~3.3]{KRV09}}]
Let $\{v_j, v_l\} \in M_i$.
Then we have
\[
 \| F_{i-1}[j] - \mu_{i-1} \|^2 + \| F_{i-1}[l] - \mu_{i-1} \|^2 -
    \| F_{i}[j] - \mu_{i-1} \|^2 - \| F_{i}[l] - \mu_{i-1} \|^2
    = \frac{1}{2}  \| F_{i-1}[j] - F_{i-1}[l] \|^2.
\]
\end{lemma}
\begin{proof}
Write $a = F_{i-1}[j] - \mu_{i-1}$ and $b = F_{l-1}[j] - \mu_{i-1}$. Then we have $(a+b)/2 = F_{i}[j] - \mu_{i-1} = F_{i}[l] - \mu_{i-1}$, since $F_{i}[j] = F_{i}[l] = (F_{i-1}[j] = F_{i-1}[l])/2$. The lemma follows from the equality $\| a \|^2 + \|b\|^2  - 2 \|(a+b)/2\|^2 = (1/2)\|a - b\|^2$, which is true for any two vectors $a$ and $b$.
\end{proof}

Remember that we have the following in the algorithm of the cut player in iteration $i$.
\[u = F_{i-1} \cdot r \in \mathbb{R}^{n} \ \ \ \text{and} \ \ \ \bar{u} = \frac{1}{|A_i|}\sum_{v_j \in A_i} u[j],  \]
and observe that
\[ u[j] = \langle F_{i-1}[j], r \rangle
\ \ \ \text{and}  \ \ \ \bar{u} = \langle \mu_{i-1}, r \rangle. 
\]
Therefore, in view of \cref{lem-projection}, with probability $1 - n^{-c/4}$, we have
\[
\Pi(i-1) - \Pi(i) \geq
\frac{n}{2c \ln n} \sum_{ \{v_j, v_l\} \in M_{i}} ( u[j] - u[l] )^2  + \frac{n}{c \ln n} 
    \sum_{v_j \in C_{i}} ( u[j] - \bar{u} )^2. 
\]
Using \cref{lem-rst-property}, we have
$$\sum_{ \{v_j, v_l\} \in M_{i}} ( u[j] - u[l] )^2 \geq \sum_{  v_j\in A^l \setminus C_i} ( u[j] - \eta )^2
\geq \frac{1}{9} \sum_{  v_j\in A^l \setminus C_i} ( u[j] - \bar{u} )^2$$
and so
\begin{align*}
    \Pi(i-1) - \Pi(i) &\geq
\frac{n}{18 c \ln n} \sum_{  v_j\in A^l \setminus C_i} ( u[j] - \bar{u} )^2  + 
\frac{n}{c \ln n} \sum_{v_j \in  A^l \cap C_{i}} ( u[j] - \bar{u} )^2 \\
&\geq \frac{n}{18 c \ln n} \sum_{  v_j\in A^l} ( u[j] - \bar{u} )^2 \\
&\geq \frac{n}{1440 c \ln n} \sum_{  v_j\in A_i} ( u[j] - \bar{u} )^2.
\end{align*}
By \cref{lem-projection}, we have 
\[
\Expect\left[ \frac{n}{1440 c \ln n} \sum_{  v_j\in A_i} ( u[j] - \bar{u} )^2 \right]
=   \frac{1}{1440 c \ln n} \sum_{  v_j\in A_i} \| F_{i-1}[j] - \mu_{i-1} \|^2
=  \frac{\Pi(i-1)}{1440 c \ln n} .
\]
By a change of variable $K \log n = 1440 c \ln n$, we finish the proof of \cref{lem-potential-drop}. Note that the $n^{-\Omega(K)}$ term is due to the fact that there is one step that holds with probability $1 - n^{-c/4}$.

\subsection{Deterministic Cut-matching Game}\label{sect-det-potential}

In \cref{sect-det-potential}, we prove \cref{lem-potential-reduction}.
For clarity, we write $\betaKKOV = 1/4$ to denote the threshold $1/4$ in  $|C^i| \geq (1/4)|V|$, and we write $\phicut = 1/2$ to denote the threshold $1/2$ in $\Phiv_{H_{i-1}}(C^i) \leq 1/2$ in the definition of the cut-matching game in \cref{def-kkov}.
% in the following calculation. 
Note that a small difference between the analysis here and the one in the original paper~\cite{khandekar2007cut} is that we only require at least half of the vertices in $S^i$ to be matched.

Consider iteration $i$ of the cut-matching game.
We are given $H^{i-1} = M^1 \cup M^2 \cup \cdots \cup M^{i-1}$ and a bipartition $V = S^i \cup T^i$ with
\begin{align*}
   &|\partial(S^i)| \leq \phicut|S^i| \\
   &\betaKKOV|V| \leq |S^i| \leq |V|/2 
\end{align*}
%$|\partial(S^i)| \leq \phi|S^i|$ and $\betaKKOV|V| \leq |S^i| \leq |V|/2$.
The matching player in the $i$th iteration returns a matching $M^i=\{(s_1, t_1), (s_2, t_2), \ldots, (s_x, t_x)\}$, where $S_\ast^i = \{s_1, s_2, \ldots, s_x\}$ and $T_\ast^i = \{t_1, t_2, \ldots, t_x\}$ are distinct vertices in $S^i$ and $T^i$. 

We write $p(j \rightsquigarrow k)$ and $p'(j \rightsquigarrow k)$ to denote the transition probability for $(M^1, M^2, \ldots, M^{i-1})$ and $(M^1, M^2, \ldots, M^{i})$, respectively.
%Due to the conductance threshold $\phi = 1/4$ 
We have the following bounds.
\begin{align*}
 &\sum_{j \in S^i, k \in T^i} p(j \rightsquigarrow k)  \leq \frac{1}{2} \phicut|S^i|  \\
 &\sum_{j \in T^i, k \in S^i} p(j \rightsquigarrow k)  \leq \frac{1}{2} \phicut|S^i| 
\end{align*}
The reason is as follows. Consider the following process.  At the beginning each vertex $v \in S^i$ has one unit of load. 
For $i = 1, 2, \ldots, i-1$, in iteration $j$, for each edge $\{u, v\} \in M^j$, the two vertices $u$ and $v$ average their load. In the end, the summation of load in $T^i$ equals $\sum_{j \in S^i, k \in T^i} p(j \rightsquigarrow k)$. Observe that the maximum load at a vertex at any time is at most 1, and so the amount of load transmitted from $S^i$ to $T^i$ via an edge $\{u, v\} \in M^j$ crossing $S^i$ and $T^i$ during iteration $j$ is at most $1/2$. Therefore, we have $$\sum_{j \in S^i, k \in T^i} p(j \rightsquigarrow k) \leq \frac{1}{2} |\partial_{H^{i-1}}(S^i)| \leq \frac{1}{2} \phicut|S^i|
%\leq \frac{1}{2} \phi \betaKKOV n
%= \betaKKOV n / 8.
$$
The calculation of $\sum_{j \in T^i, k \in S^i} p(j \rightsquigarrow k)$ is similar.

\paragraph{Good and bad triples.}
Consider any three vertices $v_j, v_k, v_l$  with $v_j \in S^i$,  $v_k \in S^i_\ast$, $v_l \in T^i_\ast$ such that $\{v_k, v_l\} \in M^i$. Since $M^i$ is a perfect matching between $S^i_\ast$ and $T^i_\ast$, $v_l \in T^i_\ast$ is uniquely determined given $v_k \in S^i_\ast$, and vice versa.

We say that such a triple $(j,k,l)$ is \emph{good} if $p(j \rightsquigarrow k) \geq 2 \cdot p(j \rightsquigarrow l)$, otherwise $(j,k,l)$ is \emph{bad}.
%We say that an ordered pair $(j, k)$ with $v_j \in S^i$ and $v_k \in S^i_\ast$
%is \emph{good}  if  there exists an index $l$ such that  $\{v_k, v_l\}$ is an edge in $H^{i-1}$, $v_l \in T^i_\ast$ and $p(j \rightsquigarrow k) \geq 2 \cdot p(k \rightsquigarrow l)$.
%\paragraph{Bad pairs.}
%To prove this claim, we say that an ordered pair $(j, k)$  with $v_j \in S^i$ and $v_k \in S^i_\ast$ is \emph{bad} if  there exists an index $l$ such that  $\{v_k, v_l\}$ is an edge in $H^{i-1}$, $v_l \in T^i_\ast$ and $p(j \rightsquigarrow k) < 2 p(k \rightsquigarrow l)$. 
It is clear that the summation of $p(j \rightsquigarrow k)$ over all \emph{good and bad} triples $(j, k, l)$ is  \[
\sum_{v_j \in S^i, v_k \in S_\ast^i} p(j \rightsquigarrow k)
=
|S^i_\ast| - \sum_{j \in T^i, k \in S^i_\ast} p(j \rightsquigarrow k)  \geq
\left(1 - \frac{\phicut}{2}\right)|S^i|.
%(\betaKKOV n / 2) - (\betaKKOV n / 8) = (3/8) \betaKKOV n.
\]
Using the definition of bad triples, the summation of $p(j \rightsquigarrow k)$ over all \emph{bad} triples $(j, k, l)$ can be upper bounded by \[
2\cdot \sum_{v_j \in S^i, v_l \in T^i_\ast} p(j \rightsquigarrow l) \leq  
2\cdot  \sum_{j \in S^i, k \in T^i} p(j \rightsquigarrow k)  \leq \phicut |S^i|.\]
Therefore, the summation of $p(j \rightsquigarrow k)$ over all \emph{good} triples $(j, k, l)$ is at least 
\[\left(1 - \frac{3\phicut}{2}\right)|S^i| = \frac{1}{4}|S^i| \geq \frac{\betaKKOV}{4}|V| = \frac{n}{16}.\]
%$\betaKKOV n / 8$.
%In view of the above calculation, it suffices to show that the summation of   $p(j \rightsquigarrow k)$ over all bad triples $(j, k, l)$ is less than $\betaKKOV n / 4$. 
%Let $z$ be the summation of $p(j \rightsquigarrow k)$ over all bad pairs $(j, k)$. According to the definition of bad vertices, we have $z/2 < \sum_{v_j \in S^i, v_k \in T^i_\ast} p(j \rightsquigarrow k) \leq \betaKKOV n / 8$, and so $z < \beta n/ 4$, as required.

\paragraph{Potential increase for a triple.}
Fix any triple $(j, k, l)$  with $v_j \in S^i$,  $v_k \in S^i_\ast$, $v_l \in T^i_\ast$ such that $\{v_k, v_l\} \in M^i$.
Note that $p'(j \rightsquigarrow k) = p'(j \rightsquigarrow l) = (1/2) (p(j \rightsquigarrow k) + p(j \rightsquigarrow l))$. 
We want to show that the following holds for some constant $\epsilon > 0$.
\begin{align*}
 &p'(j \rightsquigarrow k) \log (1/p'(j \rightsquigarrow k)) + p'(j \rightsquigarrow l) \log (1/p'(j \rightsquigarrow l))\\
 &\geq p(j \rightsquigarrow k) \log (1/p(j \rightsquigarrow k)) + p(j \rightsquigarrow l) \log (1/p(j \rightsquigarrow l))  + \epsilon p(j \rightsquigarrow k),  &\text{ if $(j,k,l)$ is good.}\\ \\
  &p'(j \rightsquigarrow k) \log (1/p'(j \rightsquigarrow k)) + p'(j \rightsquigarrow l) \log (1/p'(j \rightsquigarrow l))\\
  &\geq p(j \rightsquigarrow k) \log (1/p(j \rightsquigarrow k)) + p(j \rightsquigarrow l) \log (1/p(j \rightsquigarrow l)),  &\forall \, v_j \in V, \, \{v_k, v_l\} \in M^i.
\end{align*}
This implies that the overall potential increase $\Pi(i) - \Pi(i-1)$ is at least $\epsilon$ times the summation of $p(j \rightsquigarrow k)$ over all \emph{good} triples $(j, k, l)$, and so 
$$\Pi(i) - \Pi(i-1) = \Omega(n).$$
For notational simplicity, we write $p_1 = p(j \rightsquigarrow k) \in [0, 1]$, $p_2 = p(j \rightsquigarrow l) \in [0, 1]$, and $p_3 = (p_1 + p_2)/2 = p'(j \rightsquigarrow k) = p'(j \rightsquigarrow l) \in [0, 1]$. Denote $h(p) = p \log (1/p)$. By the concavity of $h(p)$, we have $2 h(p_3) \geq h(p_1) + h(p_2)$. If $(j,k,l)$ is a good triple, we have $p_2 \leq p_1 / 2$, and a  calculation shows that for this case we have   $2 h(p_3) \geq h(p_1) + h(p_2) + \epsilon p_1$ with $\epsilon = (1/2) \log(32/27) > 0$.

\newpage
\bibliographystyle{alpha}
\bibliography{references}

\end{document}